%% file: main.tex
\documentclass[a4paper, USenglish,thm-restate]{lipics-v2021}
\nolinenumbers
\hideLIPIcs  %uncomment to remove references to LIPIcs series (logo, DOI, ...), e.g. when preparing a pre-final version to be uploaded to arXiv or another public repository

\usepackage{nicefrac}
\usepackage{xcolor}
\usepackage{bm}
\usepackage[shortlabels]{enumitem}
\usepackage[ruled,vlined]{algorithm2e}
\usepackage{array,multirow,graphicx,float,booktabs} %%% packages for the table

\usepackage{tikz}
\usetikzlibrary{hobby,calc}
\usetikzlibrary{arrows.meta}

\usepackage{caption}
\usepackage{subcaption}

\SetCommentSty{mycommfont}
\definecolor{gainsboro}{rgb}{0.86, 0.86, 0.86}
\xdefinecolor{cobra}{cmyk}{0,0.2,1,0.02}  %%% yellow
\xdefinecolor{node}{cmyk}{0.5,0.26,0.0,0.15} %%% blue
\xdefinecolor{curious}{cmyk}{0, 0.83, 0.80, 0.0} %%%% red
\xdefinecolor{hidden}{cmyk}{0, 0, 0, 0.35} %%% gray-ish 
\xdefinecolor{gray}{cmyk}{0.05,0.03, 0.05, 0.11} %%% node 

\AtEndPreamble{%
\theoremstyle{acmplain}
\newtheorem*{theorem*}{Theorem}}

\newcommand{\namedref}[2]{\hyperref[#2]{#1~\ref*{#2}}}
\newcommand{\sectionref}[1]{\namedref{Section}{#1}}
\newcommand{\figureref}[1]{\namedref{Figure}{#1}}

\newcommand{\equationref}[1]{\hyperref[#1]{(\ref*{#1})}}
\newcommand{\theoremref}[1]{\hyperref[#1]{Theorem~\ref*{#1}}}
\newcommand{\lemmaref}[1]{\hyperref[#1]{Lemma~\ref*{#1}}}
\newcommand{\noteref}[1]{\hyperref[#1]{note~\ref*{#1}}}
\newcommand{\appendixref}[1]{\hyperref[#1]{Appendix~\ref*{#1}}}
\newcommand{\corollaryref}[1]{\hyperref[#1]{Corollary~\ref*{#1}}}
\newcommand{\algorithmref}[1]{\hyperref[#1]{Algorithm~\ref*{#1}}}
\newcommand{\definitionref}[1]{\hyperref[#1]{Definition~\ref*{#1}}}
\newcommand{\observationref}[1]{\hyperref[#1]{Observation~\ref*{#1}}}

\renewcommand{\Pr}{\mathbb P}
\newcommand{\E}{\mathbb{E}}

\newcommand{\norm}[1]{\left|\left|#1\right|\right|}

\newcommand{\eigen}{\varphi}

\renewcommand{\vec}[1]{\vec{#1}}

\DeclareMathOperator*{\argmin}{argmin}
\DeclareMathOperator*{\argmax}{argmax}

\DeclareMathOperator*{\Geom}{Geom}

\DeclareMathOperator\sink{sink}

\newcommand{\divinfty}[2]{D_\infty\left(#1 \parallel #2\right)}

\newcommand{\divkl}[2]{D_{KL}\left(#1 \parallel #2\right)}

\newcommand{\execs}{\mathcal E}
\newcommand{\sigmaalgebra}{\Sigma}
\newcommand{\sigmaevent}{\sigma}

\newcommand{\prior}{p}

\newcommand{\degset}[2]{\deg_{#1}\left(#2\right)}

\renewcommand{\vec}{\bm} 

\newcommand{\Bin}[2]{\operatorname{Bin} \left(#1,#2\right)}

\newcommand{\fadv}{\Psi^{(F)}}
\newcommand{\seq}[1]{S^{\left(#1\right)}}
\newcommand{\seqadv}[1]{S^{(#1)}_\textsc{adv}}
\newcommand{\tadv}{t_\textsc{adv}}

\newcommand{\uniform}[1]{\mathcal{U}\left(#1\right)}
\newcommand{\uniformsubset}[2]{\mathcal{U}_{#1}\left(#2\right)}

\newcommand{\smap}{\hat{s}_{MAP}}
\newcommand{\smle}{\hat{s}_{MLE}}

\newcommand{\event}{E}

\newcommand{\myparagraph}[1]{\smallskip \textbf{#1}}

\newcommand{\honest}{V\setminus F}
\newcommand{\starttime}{{t_\star}}

\usepackage{mathtools}
\newcommand{\passage}[3][ ]{\pi_C(#2\xRightarrow{#1} #3)}
\newcommand{\passagehat}[3][ ]{\pi_A(#2\xRightarrow{#1} #3)}
\newcommand{\passagedandelion}[3][ ]{\pi_D(#2\xRightarrow{#1} #3)}
\newcommand{\bandwidth}{\mathcal B}

\newcommand{\commset}{C}
\newcommand{\activeset}{X}
\newcommand{\anaconda}{Y}

\newcommand{\anonphase}{anonPhase}
\newcommand{\nice}{\mathfrak{I}}
\newcommand{\round}{S}

\newcommand{\fakecomm}{Z}

\usepackage{dsfont}
\newcommand{\ind}{\mathds{1}}

\title{On the Inherent Anonymity of Gossiping} 

\author{Rachid	Guerraoui}{Ecole Polytechnique Fédérale de Lausanne (EPFL), Switzerland}{rachid.guerraoui@epfl.ch}{https://orcid.org/0000-0002-4794-8902}{}

\author{Anne-Marie	Kermarrec}{Ecole Polytechnique Fédérale de Lausanne (EPFL), Switzerland}{anne-marie.kermarrec@epfl.ch}{https://orcid.org/0000-0001-8187-724X}{}

\author{Anastasiia	Kucherenko}{Ecole Polytechnique Fédérale de Lausanne (EPFL), Switzerland}{anastasiia.kucherenko@epfl.ch}{}{}

\author{Rafael	Pinot}{Ecole Polytechnique Fédérale de Lausanne (EPFL), Switzerland}{rafael.pinot@epfl.ch}{https://orcid.org/0000-0001-5372-8300}{}

\author{Sasha	Voitovych}{University of Toronto, Canada\footnote{Part of the work was done when Sasha Voitovych was an intern at EPFL as part of the EPFL Excellence Research Internship Program.}}{sasha.voitovych@mail.utoronto.ca}{https://orcid.org/0000-0003-1840-476X}{}

\authorrunning{R. Guerraoui, A.-M. Kermarrec, A. Kucherenko, R. Pinot, and S. Voitovych} %TODO mandatory. First: Use abbreviated first/middle names. Second (only in severe cases): Use first author plus 'et al.'

\Copyright{Rachid	Guerraoui, Anne-Marie	Kermarrec, Anastasiia	Kucherenko, Rafael	Pinot, Sasha	Voitovych} %TODO mandatory, please use full first names. LIPIcs license is "CC-BY";  http://creativecommons.org/licenses/by/3.0/


 %TODO mandatory; please add comma-separated list of keywords

\category{} %optional, e.g. invited paper

\acknowledgements{The authors are thankful to Nirupam Gupta and Pierre-Louis Roman for fruitful discussions on the early version of the paper, and to the anonymous reviewers of DISC 2023 for their constructive comments. }%optional

\EventEditors{Rotem Oshman}
\EventNoEds{1}
\EventLongTitle{37th International Symposium on Distributed Computing (DISC 2023)}
\EventShortTitle{DISC 2023}
\EventAcronym{DISC}
\EventYear{2023}
\EventDate{October 10-12, 2023}
\EventLocation{L'Aquila, Italy}
\EventLogo{}
\SeriesVolume{281}
\ArticleNo{29}
\begin{document}

\maketitle
\addtocontents{toc}{\protect\setcounter{tocdepth}{0}}

\begin{abstract}
  \input{abstract}
\end{abstract}

\section{Introduction}
\label{sec:intro}

\input{Sections/01Introduction.tex}

\section{Preliminaries}
\label{sec:preliminaries}
\input{Sections/02Preliminaries.tex}

\section{Mathematical framework for source anonymity in general graphs}
\label{sec:privacy-gossip}
\input{Sections/03PrivacyOfGossip}

\section{Fundamental limits of source anonymity: lower bound on $\mathbf{\varepsilon}$}
\label{sec:impossibility}
\input{Sections/04Impossibility}

\section{Privacy guarantees: upper bound on $\mathbf{\varepsilon}$ }
\label{sec:guarantees}
\input{Sections/05Guarantees}

\section{Trade-off: Dissemination time vs. privacy}
\label{sec:analysis}
\input{Sections/06Analysis}

\section{Summary \& future directions}
\label{sec:conclusion}
\input{Sections/07Conclusion}

%% Please use bibtex, 

\bibliography{Bibliography}
\addtocontents{toc}{\protect\setcounter{tocdepth}{2}}
\renewcommand{\contentsname}{Appendix Contents}
\newpage
\tableofcontents
\newpage
\appendix
\input{Sections/08Appendix}

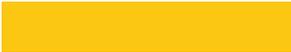
\end{document}

%% file: abstract.tex
Detecting the \emph{source of a gossip} is a critical issue, related to identifying \emph{patient zero} in an epidemic, or the \emph{origin of a rumor} in a social network. Although it is widely acknowledged that random and local gossip communications make source identification difficult, there exists no general quantification of the level of anonymity provided to the source. This paper presents a principled method based on $\varepsilon$-\emph{differential privacy} to analyze the inherent source anonymity of gossiping for a large class of graphs. First, we quantify the fundamental limit of source anonymity any gossip protocol can guarantee in an arbitrary communication graph. In particular, our result indicates that when the graph has poor connectivity, no gossip protocol can guarantee any meaningful level of differential privacy. This prompted us to further analyze graphs with controlled connectivity. We prove on these graphs that a large class of gossip protocols, namely \emph{cobra walks}, offers tangible differential privacy guarantees to the source. In doing so, we introduce an original proof technique based on the reduction of  a gossip protocol to what we call a \emph{random walk with probabilistic die out}. This proof technique is of independent interest to the gossip community and readily extends to other protocols inherited from the security community, such as the \emph{Dandelion} protocol. Interestingly, our tight analysis precisely captures the \emph{trade-off} between dissemination time of a gossip protocol and its source anonymity.

%% file: Sections/01Introduction.tex
A gossip protocol (a.k.a., an epidemic protocol) is a distributed algorithm that disseminates information in a peer-to-peer system~\cite{pittel1987spreading,acan2017push,karp2000randomized,kowalski2013estimating,dwork1988consensus, 10.1145/2450142.2450147}. Gossip protocols have been long used to model the propagation of infectious diseases~\cite{hethcote2000mathematics, math_of_epidemics, using_features_predict_epidemics}, as well as rumors in social networks where users randomly exchange messages~\cite{doerr2011social,giakkoupis2015privacy}. It is commonly accepted that random and local communications between the users make source identification hard, and thus provide \emph{inherent} anonymity to the source of the gossip, i.e.,  anonymity that comes solely from the spreading dynamic without relying on any additional cryptographic primitives (as in~\cite{gossip_that_preserves_privacy_for_distr_computing}). Source anonymity in gossip protocols constitutes an active area of research. On the one hand, many works aim to establish \emph{privacy guarantees} for the source of the gossip by concealing it against an adversary, e.g., hiding the whistleblower on social media~\cite{Karol2017LocationHiding, 5961737, hiding_the_source_fanti,giakkoupis2015privacy, bojja_venkatakrishnan_dandelion_2017, irreg_source_obfuscation}. On the other hand, a large effort  is put towards identifying \emph{privacy limits} for the source of a gossip by designing adversarial strategies that accurately recover the source, e.g., ``patient zero’’ identification in epidemics~\cite{review_on_identifying_sources,degradation_anonymous_protocols,Pinto_2012,who_is_the_culprit, 9318999, liu2019information}.

Although a significant amount of research  is dedicated to the investigation of source anonymity, existing approaches (as summarized in~\cite{review_on_identifying_sources}) mainly focus on specific settings, such as locating the source of a gossip for a particular protocol, hiding it against a chosen adversarial strategy or examining the problem on a narrow family of graphs (trees, complete graphs, etc.). This prevents the results from being generalized, and it remains unclear how hard it is to recover the source of a gossip in general,  naturally raising the following question.

\begin{center}
\emph{What are the fundamental limits and guarantees on the inherent \\ source anonymity of gossiping in a general setting?}    
 \end{center}
 
We take an important step towards addressing this question by adapting the celebrated mathematical framework of $\varepsilon$-differential privacy ($\varepsilon$-DP) to our context~\cite{Dwork_2006,Dwork_2013}. Although the concept is a gold standard to measure privacy leakage from queries on tabular databases, it can be also adapted to different privacy semantics and threat models~\cite{DesfontainesPejo2020Survey}. In our context, we use $\varepsilon$-DP to measure the \emph{inherent} source anonymity of gossiping in general graphs. We adopt a widely used threat model where the adversary aims to guess the source by monitoring the communications of a set of \emph{curious} nodes in the graph~\cite{review_on_identifying_sources, Pinto_2012, Anonymous_communication_Crowds, degradation_anonymous_protocols, towards_measuring_anonymity, hiding_the_source_fanti}. Using differential privacy enables us to overcome the limitations of previous work, as DP guarantees hold regardless of the exact strategy of the attacking adversary. Additionally, DP guarantees can be combined with any prior knowledge the adversary has on the location of the source, making our results generalizable. Our contributions can be summarized as follows.

\subsection{Main results}

We propose a mathematical framework that adapts the concept of differential privacy to quantify source anonymity in any graph (Section~\ref{sec:privacy-gossip}). In doing so, we highlight the importance of considering two types of adversaries: 
the \emph{worst-case} and the \emph{average-case}. For the worst-case adversary, we focus on privacy guarantees that hold \emph{regardless} of the location of the curious nodes in the graph. In other words, these guarantees hold even if the adversary knows the communication graph in advance and chooses curious nodes strategically. For the average-case adversary, we focus on privacy guarantees that hold with high probability when curious nodes are chosen uniformly at random. Here, the adversary does not know the structure of the underlying communication graph in advance. Within our mathematical framework, we establish the following results for both adversarial cases.

\myparagraph{Privacy limits.} We first quantify a fundamental limit on the level of $\varepsilon$-DP any gossip protocol can provide on any graph topology (Section \ref{sec:impossibility}).  This result indicates that no gossip protocol can ensure any level of differential privacy on poorly connected graphs. This motivates us to consider graphs with controlled connectivity, namely expander graphs.
Expanders are an important family of strongly connected graphs that are commonly considered in the gossip protocols literature~\cite{boyd2006randomized,guo2014gossip,CobraExpanders2016}. On this class, we get the following results. 

\myparagraph{Privacy guarantees.} We prove that a large class of gossip protocols provides tangible differential privacy guarantees to the source (\sectionref{sec:guarantees}).  We first consider the parameterized family of gossip protocols known as $(1+\rho)$-cobra walks~\cite{Dutta2013Cobra,CobraExpanders2016,Mitzenmacher2018BetterBF,berenbrin_tight_2018}, which constitutes a natural generalization of a simple random walk. A cobra walk can be seen as an SIS (Susceptible-Infected-Susceptible) epidemic, a well-established model for analyzing the spread of epidemics and viruses in computer networks~\cite{hethcote2000mathematics, math_of_epidemics}. In particular, a $(1+\rho)$-cobra walk is an instance of an SIS epidemic scheme where active nodes constitute the infectious set, the duration of the infectious phase is equal to one and every infected node can only infect one or two of its neighbors at a time. In order to establish  differential privacy guarantees on this class of gossip protocols, we rely on the critical observation that the cobra walk has a quantifiable probability of mixing before hitting a curious node (see Section~\ref{sec:ProofTechniques} for more details on this observation). This characteristic is not unique to cobra walks, as it is shared by several other types of gossip protocols. Accordingly, we also show how to generalize our privacy guarantees to the $\rho$-Dandelion protocol~\cite{bojja_venkatakrishnan_dandelion_2017}, first introduced as an anonymous communication scheme for Blockchains.

\myparagraph{Dissemination time vs. privacy trade-off.}
As an important by-product of our analysis, we precisely capture the trade-off between dissemination time and privacy of a large class of gossip protocols operating on sufficiently dense graphs we call near-Ramanujan graphs. The privacy-latency tension has been  suggested several times in the literature~\cite{bojja_venkatakrishnan_dandelion_2017, who_started_this_rumor, DP_gossip_in_general_networks}. However, our work presents the first formal proof of this long-standing empirical observation. Specifically, we show that our privacy guarantees are tight for both $(1+\rho)$-cobra walks~\cite{CobraExpanders2016} and $\rho$-Dandelion protocol~\cite{bojja_venkatakrishnan_dandelion_2017}. Additionally, we give a tight analysis of the dissemination time as a function of parameter $\rho$. This analysis leads  us to conclude that increasing parameter $\rho$ results in a faster dissemination, but decreases privacy guarantees of the protocol, formally establishing the existence of a trade-off between privacy and dissemination time. As cobra walks are strongly related to SIS-epidemics, and Dandelion to anonymous protocols in peer-to-peer networks, our results are relevant for both epidemic and source anonymity communities.

\subsection{Technical challenges \& proof techniques} 
\label{sec:ProofTechniques}

A major technical contribution of our paper is the privacy guarantee of $(1+\rho)$-cobra walks in non-complete graphs. The derivation of this result has been  challenging to achieve for two reasons.  Firstly, our objective is to establish differential privacy guarantees in general graphs, which is a more complex scenario than that of complete graphs (as seen in~\cite{who_started_this_rumor}), where any communication between pairs of nodes is equiprobable, and symmetry arguments can be utilized. Yet, this technique is no longer applicable to our work. The fact that no symmetry assumptions about graph structure can be made calls for new more sophisticated proof techniques. Second, cobra walks are challenging to analyze directly. State-of-the-art approaches analyzing the dissemination time of cobra walks circumvent this issue by analyzing a dual process instead, called BIPS~\cite{CobraExpanders2016, cooper_improved_2017, berenbrin_tight_2018}. There, the main idea is to leverage the duality of BIPS and cobra walks with respect to hitting times~\cite{CobraExpanders2016}. While hitting times provide sufficient information for analyzing the dissemination time of a cobra walk, they cannot be used to evaluate differential privacy, as they do not provide sufficient information about the probability distribution of the dissemination process. We overcome this difficulty through a two-step proof technique, described below.

\myparagraph{Step I: Reduction to a random walk with probabilistic die out.} To establish $\varepsilon$-differential privacy, we essentially show that two executions of the same $(1+\rho)$-cobra walk that started from different sources are statistically indistinguishable to an adversary monitoring a set of curious nodes. In doing so, we design a novel proof technique that involves reducing the analysis of gossip dissemination in the presence of curious nodes, to a \emph{random walk with probabilistic die out}. Such a protocol behaves as a simple random walk on the communication graph $G$, but it is killed at each step (i) if it hits a curious node, or otherwise (ii) with probability $\rho$. We show that disclosing the death site of such a random walk to the adversary results in a bigger privacy loss than all the observations reported by the curious nodes during the gossip dissemination. Then, we can reduce the privacy analysis of cobra walks to the study of such a random walk with probabilistic die out.

\myparagraph{Step II: Analysis of a random walk %walks 
with probabilistic die out.} To study a random walk with probabilistic die out, we characterize the spectral properties of the (scaled) adjacency matrix $\vec Q$ corresponding to the subgraph of $G$ induced by the non-curious nodes. In particular, we show that if curious nodes occupy a small part of every neighborhood in $G$, then the subgraph induced by non-curious nodes (i) is also an expander graph (\lemmaref{lemma:spectral-Q}) and (ii) has an almost-uniform first eigenvector (\lemmaref{lemma:Q-delocalization}). While (i) is a direct consequence of the Cauchy Interlacing Theorem, (ii) is more challenging to obtain. We need to bound $\vec Q$ from above and below by carefully designed matrices with an explicit first eigenvector (\lemmaref{lemma:spectral-Q-hat}). Combining (i) and (ii) allows us to precisely estimate the behavior of the random walk with probabilistic die out, which yields the desired differential privacy guarantees.

\myparagraph{Generality of the proof.} The reduction to a random walk with probabilistic die out is the most critical step of our proof. It is general and allows us to analyze several other protocols without having to modify the most technical part of the proof (Step II above). We demonstrate the generality of this technique by applying this reduction to the Dandelion protocol and obtain similar privacy guarantees to cobra walks (\lemmaref{lemma:dandelion-reduction} and \theoremref{thm:main-dandelion}).

\subsection{Related work}

\myparagraph{Inherent anonymity of gossiping.} 
To the best of our knowledge, only two previous works have attempted to quantify the inherent source anonymity of gossiping through differential privacy~\cite{who_started_this_rumor,DP_gossip_in_general_networks}. The former work~\cite{who_started_this_rumor} is the first to analyze source anonymity using differential privacy. It measures the guarantees of a class of gossip protocols with a muting parameter (which we call ``muting push'' protocols) and contrasts these guarantees with the dissemination time of these protocols on a complete graph. Both the threat model and the nature of the technical results in~\cite{who_started_this_rumor} heavily depend on the completeness of the graph. In such a context, the analysis is considerably simplified for two reasons. Firstly, the presence of symmetry allows for the curious node locations to be ignored, rendering the average-case and the worst-case adversaries equivalent. 
Secondly, in contrast to what would happen in non-complete graphs, since any node can communicate with any other node in each round, a single round of communication is sufficient to hide the identity of the source. However, when considering the spread of epidemics or the propagation of information in social networks, communication graphs are seldom complete~\cite{Melancon2006HowDenseAreGraphs}. Our work highlights that non-completeness of the graph potentially challenges the differential privacy guarantees that gossip protocols can achieve and also makes it important to distinguish between average and worst-case threat models. Therefore, our results constitute a step toward a finer-grained analysis of the anonymity of gossiping in general graphs. Note that our work can be seen as a strict generalization of the results of~\cite{who_started_this_rumor}, since, in addition to cobra walks and Dandelion, we also show that our proof techniques described in \sectionref{sec:ProofTechniques} apply to ``muting push'' protocols (see \appendixref{sec:muting-privacy}).

The second approach~\cite{DP_gossip_in_general_networks} addresses a problem that appears to be similar to ours at first glance, as it aims to quantify source anonymity in non-complete graphs. However, the authors consider a different threat model, where an adversary can witness any communication with some probability instead of only those passing through the curious nodes. Furthermore, the paper only gives negative results and does not provide any differential privacy guarantees, which is the most technically challenging part of our paper.  

\myparagraph{Dissemination time vs. privacy trade-off.}
Several previous works~\cite{venkitasubramaniam_anonymity_2008, beimel2003buses, das2018anonymity, snader2008tune} have suggested the existence of a tension between source anonymity (i.e., privacy) and latency of message propagation. Under the threat model we consider in this work (with curious nodes),~\cite{bojja_venkatakrishnan_dandelion_2017} conjectured that the Dandelion protocol would exhibit a trade-off between (their definition of) source anonymity and dissemination time. Later, works~\cite{who_started_this_rumor} and~\cite{DP_gossip_in_general_networks} provided more tangible evidence for the existence of a dissemination time vs. privacy trade-off when analyzing source anonymity through differential privacy. 
However, these works do not provide a tight analysis of the tension between dissemination time and privacy, hence making their observation incomplete. To the best of our knowledge, our work is the first to rigorously demonstrate the existence of a trade-off between the dissemination time of a gossip protocol and the privacy of its source thanks to the \emph{tightness} of our analysis.

%% file: Sections/02Preliminaries.tex
For a vector $\vec x \in \mathbb R^m$, we denote by $x_i$ its $i$th coordinate, i.e., $\vec x = (x_1, x_2, \ldots,  x_m)^\top$. Similarly, for a matrix $\vec M \in \mathbb R^{m \times m'}$, we denote by $M_{ij}$ its entry for the $i$th row and $j$th column. Furthermore, for any symmetric matrix $\vec M \in \mathbb R^{m \times m}$, we denote by $\lambda_1(\vec M) \ge \lambda_2(\vec M) \ge \ldots \ge \lambda_m(\vec M)$ its eigenvalues. We use $\vec 1_m \in \mathbb R^m$ to denote an all-one vector, $\vec I_m \in \mathbb R^{m\times m}$ to denote the identity matrix, $\vec J_{m}\in \mathbb R^{m\times m}$ to denote an all-one square matrix, and $\vec O_{m \times m'} \in \mathbb R^{m\times m'}$ to denote an all-zero matrix. Finally, for any $\vec x \in \mathbb R^m$, we denote by $\norm{\vec x}_p \triangleq \left( \sum_{i=1}^{m} |x_i|^p \right)^{1/p}$ the $\ell_p$ norm of $\vec x$ for $p \in [1, \infty)$ and by $\norm{\vec x}_\infty \triangleq \max_{i \in m} |x_i|$ the $\ell_\infty$ norm of $\vec x$.

Throughout the paper, we use the \emph{maximum divergence} to measure similarities between probability distributions. We consider below a common measurable space $(\Omega, \Sigma)$ on which the probability measures are defined. Let $\mu$, $\nu$ be two probability measures over $\Sigma$.
The \emph{max divergence} between $\mu$ and $\nu$ is defined as\footnote{Note that we allow $\nu(\sigma) = 0$ in the definition. If $\nu(\sigma) = 0$ but $\mu(\sigma) > 0$ for some $\sigma \in \Sigma$, the max divergence is set to $\infty$ by convention.}
\[
\divinfty{\mu}{\nu} \ \triangleq \sup_{\sigma \in \Sigma,~ \mu(\sigma) > 0} \ln\frac{\mu(\sigma)}{\nu(\sigma)} .
\]
Furthermore, for two random variables $X,Y$ with laws $\mu$ and $\nu$ respectively, we use the notation $\divinfty{X}{Y}$ to denote $\divinfty{\mu}{\nu}$.

\subsection{Graph theoretical terminology}
\label{sec:pre-graphs}
Consider an undirected connected graph $G = (V,E)$, where $V$ is the set of nodes and $E$ is the set of edges. $G$ cannot have self-loops or multiple edges.  For any $v \in V$, we denote by $N(v)$ the set containing the neighbours of $v$ in $G$ and by $\deg(v)$ the number of edges incident to $v$. Furthermore, $G$ is said to be a \emph{regular graph}, if there exists $d(G)$ such that $\deg(v) = d(G)$ for every $v \in V$; $d(G)$ is called the degree of the graph. Additionally, for a set $U \subseteq V$ and $v \in V$, we denote by $\degset{U}{v}$ the number of neighbours of $v$ contained in $U$, i.e., $\degset{U}{v} = |N(v) \cap U|$. Below, we introduce some additional graph terminology. 

\begin{definition}[Vertex cut \& connectivity]
    A \emph{vertex cut} of $G$ is a subset of vertices $K \subseteq V$ whose removal disconnects $G$ or leaves just one vertex. A minimum vertex cut of $G$ is a vertex cut of the smallest size. The size of a minimum vertex cut for $G$, denoted $\kappa(G)$, is called the vertex connectivity of $G$.
\end{definition} 

Consider an undirected connected graph $G = (V,E)$ of size $n$ where $V$ is an ordered set of nodes. We denote by $\vec A$ the adjacency matrix of $G$, i.e., $A_{vu} = 1$ if $\{v,u\} \in E$ and $A_{vu} = 0$ otherwise. We also denote by $\Hat{\vec A} = \vec D^{-1/2}\vec A\vec D^{-1/2}$ the normalized adjacency matrix of $G$, where $\vec D$ is the diagonal degree matrix, i.e., $D_{vu} = \deg(v)$ if $v=u$ and $0$ otherwise. Since $\Hat{\vec A}$ is a symmetric and normalized matrix, the eigenvalues of $\Hat{\vec A}$ are real valued and $\lambda_1(\Hat{\vec A}) = 1$. Using this terminology, the \emph{spectral expansion} of $G$ is defined as
\begin{equation}
    \lambda(G) \triangleq \max\{|\lambda_2(\Hat{\vec A})|, |\lambda_n(\Hat{\vec A})|\}.
\end{equation}

\begin{definition}[Expander graph]
\label{def:expander-graphs}
Consider an undirected regular graph $G$. If $d(G) = d$ and $\lambda(G) \le \lambda$, then $G$ is said to be a $(d,\lambda)$-expander graph.
\end{definition}

\subsection{Gossip protocols}
\label{section:gossip}
Consider an undirected connected communication graph $G=(V,E)$ where two nodes $u,v \in V$ can directly communicate if and only if $\{u,v\} \in E$. One node $s \in V$, called the \emph{source}, holds a unique gossip $g$ to be propagated throughout the graph. In this context, \emph{a gossip protocol} is a predefined set of rules that orchestrates the behavior of the nodes with regard to the propagation of $g$. Essentially, the goal of a protocol is that with probability $1$ every node in $G$ eventually receives $g$. We assume discrete time steps and synchronous communication, i.e., the executions proceed in rounds of one time step.\footnote{Although, for clarity, we focus on a synchronous communication, our analysis of privacy guarantees in \sectionref{sec:guarantees} readily extends to an asynchronous setting.} While every node in $G$ has access to the global clock, we assume that the execution of the protocol starts at a time $\starttime\in\mathbb Z$, which is \emph{only} known to the source $s$. 

\myparagraph{Execution of a gossip protocol.} At any point
of the execution of the protocol, a node $u\in V$ can either be active or non-active. Only active nodes are allowed to send messages during the round. A gossip protocol always starts with the source $s$ being the only active node, and at every given round $t+1$ active nodes are the nodes that received the gossip at round $t$. We will use $\activeset_t \subseteq V$ to denote the set of active nodes at the beginning of round $t\ge\starttime$ and set $\activeset_{\starttime} = \{s\}$ by convention. 
Denoting by $(u \to v)$ a communication between nodes $u$ and $v$, we define $\mathcal C$ to be the set of all possible communications in $G$, i.e., $\mathcal C = \left \{(u\to v): \{u,v\} \in E\right\} \cup \left \{( u\to u): u \in V\right\}$. Note that we allow an active node $u$ to send a fictitious message to itself to stay active in the next communication round. Then, the $t^{\text{th}}$ round of an execution for a given protocol $\mathcal{P}$ can be described by a pair $(\activeset_t, \commset_t)$, where $\activeset_t \subseteq V$ is a set of active nodes, and $\commset_t$ is the (multi)set of communications of $\mathcal C$ which happened at round $t$. We denote by $S$ the random variable characterizing the \emph{execution} of the protocol. Naturally, an \emph{execution} is described by a sequence of rounds, i.e., $S = \{(\activeset_t, \commset_t)\}_{t \ge \starttime}$. We define \emph{expected dissemination time} of the protocol as the expected number of rounds for all nodes to receive the gossip during an execution. Finally, we denote $\execs$ the set of all possible executions.

\myparagraph{Cobra and random walk. } Coalescing-branching random walk protocol (a.k.a., cobra walk)~\cite{Dutta2013Cobra, CobraExpanders2016, Mitzenmacher2018BetterBF,berenbrin_tight_2018} is a natural generalization of a simple random walk that is notably useful to model and understand Susceptible-Infected-Susceptible (SIS) epidemic scheme~\cite{hethcote2000mathematics, math_of_epidemics}. We consider a $(1 + \rho)$-cobra walk as studied in~\cite{CobraExpanders2016} with $\rho \in [0,1]$\footnote{Some prior works also study $k$-cobra walks with branching parameter $k \ge 3$ ~\cite{Dutta2013Cobra}. We do not consider this class, since our negative result for a $2$-cobra walk (\theoremref{thm:cobra-privacy-lower}) implies that a $k$-cobra walk for any $k \ge 3$ does not satisfy a reasonable level of differential privacy.}. This is a gossip protocol where at every round $t \geq \starttime$, each node $u \in \activeset_t$ samples a token from a Bernoulli distribution with parameter $\rho$. If the token equals zero, $u$ samples uniformly at random a node $v$ from its neighbors $N(u)$ and communicates the gossip to it, i.e., $(u\rightarrow v)$ is added to $C_t$. If the token equals one, the protocol \emph{branches}. Specifically, $u$ independently samples two nodes $v_1$ and $v_2$ at random (with replacement) from its neighbors and communicates the gossip to both of them, i.e., $(u\rightarrow v_1),$ and $(u\rightarrow v_2)$ are added to $C_t$. At the end of the round, each node $u \in \activeset_t$ deactivates.  Note that, when $\rho = 0$, this protocol degenerates into a simple random walk on the graph; hence it has a natural connection with this random process.

\myparagraph{Dandelion protocol.} Dandelion is a gossip protocol designed to enhance source anonymity in the %a 
Bitcoin peer-to-peer network. Since it was introduced in~\cite{bojja_venkatakrishnan_dandelion_2017}, it has received a lot of attention from the cryptocurrency community. Dandelion consists of two phases: (i) the anonymity phase, and (ii) the spreading phase. The protocol is parameterized by $\rho \in [0,1)$, the probability of transitioning from the anonymity phase to the spreading phase. Specifically, the phase of the protocol is characterized by a token $\anonphase \in \{0,1\}$ held by a global oracle and initially equal to $0$.
At the beginning of each round of the Dandelion execution, if $\anonphase = 1$ the global oracle sets $\anonphase = 0$ with probability $\rho$ and keeps $\anonphase = 1$ with probability $1-\rho$. Once $\anonphase = 0$, the global oracle stops updating the token. Based on this global token, at each round, active nodes behave as follows. If the $\anonphase = 1$, the execution is in the anonymity phase and an active node $u$ samples a node $v$ uniformly at random from its neighborhood $N(u)$ and communicates the gossip to it, i.e., $(u\rightarrow v)$ is added to $C_t$. Afterwards, node $u$ deactivates, i.e., in the anonymity phase only one node is active in each round. If the $\anonphase = 0$, the execution is in the spreading phase. Then the gossip is broadcast, i.e., each node $u \in \activeset_{t}$ communicates the gossip to all of its neighbors and for $\forall v \in N(u)$, $(u\rightarrow v)$ is added to $C_t$.

%% file: Sections/03PrivacyOfGossip.tex
Given a source and a gossip protocol, we fix the probability space $(\execs, \sigmaalgebra, \Pr)$, where $\sigmaalgebra$ is the standard cylindrical $\sigma$-algebra on $\execs$ (as defined in Appendix A.1 of~\cite{Yu2014AUV}) and $\Pr$ is a probability measure characterizing the executions of the protocol. In the remaining, to avoid measurability issues, we  only refer to subsets of $\execs$ from $\sigmaalgebra$.

\subsection{Measuring source anonymity with differential privacy}

\label{sec:source-anonymity-with-DP}

We now describe the mathematical framework we use to quantify source anonymity of gossiping. We consider a threat model where an external adversary has access to a subset $F \subset V$ of size $f < n -1$ of \emph{curious} nodes. Curious nodes in $F$ execute the protocol correctly, but report their communications to the adversary. The adversary aims to identify the source of the gossip using this information. We distinguish two types of adversaries, namely worst-case and average-case, depending on the auxiliary information they have on the graph.

\myparagraph{Threat models: worst-case and average-case adversaries.} On the one hand, a \emph{worst-case} adversary is aware of the structure of the graph $G$ and may choose the set of curious nodes to its benefit. On the other hand, the \emph{average-case} adversary is not aware of the topology of $G$ before the start of the dissemination, hence the set of curious nodes is chosen uniformly at random among all subsets of $V$ of size $f$. We assume that the messages shared in the network are unsigned and are passed unencrypted. Also, the contents of transmitted messages (containing the gossip) do not help to identify the source of the gossip. In other words, adversaries can only use the information they have on the dissemination of the gossip through the graph to locate the source. We also assume that the adversary does not know the exact starting time $\starttime \in \mathbb Z$ of the dissemination. 
To formalize the observation received by the external adversary given a set of curious nodes $F$, we introduce a function $\fadv$ that takes as input communications $\commset$ from a single round and outputs only the communications of $\commset$ visible to the adversary. Note that a communication $(v\to u)$ is visible to the adversary if and only if either $v$ or $u$ belongs to $F$. 
Consider an execution $S = \{(\activeset_t, \commset_t)\}_{t\ge \starttime}$ of a gossip protocol, and denote by $\tadv$ the first round in which one of the curious nodes received the gossip. Then we denote by $S_\textsc{adv} = \{\fadv(\commset_t)\}_{t \ge \tadv}$ the random variable characterizing the observation of the adversary for the whole execution. Note that the adversary does not know $\starttime$, hence it cannot estimate how much time passed between $\starttime$ and $\tadv$.

\begin{remark}
    For Dandelion, the adversary actually also has access to the value of $\anonphase$ in round $t$, i.e., we have $S_\textsc{adv} = \{\fadv(\commset_t), \anonphase_t\}_{t \ge \tadv}$. We omit this detail from the main part of the paper for simplicity of presentation, but it does not challenge our results on privacy guarantees. See \appendixref{sec:dandelion-privacy} for more details.
\end{remark}

\myparagraph{Measuring source anonymity.} We formalize source anonymity below by adapting the well-established definition of differential privacy. In the remaining of the paper, for a random variable $A$, we will write $A^{(s)}$ to denote this random variable conditioned on the node $s\in \honest$ being the source. In our setting, we say that a gossip protocol satisfies differential privacy if for any $u,v \in V$ the random sequences $\seqadv{v}$ and $\seqadv{u}$ are statistically indistinguishable. More formally, we define differential privacy as follows. 

\begin{definition}[Differential privacy]
    Consider an undirected graph $G=(V,E)$ and a set of curious nodes $F \subset V$. Then, a gossip protocol satisfies $\varepsilon$-\emph{differential privacy} ($\varepsilon$-DP) for the set $F$ if, for any two nodes $v,u\in \honest$, the following holds true
\[
\divinfty{\seqadv{v}}{\seqadv{u}} \le \varepsilon.
\]
\end{definition}

When establishing differential privacy guarantees against a \emph{worst-case adversary}, we aim to find a value $\varepsilon$ which only depends on the number of curious nodes $f$, and is \emph{independent} of the identity of the nodes in $F$. Accordingly, we say that a gossip protocol satisfies \emph{$\varepsilon$-DP against a worst-case adversary} if it satisfies $\varepsilon$-DP for any set $F \subset V$ such that $|F|=f$. 
    
When establishing differential privacy against an \emph{average-case adversary}, we aim to find a value of $\varepsilon$ for which the protocol satisfies $\varepsilon$-DP \emph{with high probability}\footnote{An event is said to hold with high probability on graph $G$ of size $n$, if it holds with probability $\ge 1 - 1/n$.} when choosing the $f$ curious nodes uniformly at random from $V$. Formally, let $\uniformsubset{f}{V}$ be the uniform distribution over all subsets of $V$ of size $f$, a gossip protocol satisfies \emph{$\varepsilon$-DP against an average-case adversary} if
\begin{equation}
\label{eq:average-guarantee-probability}  
  \Pr_{F \sim \uniformsubset{f}{V}} \left[\max_{v,u\in \honest} \divinfty{\seqadv{v}}{\seqadv{u}} \le \varepsilon \right] \ge 1 - \frac{1}{n}.
\end{equation}

\subsection{Semantic of source anonymity}
\label{subsection:DPsemantics}

Differential privacy is considered the gold standard definition of privacy, since $\varepsilon$-DP guarantees hold \emph{regardless} of the strategy of the adversary and any prior knowledge it may have on the location of the source. Yet, the values of $\varepsilon$ are notoriously hard to interpret~\cite{lee2011much, 6957125}. To better understand the semantic of our definition of differential privacy, we consider below two simple examples of adversarial strategies: maximum a posteriori and maximum likelihood estimations. For these strategies, we derive bounds on the probability of an adversary successfully guessing the source in an effort to give a reader an intuition on the meaning of the parameter $\varepsilon$. The proofs are deferred to \appendixref{appendix:semantic}.

\myparagraph{Maximum a posteriori strategy.}
Maximum a posteriori (MAP) strategy can be described as follows. Suppose an adversary has an a priori distribution $\prior$ that assigns to every node in $V\backslash F$ a probability of being the source of the gossip. Intuitively, $\prior$ corresponds to the set of beliefs the adversary has on the origin of the gossip before observing the dissemination. This prior might reflect information acquired from any auxiliary authority or some expert knowledge on the nature of the protocol. Suppose the adversary observes an event $\sigmaevent$. Then, a MAP-based adversary ``guesses'' which node is the most likely to be the source, assuming event $\sigma$ occurred and assuming the source has been sampled from the prior distribution $\prior$. Such guess is given by
\begin{equation}
    \label{eq:map}
    \smap = \argmax_{v\in \honest}\Pr_{s \sim \prior}\left[v = s \mid \seqadv{s} \in \sigmaevent \right] =\argmax_{v\in \honest}\Pr\left[\seqadv{v} \in \sigmaevent \right] \prior(v). 
\end{equation}
Using $\varepsilon$-DP, we can upper bound the success probability of such a guess. Suppose the protocol satisfies $\varepsilon$-DP, then the probability of correctly identifying a source $s\sim p$ conditioned on $\sigma$ happening is upper bounded as follows
\begin{equation}
    \label{eq:map-success}
    \Pr_{s \sim \prior}\left[\smap = s \mid \seqadv{s} \in \sigmaevent \right] \le \exp(\varepsilon) p\left(\smap\right). 
\end{equation}
Such an upper bound has a simple interpretation. Note that $p(\smap)$ characterizes the maximum probability of a successfully guessing $\smap$ based solely on adversary's prior knowledge. Then, the upper bound above states that the probability of a successful guess after observing the dissemination is amplified by a factor of at most $\exp(\varepsilon)$ compared to success probability of a guess based on a priori knowledge only.

\myparagraph{Maximum likelihood strategy.} Maximum likelihood estimation (MLE) occupies a prominent place~\cite{hiding_the_source_fanti, who_is_the_culprit, 10.1145/1811099.181106, Pinto_2012} in the literature, both for designing source location attacks, and for defending against adversaries that follow an MLE strategy. This method is a special instance of MAP estimator in (\ref{eq:map}) with a uniform prior distribution $\prior = \uniform{\honest}$ on the source. 
We can show that, if the protocol satisfies $\varepsilon$-DP, such guess has a bounded success probability.
\begin{equation}
\label{eq:mle-success}
\Pr_{s \sim \uniform{\honest}}\left[\smle = s \mid \seqadv{s} \in \sigmaevent \right] \le \frac{\exp(\varepsilon)}{n-f}.
\end{equation}

%% file: Sections/04Impossibility.tex
We start by studying the fundamental limits of differential privacy in general graphs. Specifically, we aim to show that vertex connectivity constitutes a hard threshold on the level of source anonymity gossiping can provide. First, we present a warm-up example indicating that in a poorly connected graph, no gossip protocol can achieve any meaningful level of differential privacy against a worst-case adversary. We then validate this intuition by devising a universal lower bound on $\varepsilon$ that applies for any gossip protocol and any undirected connected graph. Complete proofs related to this section can be found in \appendixref{appendix:impossibility}.

\subsection{Warm-up}
\label{subsect:warmup}

Consider a non-complete graph $G=(V,E)$ and $K \subset V$, a vertex cut of $G$. Then, by definition, deleting $K$ from $G$ partitions the graph into two disconnected subgraphs. When $f \geq \left| K \right|$, a worst-case adversary can take $F$ such that $K\subseteq F$. Then, the curious nodes can witness all the communications that pass from one subgraph to the other. Intuitively, this means that any two nodes that are not in the same subgraph are easily distinguishable by the adversary. Hence, differential privacy cannot be satisfied. This indicates that the level of differential privacy any gossip protocol can provide in a general graph fundamentally depends on the connectivity of this graph. 
To validate this first observation and determine the fundamental limits of gossiping in terms of source anonymity, we now determine a lower bound on $\varepsilon$.

\subsection{Universal lower bound on $\varepsilon$}
\label{section:universal-LB}

We present, in \theoremref{thm:universal-impossibility}, a universal lower bound on $\varepsilon$ which holds for any gossip protocol, on any connected graph and for both the worst-case and the average-case adversaries.

\begin{restatable}{theorem}{UniversalImprossibility}
    \label{thm:universal-impossibility}
    Consider an undirected connected graph $G = (V,E)$ of size $n$, a number of curious nodes $f > 1$, and an arbitrary gossip protocol $\mathcal{P}$. If $\mathcal{P}$ satisfies $\varepsilon$-DP against an average-case or a worst-case adversary, then 
    \[\varepsilon \ge \ln(f - 1).\]
    Moreover, if $\kappa(G) \le f$, then $\mathcal{P}$ cannot satisfy $\varepsilon$-DP with $\varepsilon < \infty$ against a worst-case adversary.
\end{restatable}

\begin{proof}[Proof sketch]
    To establish the above lower bound, we assume that the adversary simply 
    predicts that the first non-curious node to contact the curious set is the source of the gossip. As the definition of differential privacy does not assume a priori knowledge of the adversarial strategy, computing the probability of success for this attack provides a lower bound on $\varepsilon$. 
    
    We first demonstrate the result for the average-case adversary. Assume that $F$ is sampled uniformly at random from $V$. We can show that there exists $v \in V$ such that the attack implemented by the adversary succeeds with large enough probability when $v$ is the source of the gossip. This fact essentially means that this $v$ is easily distinguishable from any other node in the graph, which yields the lower bound $\varepsilon \geq \ln(f-1)$ in the average case.     
   We now consider the worst-case adversary. Assume that $F$ can be chosen by the adversary. As the lower bound $\varepsilon \ge \ln(f-1)$ holds with positive probability when $F$ is chosen at random, there exists at least one set $F$ for which it holds. Choosing this set of curious nodes establishes the claim for the worst-case adversary. Furthermore, when $\kappa(G) \le f$, we follow the intuition from \sectionref{subsect:warmup} to build a set $F$ that disconnects the graph. Using this set, we prove that $\varepsilon$ cannot be finite.
\end{proof}

\theoremref{thm:universal-impossibility} shows that the connectivity of the graph is an essential bottleneck for differential privacy in a non-complete graph. This stipulates us to study graphs with controlled connectivity, namely $(d,\lambda)$-expander graphs. Note that in a $(d,\lambda)$-expander, the vertex connectivity does not exceed $d$. Hence, \theoremref{thm:universal-impossibility} implies that no gossip protocol can satisfy any meaningful level of differential privacy against a worst-case adversary on a $(d,\lambda)$-expander if $f \ge d$. Considering this constraint, while studying a gossip against a worst-case adversary, we only focus on cases where the communication graph $G$ has a large enough degree $d$.

%% file: Sections/05Guarantees.tex
We now present a general upper bound on $\varepsilon$ that both holds for $(1 + \rho)$-cobra walks and $\rho$-Dandelion on $d$-regular graphs with fixed expansion, i.e., $(d,\lambda)$-expander graphs. Complete proofs related to this section can be found in \appendixref{sec:appendix-privacy-guarantees}. Our privacy guarantees are quite technical, which is justified by the intricacies of the non-completeness of the graph. Recall that, in the case of complete topologies analyzed in~\cite{who_started_this_rumor}, after one round of dissemination all information on the source is lost unless a curious node has been contacted. However, in a general expander graph, this property does not hold anymore.  Indeed, even after multiple rounds of propagation, the active set of the protocol can include nodes that are close to the location of the source $s$. Thus, differential privacy may be compromised.

\subsection{Adversarial density}
\label{sec:adversarialdensity}
The attainable level of source anonymity for a given protocol is largely influenced by the location of curious nodes. However, accounting for all possible placements of curious nodes is a very challenging and intricate task. To overcome this issue and state our main result, we first introduce the notion of \emph{adversarial density} that measures the maximal fraction of curious nodes that any non-curious node may have in its neighborhood. Upper bounding the adversarial density of a graph is a key element to quantifying the differential privacy guarantees of a gossip protocol. Formally, this notion is defined as follows.  

\begin{definition}
    Consider an undirected connected $d$-regular graph $G=(V,E)$, and an arbitrary set of curious nodes $F \subseteq V$. The \emph{adversarial density} of $F$ in $G$, denoted $\alpha_F$, is the maximal fraction of curious nodes that any node $v \in V \setminus F$ has in its neighborhood. Specifically, \[\alpha_F \triangleq \max_{v \in \honest} \frac{\deg_F(v)}{d}.\]
\end{definition}

For any set of curious nodes $F$, we have $\alpha_F \leq f/d$. Hence, even when $F$ is chosen by a worst-case adversary, the adversarial density is always upper bounded by $f/d$. However, for the average-case adversary we can obtain a much tighter bound, stated in~\lemmaref{lemma:adversarial-density} below. 

\begin{restatable}{lemma}{AdversarialDensity}
\label{lemma:adversarial-density}
    Consider an undirected connected $d$-regular graph $G=(V,E)$ of size $n$ and a set of curious nodes $F \sim \uniformsubset{f}{V}$, with adversarial density $\alpha_F$. We denote $\beta = f/n$ and $\gamma = \ln(n)/(e d)$, where $e$ is Euler's constant. Then, with probability at least $1-1/n$, $\alpha_F \leq \alpha$ with
    \[\alpha \leq 4e\frac{\max\{\gamma, \beta\}}{1 + \max \{\ln(\gamma) - \ln(\beta), 0\}}. \]
    Furthermore, if there exist $\delta > 0, c > 0$ such that $f/n > c$ and $d > \ln(n)/(c^2\delta^2)$ then a similar statement holds with $\alpha \leq (1 + \delta)\beta$. 
\end{restatable}

We deliberately state this first lemma in a very general form. This allows us to precisely quantify how the upper bound on the adversarial density improves as $f$ decreases. To make this dependency clearer, we provide special cases in which the bound on $\alpha_F$ is easily interpreted.  First, assume that $d \in \omega_n(\log(n))$ and $f/n \in \Omega_n(1)$.
Then, $\alpha_F$ is highly concentrated around $f/n$, up to a negligible multiplicative constant, when $n$ is large enough. On the other hand, when the ratio $f/n$ becomes subconstant, the concentration becomes looser. In particular, if $d \in \omega_n(\log(n))$ and $f/n \in o_n(1)$, then $\alpha_F \in o_n(1)$ with high probability. Finally, if $f/n$ drops even lower (e.g., when $f/n \in n^{-\Omega_n(1)}$), we get $\alpha_F \in O_n(1/d)$ or $\alpha_F \in n^{-\Omega_n(1)}$ with high probability for any $d$. 

\subsection{General upper bound on $\varepsilon$}

Thanks to~\lemmaref{lemma:adversarial-density} bounding adversarial density, we can now state our main theorem providing a general upper bound on $\varepsilon$ for $(1+\rho)$-cobra walks and $\rho$-Dandelion. 

\begin{restatable}{theorem}{MainThm}
\label{thm:main}
\label{thm:main-dandelion}
Consider an undirected connected $(d,\lambda)$-expander graph $G=(V,E)$ of size $n$, let $f$ be the number of curious nodes, and let $\mathcal P$ be a $(1+\rho)$-cobra walk with $\rho < 1$. Set $\alpha = f/d$ (resp. set $\alpha$ as in \lemmaref{lemma:adversarial-density}). If $\lambda < 1 - \alpha$, then $\mathcal P$ satisfies $\varepsilon$-DP against a worst-case adversary (resp. an average-case adversary) with 
\[
\varepsilon = \ln(\rho(n-f) + f) - 2\Tilde{T}\ln(1 - \alpha) - \Tilde{T}\ln(1 - \rho) - \ln(1-\lambda) + \ln(24),
\]
and $\Tilde{T}  = \left\lceil\log_{\frac{\lambda}{1 - \alpha}}\left(\frac{1 - \alpha}{4(n-f)}\right)\right\rceil \left(\log_{\frac{\lambda}{1 - \alpha}}(1 - \alpha) + 2\right) + 2.$ \\[10pt]
The above statement also holds if $\mathcal P$ is a $\rho$-Dandelion protocol with $\rho < 1$.
\end{restatable}

Note that the upper bound on $\varepsilon$ in \theoremref{thm:main} improves as the number of curious nodes $f$ decreases (since $\alpha$ decreases with $f$) or when the expansion improves (as $\lambda$ decreases, $\Tilde{T}$ also decreases). Yet, there is a complex interplay between the parameters $n,f,d,$ and $\lambda$ above. Additionally, we point out that for a worst-case adversary the privacy guarantees can be established only if $f/d < 1$. For the average-case, this assumption can be dropped, and we are able to establish positive results for $f$ as high as $\Theta_n(n)$. 
\section{Proof sketch for \theoremref{thm:main}}
\label{subsec:proof-structure}

Although results for worst-case and average-case adversaries have their own technical specificity, they both share the same general idea. Specifically, we introduce a random process that helps bounding from above the value of $\varepsilon$. This random process resembles a random walk that at each step reveals its position to the adversary with some probability that depends on $\rho$ and on the state of the process. We call this process a \emph{random walk with probabilistic die out}. Then, we show that such random walk mixes sufficiently well before its position is revealed, which provides indistinguishability between any two possible sources. 

The first half of our proof (step I) relies on the reduction of a gossip protocol to a random walk with probabilistic die out. This part is slightly different for different protocols, but for simplicity we only present step I for the cobra walk, and defer the proof for Dandelion to \appendixref{sec:dandelion-privacy}. In the second half (step II), we only analyze a random walk with probabilistic die out. It is hence universal and applies to both cobra walks and Dandelion protocols. 

\subsection{Step I: reduction to a random walk with probabilistic die out}

\input{tikz/ProbCobraWalk.tex}

Consider a $(1+\rho)$-cobra walk started at $s$ and denote $W^{(s)}$ the random variable indicating the last position of the cobra walk before it either branches or hits a curious node. More formally, if the round at which the cobra walk branches or contacts a curious node for the first time is $\tau$, then the active set at this round would be $\activeset_\tau^{(s)} = \{W^{(s)}\}$, with $W^{(s)} \in \honest$. We first show that disclosing $W^{(s)}$ to the adversary reveals more information about the source than $\seqadv{s}$ (see \lemmaref{lemma:sufficient-varepsilon-0}). Intuitively, this follows from the Markov property of the active set $\left \{\activeset_t^{(s)} \right \}_{t \geq \starttime}$ of the cobra walk. In fact, by definition of $\tau$, we have $\tau \leq \tadv$. Hence, the sequence of adversarial observations $\seqadv{v}$ can be obtained from $\activeset_\tau^{(s)} = \left \{ W^{(s)} \right \}$ via a randomized mapping independent of the initial source $s$. Then, using the data processing inequality Theorem 14 of~\cite{dataprocessing}) we show that for any two possible sources $u,v \in V \backslash F$, we have
\begin{equation}
    \label{eq:inequalitydiv}
    \divinfty{\seqadv{v}}{\seqadv{u}} \le \divinfty{W^{(v)}}{W^{(u)}}.
\end{equation}
This means that it suffices to obtain an upper bound on $\divinfty{W^{(v)}}{W^{(u)}}$ for any $u,v \in V \backslash F$ to obtain an appropriate value for $\varepsilon$. Then, we note that $W^{(s)}$ can be described as the death site of a process we refer to as \emph{random walk with probabilistic die out}, which was started at $s$. Such a process constitutes a random walk which is killed at each step either (i) if it hits a curious node, or otherwise (ii) with probability $\rho$. We illustrate this process in Figure \ref{fig:prob-die-out} and how it relates to the cobra walk.

\subsection{Step II: upper bounding the max divergence between death sites}

The rest of the proof is dedicated to analyzing the probability distribution of the death site of such a process. Let $\vec Q = \Hat{\vec A}[\honest]$ be the principled submatrix of $\Hat{\vec A}$ induced by the rows and columns of $V \setminus F$ and let $\vec R$ be a diagonal matrix of size $(n-f)\times (n-f)$ such that $R_{ww} = \degset{F}{w}/d$ for every $w \in \honest$. Then, $W^{(s)}$ can be described as an absorbing Markov chain. More precisely, let nodes from $\honest$ be transient states, and equip every node $w \in \honest$ with an absorbing state $\sink(w)$ which corresponds to the event of dying at $w$.
The transition matrix of our absorbing Markov chain can be written in a block form as
\begin{equation}
\label{eq:block-form-absorbing-markov}
\vec P = 
\begin{bmatrix}
(1 - \rho)\vec Q & \vec O_{(n-f) \times (n-f)} \\
\rho\vec I_{n-f} + (1 - \rho)\vec R & \vec I_{n - f} 
\end{bmatrix}.
\end{equation}
In the above, $\vec P_{xy}$ denotes the transition probability from a state $y$ to a state $x$. The first $n-f$ columns correspond to transition probabilities from transient states $w \in \honest$ and the last $n-f$ ones correspond to transition probabilities from absorbing states $\sink(w)$ for $w\in \honest$. The probability of transitioning between two transient states $v, u \in \honest$ (top-left block of $\vec P$) is defined similarly to a simple random walk on $G$, multiplied by the probability of not branching $(1 - \rho)$. The transition probability between $w$ and $\sink(w)$ (bottom-left block of $\vec P$) is naturally defined as the probability of branching plus the probability of contacting a curious node at the current step without branching. 

According to the above, being absorbed in $\sink(w)$ corresponds to the event $W^{(s)} = w$. Hence, using $\vec Q$ and $\vec R$ to compute a closed form expression for absorbing probabilities of the above Markov chain (see \lemmaref{lemma:absorbtion-prob}), we can rewrite  $\divinfty{W^{(v)}}{W^{(u)}}$ as follows
\begin{equation}
   \label{eq:equalitydiv}
   \divinfty{W^{(v)}}{W^{(u)}} = \max_{w\in \honest} \ln\frac{(\vec I_{n-f} - (1 - \rho)\vec Q)^{-1}_{vw}}{(\vec I_{n-f} - (1 - \rho)\vec Q)^{-1}_{uw}}.
\end{equation}

To conclude the proof, we now need to upper bound the right-hand side~\eqref{eq:equalitydiv}. To do so, we first note that, as per Theorem 3.2.1 in~\cite{Kemeny1960FiniteMC}, we can use the following series decomposition,
\begin{equation}
    \label{eq:Q-series}
(\vec I_{n-f} - (1 - \rho)\vec Q)^{-1} = \sum_{t = 0}^{\infty} (1 - \rho)^t \vec Q^t.
\end{equation}
This means that we can reduce the computation of $\divinfty{W^{(v)}}{W^{(u)}}$ to analyzing the powers of the matrix $\vec Q^t$. Furthermore, for large values of $t$, we can approximate $\vec Q^t$ by a one-rank matrix using the first eigenvalue and the first eigenvector of $\vec Q$ (see \lemmaref{lemma:matrix-powers-first-eigenvector}). This motivates us to study the spectral properties of $\vec Q$. We begin by showing (see \lemmaref{lemma:spectral-Q}) that $\vec Q$ is dominated by its first eigenvalue. To further estimate the coordinates of the first eigenvector of $\vec Q$, we need to introduce subsidiary matrices $\overline{\vec Q}$ and $\underline{\vec Q}$ (see \lemmaref{lemma:spectral-Q-hat}). We carefully design these matrices to have an explicit first eigenvector and so that their entries bound from above and below respectively those of $\vec Q$. Using these two properties, we obtain a measure of how far the first eigenvector of $\vec Q$ is from the uniform vector $\vec 1_{n-f} /\sqrt{n-f}$ (see \lemmaref{lemma:Q-delocalization}). By controlling spectral properties of $\vec Q$, we establish efficient one-rank approximations of high powers of $\vec Q$. Applying this to~\eqref{eq:equalitydiv}, we obtain an upper bound on the max divergence between $W^{(v)}$ and $W^{(u)}$, for any $u , v \in \honest$. Specifically, assuming that the adversarial density $\alpha_F < 1-\lambda$, we get
\begin{equation*}
    \divinfty{W^{(v)}}{W^{(u)}} \le \ln(\rho(n-f) + f) - 2\Tilde{T}\ln(1 - \alpha_F) - \Tilde{T}\ln(1 - \rho) - \ln(1-\lambda) + \ln(24) ,
\end{equation*}
where $\Tilde{T}  = \left\lceil\log_{\frac{\lambda}{1 - \alpha_F}}\left(\frac{1 - \alpha_F}{4(n-f)}\right)\right\rceil \left(\log_{\frac{\lambda}{1 - \alpha_F}}(1 - \alpha_F) + 2\right) + 2$. Finally, substituting~\eqref{eq:inequalitydiv} in the above, and upper bounding $\alpha_F$ as per Section~\ref{sec:adversarialdensity} we get the expected result. 

%% file: tikz/ProbCobraWalk.tex
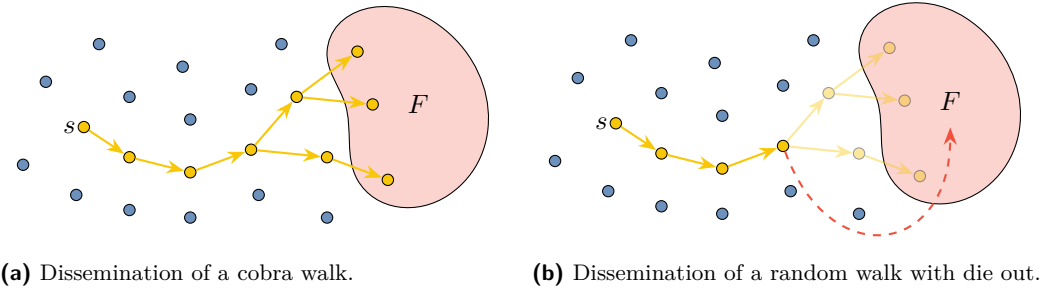
\begin{figure}[!ht]
    \centering
    \begin{subfigure}[b]{.5\textwidth}
    \centering
    \begin{tikzpicture}
        \begin{scope}[scale=2.0]
            \path[draw,use Hobby shortcut,closed=true, fill=curious, fill opacity=0.25]
            (2.4,-0.45) .. (2.3,0.7) .. (1.6,0.5) .. (1.72,0.15) .. (1.8,-0.35);
            \node[circle] at (2.2, 0.15) {$F$};
        \end{scope}
    
        \begin{scope}[every node/.style={circle, draw, inner sep=0pt,text width=1.5mm,fill=cobra}, scale=2.0]
            %%% cobra nodes
            \node[label=left:{$s$}] (A) at (0,0) {};            
            \node (B) at (0.3,-0.2) {};
            \node (C) at (0.7,-0.3) {};
            \node (D) at (1.1,-0.15) {};
            \node (E) at (1.4,0.2) {};
            \node (F) at (1.6,-0.2) {};
            \node (G) at (1.8,0.5) {};
            \node (H) at (1.9,0.15) {};
            \node (I) at (2.0,-0.35) {};
        \end{scope}
        \begin{scope}[every node/.style={circle, draw, inner sep=0pt,text width=1.5mm,fill=node}, scale=2.0]
            %%% regular nodes
            \node (1) at (0.3,0.2) {};
            \node (2) at (0.7,0.05) {};
            \node (3) at (-0.05,-0.45) {};
            \node (4) at (0.1,0.55) {};
            \node (5) at (0.65,0.4) {};
            \node (6) at (0.3,-0.55) {};
            \node (7) at (0.7,-0.6) {};
            \node (8) at (1.15,-0.45) {};
            \node (9) at (1.1,0.25) {};
            \node (10) at (1.3,0.55) {};
            \node (11) at (1.6,-0.6) {};
            \node (12) at (-0.25,0.3) {};
            \node (13) at (-0.4,-0.25) {};
        \end{scope}
        
        \begin{scope}[>={Stealth[cobra]},
                      every edge/.style={draw=cobra, thick}]
            %%% cobra edges
            \path[->] (A) edge (B);
            
            \path[->] (B) edge (C);
            \path[->] (C) edge (D);
            \path[->] (D) edge (E);
            \path[->] (D) edge (F);
    
            \path[->] (E) edge (G);
            \path[->] (E) edge (H);
            \path[->] (F) edge (I);
    
        \end{scope}
    \end{tikzpicture}
    \vspace{0.30cm}

    \caption{Dissemination of a cobra walk.}
    \label{fig:test1}
    \end{subfigure}%
    \begin{subfigure}[b]{.5\textwidth}
    \centering
    \begin{tikzpicture}
        \begin{scope}[scale=2.0]
            \path[draw,use Hobby shortcut,closed=true, fill=curious, fill opacity=0.25]
            (2.4,-0.45) .. (2.3,0.7) .. (1.6,0.5) .. (1.72,0.15) .. (1.8,-0.35);
            \node[circle] (curious) at (2.2, 0.15) {$F$};
        \end{scope}
    
        \begin{scope}[every node/.style={circle, draw, inner sep=0pt,text width=1.5mm,fill=cobra}, scale=2.0]
            %%% cobra nodes
            \node[label=left:{$s$}] (A) at (0,0) {};            
            \node (B) at (0.3,-0.2) {};
            \node (C) at (0.7,-0.3) {};
            \node (D) at (1.1,-0.15) {};
        \end{scope}
        \begin{scope}[every node/.style={circle, draw, inner sep=0pt,text width=1.5mm,fill=node}, scale=2.0]
            %%% regular nodes
            \node (1) at (0.3,0.2) {};
            \node (2) at (0.7,0.05) {};
            \node (3) at (-0.05,-0.45) {};
            \node (4) at (0.1,0.55) {};
            \node (5) at (0.65,0.4) {};
            \node (6) at (0.3,-0.55) {};
            \node (7) at (0.7,-0.6) {};
            \node (8) at (1.15,-0.45) {};
            \node (9) at (1.1,0.25) {};
            \node (10) at (1.3,0.55) {};
            \node (11) at (1.6,-0.6) {};
            \node (12) at (-0.25,0.3) {};
            \node (13) at (-0.4,-0.25) {};
        \end{scope}
        \begin{scope}[>={Stealth[cobra]},
                      every edge/.style={draw=cobra, thick}]
            %%% cobra edges
            \path[->] (A) edge (B);
            
            \path[->] (B) edge (C);
            \path[->] (C) edge (D);
        \end{scope}
        \begin{scope}[>={Stealth[curious]},
                      every edge/.style={draw=curious, thick}]
            %% curious arrow 
            \path[->,dashed,thick] (D) edge[in=270,out=290, looseness=2] node [left] {} (curious);
        \end{scope}
        
        \begin{scope}[transparency group, opacity=0.4]
            \begin{scope}[every node/.style={circle, draw, inner sep=0pt,text width=1.5mm,fill=cobra}, scale=2.0]
                %%% hidden nodes
                \node (E) at (1.4,0.2) {};
                \node (F) at (1.6,-0.2) {};
                \node (G) at (1.8,0.5) {};
                \node (H) at (1.9,0.15) {};
                \node (I) at (2.0,-0.35) {};
            \end{scope}
            
            \begin{scope}[>={Stealth[cobra]},
                      every edge/.style={draw=cobra, thick}]
                %%% hidden edges
                \path[->] (D) edge (E);
                \path[->] (D) edge (F);
                \path[->] (E) edge (G);
                \path[->] (E) edge (H);
                \path[->] (F) edge (I);
            \end{scope}
        \end{scope}

    \end{tikzpicture}
    \vspace{-0.37cm}
    
    \caption{Dissemination of a random walk with die out.}
    \label{fig:test2}
    \end{subfigure}

    \caption{Illustration of the reduction from a cobra walk (Fig.~\ref{fig:test1}) to a random walk with probabilistic die out (Fig.~\ref{fig:test2}). In Fig.~\ref{fig:test1}, the dissemination continues after the walk branches and hits the curious set $F$ in several places. In the random walk with die out, instead of letting the dissemination branch, we stop the dissemination as soon as the cobra walk branches and report the position of the branching node.}
    \label{fig:prob-die-out}
\end{figure}

%% file: Sections/06Analysis.tex
Note that when the gossip protocol parameter $\rho$ decreases, the privacy guarantees in \theoremref{thm:main} improve. Yet, this worsens the dissemination time, which suggests the existence of a \emph{trade-off} between the dissemination time and the source anonymity of the protocol. In this section, we formalize this observation by showing the tightness of \theoremref{thm:main} on a family of strong expanders called \emph{near-Ramanujan graphs}. Intuitively, for dense enough graph topologies, most terms in \theoremref{thm:main} vanish, hence considerably simplifying the analysis of the result. Near-Ramanujan graphs can be defined as follows.

\begin{definition}[Near-Ramanujan family of graphs]
\label{def:nr-graphs}
Let $\mathcal G$ be an infinite family of regular graphs. $\mathcal G$ is called near-Ramanujan if there exists a constant $c > 0$ such that $\lambda(G) \le c d(G)^{-1/2}$ for any graph $G \in \mathcal G$ of large enough size. 
\end{definition}

This choice of graph family is motivated by the fact that near-Ramanujan graphs naturally arise in the study of dense random regular graphs. In fact, for any large enough $n$ and any $3 \le d \le n/2$ (with $dn$ even) a random $d$-regular graph on $n$ nodes is near-Ramanujan with high probability as shown in~\cite{Cook2018SizeBC,Tikhomirov2019TheSG}. That means that almost every $d$-regular graph is near-Ramanujan. 
Besides using near-Ramanujan graphs, we assume the topologies to be dense enough, i.e., $d \in n^{\Omega_n(1)}$. Refining the statement of \theoremref{thm:main} to this family of graphs, we obtain the following corollary.

 \begin{restatable}{corollary}{PrivacyNRWorst}
    \label{corollary:privacy-nr-worst}
    \label{corollary:privacy-nr-avg}
    \label{corollary:privacy-nr}
    Let $\mathcal P$ be a $(1+\rho)$-cobra walk and let $\mathcal G$ be a family of $d$-regular near-Ramanujan graphs with $n$ nodes and $d \in n^{\Omega_n(1)}$. Suppose $f/d \in 1 - \Omega_n(1)$ (resp. $f/n \in 1 - \Omega_n(1)$). Then, for any $G \in \mathcal G$ of large enough size $n$ and any $\rho \in 1 - \Omega_n(1)$, $\mathcal P$ satisfies $\varepsilon$-DP against a worst-case adversary (resp. an average-case adversary) for some 
    \[
    \varepsilon \in \ln\left(\rho(n-f) + f\right) + O_n(1).
    \]
    The above statement also holds if $\mathcal P$ is a $\rho$-Dandelion protocol with $\rho < 1$.
\end{restatable}

From~\corollaryref{corollary:privacy-nr-worst}, when $\rho=0$, we obtain a level of differential privacy that matches, up to an additive constant, the universal lower bound $\varepsilon\geq \ln(f-1)$. Accordingly, $\rho=0$ leads to an \emph{optimal} differential privacy guarantee. However, in this case, both the cobra walk and the Dandelion protocol degenerate into simple random walks with dissemination time in $\Omega_n(n\log(n))$~\cite{Aldous1989LowerBF}. Increasing $\rho$ parameter makes the dissemination faster, but potentially worsens the privacy guarantees. 

Studying Dandelion and cobra walks, we show that the result in \corollaryref{corollary:privacy-nr} is tight up to an additive constant. Then, we formally validate our intuition that decreasing $\rho$ increases the dissemination time by providing corresponding tight guarantees on dissemination time.
Finally, to put our results in perspective, we compare them to a random walk (optimal privacy but high dissemination time), and to a $2$-cobra walk (optimal dissemination time with bad, completely vacuous, privacy guarantees). We summarize our findings for both worst-case and average-case adversaries in the table below and defer the detailed analysis to \appendixref{sec:appendix-tradeoff}.

\begin{table}[!ht]
    \centering
\label{tab:tradeoff}
\begin{tabular}{@{}cccc@{}} 
    \toprule
     Protocol & Privacy ($\varepsilon$) & Dissemination time & References \\
 \midrule
\multirow{2}*{Random walk} &   \multirow{2}*{$\ln(f) + \Theta_n\left(1\right)$} & \multirow{2}*{$\Theta_n\left(n\log{(n)}\right)$} & \corollaryref{corollary:privacy-nr},\\
& & & \theoremref{thm:universal-impossibility},~\cite{Aldous1989LowerBF}\\
 \midrule 
 \multirow{2}*{$\rho$-Dandelion} & \multirow{2}*{$\ln\left(\rho (n - f)+ f\right) + \Theta_n\left(1\right)$} & \multirow{2}*{$\Theta_n\left(\frac{1}{\rho} + D\right)$} & \corollaryref{corollary:privacy-nr}, \\
 & & & Theorem~\ref{thm:dandelion-privacy-lower} and~\ref{thm:dandelion-speed}\\
 \midrule
 \multirow{2}*{$(1+\rho)$-Cobra walk} &  \multirow{2}*{$ \ln\left(\rho (n - f)+ f\right) + \Theta_n\left(1\right)$} & \multirow{2}*{$O_n\left(\frac{\log{(n)}}{\rho^3}\right)$, $\Omega_n\left(\frac{\log{(n)}}{\rho}\right)$} & \corollaryref{corollary:privacy-nr},\\
 & & & Theorem~\ref{thm:cobra-privacy-lower} and~\ref{thm:cobra-speed}\\
 \midrule
 $2$-Cobra walk  &   $ \ln(n) + \Omega_n\left(1\right)$ & $\Theta_n\left(\log{(n)}\right)$ & Theorem~\ref{thm:cobra-privacy-lower} and~\ref{thm:cobra-speed}\\
 \bottomrule
\end{tabular}
\vspace{5pt}
    \caption{ Summary of the tension between differential privacy of a $(1+\rho)$-cobra walk and Dandelion gossip and their dissemination time on dense near-Ramanujan graphs. Graphs have diameter $D$ and consist of $n$ nodes, $f$ of which are curious. Note that the upper bounds on $\varepsilon$ hold under assumptions in \corollaryref{corollary:privacy-nr}. Lower bounds on $\varepsilon$ hold assuming $f/n \in 1 - \Omega_n(1)$, and for cobra walk we also assume $f \in n^{\Omega_n(1)}$. Dissemination time bounds for cobra walk and Dandelion hold for $\rho \in \omega_n\left(\sqrt{\log(n)/n}\right)$ and $\rho \in \Omega_n(1/n)$ respectively.
    }
\end{table}

%% file: Sections/07Conclusion.tex
This paper presents an important step towards quantifying the inherent level of source anonymity that gossip protocols provide on general graphs.
We formulate our results through the lens of differential privacy. First, we present a universal lower bound on the level of differential privacy an arbitrary gossip protocol can satisfy. Then, we devise an in-depth analysis of the privacy guarantees of $(1+\rho)$-cobra walk and $\rho$-Dandelion protocols on expander graphs. When $\rho = 0$, the protocols spread the gossip via a random walk, which achieves optimal privacy, but has poor dissemination time. On the other hand, we show that increasing $\rho$ improves the dissemination time while the privacy deteriorates. In short, our tight analysis allows to formally establish the trade-off between dissemination time and the level of source anonymity these protocols provide. 
An interesting open research question would be to establish whether this ``privacy vs dissemination time'' trade-off is fundamental or if there exists a class of gossip protocols that could circumvent this trade-off. 

We consider differential privacy, because, unlike other weaker notions of privacy (e.g., MLE-based bounds), it can be applied against an \emph{arbitrary} strategy of the adversary, factoring in \emph{any} prior beliefs an adversary may have about the location of the source and the nature of the gossip protocol. This makes differential privacy strong and resilient. However, differential privacy is often criticized for being too stringent in some settings. Consequently, a number of possible interesting relaxations have been proposed in the literature such as Pufferfish~\cite{kifer2014pufferfish} and Renyi differential privacy~\cite{mironov2017renyi}. Adapting our analysis to these definitions constitutes an interesting open direction as it would enable consideration of less stringent graphs structures and probability metrics.

Finally, we believe that our results could be applied to solve privacy related problems in other settings. For example, it was recently observed in~\cite{Cyffers2020PrivacyAB} that sharing sensitive information via a randomized gossip can amplify the privacy guarantees of some learning algorithms, in the context of privacy-preserving decentralized machine learning. However, this work only considers the cases when the communication topology is a clique or a ring. We believe that the techniques we develop in this paper can be useful to amplify privacy of decentralized machine learning on general topologies. This constitutes an interesting open problem.

%% file: Sections/08Appendix.tex
\section{Additional preliminaries \& notation}
\label{appendix:preliminaries}

\subsection{Notations}
\label{appendix:preliminaries-notations}

In the remaining, we consider cases where the sample space $\Omega$ is a \emph{discrete ordered} set. Then, we often characterize a probability measure $\mu$ on $\Omega$ by a vector $\boldsymbol \mu$ in $\mathbb R^{|\Omega|}$ where for any $i \in \Omega$ we have $\boldsymbol \mu_i = \mu(\{i\})$. Furthermore, for any two real-valued random variables $X$ and $Y$, we say that $X$ stochastically dominates $Y$ if and only if for every $x\in\mathbb R$ we have
\[
\Pr[X\ge x]\ge \Pr[Y \ge x].
\]
Additionally, for an assertion $E$, we denote by $\ind_E$ the indicator function that takes value $1$ if $E$ holds true and $0$ otherwise. Finally, let $G=(V,E)$ be an undirected graph. For $U \subseteq V$, we denote by $\partial U \subseteq E$ the set of edges connecting $U$ and $V\setminus U$. Also, for $U\subseteq V$, we denote by $N(U) = \bigcup_{u\in U} N(u)$ the neighbourhood of $U$ in $G$.

\subsection{Markov chains}

\label{sec:preliminaries-markov-chains}

For a Markov chain with finite ordered state space $\Lambda$, we will denote by $\vec P \in \mathbb{R}^{|\Lambda| \times |\Lambda|}$ its transition matrix. % of a Markov chain. 
Specifically, for any $i,j \in \Lambda$, $\vec P_{ij}$ denotes the probability of transition from state $j$ to state $i$ in one step and thus $\vec P$ is a \emph{left stochastic matrix}.
If $\vec \mu$ is a vector corresponding to the distribution of a Markov chain at time $t-1$, then the distribution at time $t$ corresponds to $\vec P \vec \mu$. In particular, we will be interested in two types of Markov chains: random walks on a fixed graph and absorbing Markov chains. A random walk on a graph $G=(V,E)$ is a Markov chain with state space $\Lambda=V$ and the transition matrix defined by $P_{uv} = \frac{1}{\deg(v)}$ if $\{v,u\}\in E$ and $0$ otherwise. 
A Markov chain is called absorbing if there exists a state $s \in \Lambda$ such that $\vec P_{ss}=1$. Such a state is called absorbing; conversely, a non-absorbing state is called transient. 

After possibly permuting the set of states $\Lambda$, the transition matrix $\vec P$ of an absorbing Markov chain with $m$ transient states and $m'$ absorbing states can always be written in a block form as follows
\begin{equation}
\label{eq:block-form-absorbing-markov-generic}
\vec P = 
\begin{bmatrix}
\vec Q & \vec O_{m \times m'} \\
\vec R & \vec I_{m'} 
\end{bmatrix},
\end{equation}
where $\vec Q$ is a matrix characterizing transition between transient states, and $\vec R$ characterizes transitions between transient and absorbing states. Using this characterization, the following holds true. 

\begin{restatable}
[\cite{Kemeny1960FiniteMC}, Theorem 3.3.7]{lemma}{AbsorbtionProb}
\label{lemma:absorbtion-prob}
Let $\vec P$ be a transition matrix of an absorbing Markov chain with block form as in \equationref{eq:block-form-absorbing-markov-generic}. Then, $\vec I_{m} - \vec Q$ is invertible, and series $\sum_{i = 0}^\infty \vec R\vec Q^i$ converges. Moreover, 
\[
\vec M = \sum_{t = 0}^\infty \vec R\vec Q^t = \vec R(\vec I_{m} - \vec Q)^{-1},
\]
where $M_{ji}$ characterizes the probability of a Markov chain being absorbed in an absorbing state $j$ starting from a transient state $i$.
\end{restatable}

\subsection{Useful spectral lemmas}

In this subsection, we restate some results that will be useful in the remaining proofs. First, recall the definition of the operator norm.
\begin{definition}
    Consider a real matrix $\vec M \in \mathbb R^{p\times n}$. Then, we define the operator norm of $\vec M$ as follows
    \[
    \norm{\vec M}_{op} \triangleq \sup_{\vec x \in \mathbb R^{n} : \norm{\vec x}_2 = 1} \norm{\vec M \vec x}_2.
    \]
\end{definition}

If for two real square symmetric matrices of the same size $A$ and $B$, if the operator norm of $A-B$ is small, then their spectra are close in $\ell_\infty$ norm. More precisely, the following holds true.
\begin{lemma}[Eigenvalue stability inequality {{\cite[Equation (1.64)]{tao2012topics}}}]
\label{lemma:eigenvalue-stability}
    Consider two real symmetric matrices $\vec A, \vec B \in \mathbb R^{p\times p}$. Then, for any $1 \le i \le p$
    \[
    \left|\lambda_i(\vec A) - \lambda_i(\vec B)\right| \le \norm{\vec A - \vec B}_{op}.
    \]
\end{lemma}

We also restate Cauchy interlacing law.
\begin{lemma}[Cauchy interlacing law {{\cite[Equation (1.75)]{tao2012topics}}}]
\label{lemma:cauchy-interlace}
    Consider a real symmetric matrix $\vec M \in \mathbb R^{p\times p}$ and let $\Hat{\vec M}$ be its principal submatrix of size $n\times n$ where $n < p$. Then, for any $1 \le i \le n$
    \[
    \lambda_i(\vec M) \ge \lambda_i(\Hat{\vec{M}}) \ge \lambda_{i + p - n}(\vec M).
    \]
\end{lemma}

\subsection{Partial order for matrices}

We now introduce a partial ordering on the set of real square matrices. For two square matrices $\vec A$, $\vec B$ of the same size, we write $\vec A \succeq \vec B$ to indicate that $A_{ij} \ge B_{ij}$ for any $i,j$. With this notation, we can present the following lemmas.
\begin{lemma}
\label{lemma:partial-order-mult}
Consider three matrices $\vec A,\vec B, \vec C \in \mathbb R^{p \times p}$. Suppose that $\vec A \succeq \vec B \succeq \vec O_{m \times m}$ and $\vec C \succeq \vec O_{m\times m}$. Then, $\vec A \vec C \succeq \vec B\vec C \text{\, and \, } \vec C \vec A \succeq \vec C \vec B.$
\end{lemma}
\begin{proof}
Let $1 \le i \le p$ and $1 \le j \le p$. Then
\[
(\vec A\vec C)_{ij} = \sum_{k = 1}^{p} A_{ik}C_{kj}.
\]
Similarly
\[
(\vec B\vec C)_{ij} = \sum_{k = 1}^{p} B_{ik}C_{kj}.
\]
Since $C_{kj} \ge 0$ for all $k,j \in [p]$ and $A\succeq B$, we have $A_{ik}C_{kj} \ge B_{ik}C_{kj}$ for any $i, j, k\in [p]$. Hence,
\[
(\vec A\vec C)_{ij} = \sum_{k = 1}^{p} A_{ik}C_{kj} \ge \sum_{k = 1}^{p} B_{ik}C_{kj} = (\vec B\vec C)_{ij}.
\]
As this holds true  for any $1 \le i \le m$ and $1 \le j \le m$, we can write
\[
\vec A \vec C \succeq \vec B\vec C.
\]
Following the same steps, we also get $\vec C \vec A \succeq \vec C \vec B$.
\end{proof}

\begin{lemma}
\label{lemma:partial-order-power}
Consider $\vec A, \vec B \in \mathbb R^{p \times p}$ such that $\vec A \succeq \vec B \succeq \vec O_{p \times p}$. Then, for any $t \ge 0$,  $\vec A^{t} \succeq \vec B^{t}$.
\end{lemma}
\begin{proof}
The proof follows from \lemmaref{lemma:partial-order-mult} by induction. As a base for the induction, we will use $\vec A^1 \succeq \vec B^1$. Now, for the induction step, suppose $\vec A^{t-1} \succeq \vec B^{t-1}$ for some $t \ge 2$. Then
\begin{align*}
        \vec A^t &= \vec A \cdot \vec A^{t-1}.
        \intertext{By the induction hypothesis and \lemmaref{lemma:partial-order-mult}, we have}
        \vec A^t &\succeq \vec A \cdot \vec B^{t-1}.
        \intertext{From $\vec A\succeq \vec B$ and \lemmaref{lemma:partial-order-mult}, we get}
        \vec A^t &\succeq \vec B \cdot \vec B^{t-1} \\
        &= \vec B^t,
\end{align*}
which completes the induction step.
\end{proof}

We also present below the Perron-Frobenius theorem about the first eigenvectors and the first eigenvalues of matrices with non-negative entries.
\begin{restatable}[Perron-Frobenius theorem {{\cite{meyer2000matrix}}}]{lemma}{PerronFrobenius}
\label{lemma:perron}
    Consider $\vec A \in \mathbb R^{p \times p}$ such that $\vec A \succeq \vec O_{p \times p}$. Then the following assertions hold true.
    \begin{enumerate}[(a)]
    \item $\lambda_1(\vec A) \ge 0$ and for every $2 \le i \le p$ we have $\left|\lambda_i(\vec A)\right| \le \lambda_1(\vec A)$.
    \item There exists $\vec \eigen_1 \in \mathbb R^{p}$ such that $A\vec\eigen_1 = \lambda_1(A) \vec\eigen_1$ and all coordinates of $\vec\eigen_1$ are non-negative.
    \end{enumerate}
\end{restatable} 

In the remaining, we will often approximate high powers of a matrix using its first eigenvalue and an outer product of its first eigenvector with itself. Formally, the following lemma holds true.
\begin{restatable}{lemma}{MatrixPowersFirstEigenvector}
\label{lemma:matrix-powers-first-eigenvector} 
Consider a real symmetric matrix $\vec M \in \mathbb R^{p\times p}$. Let $\lambda_1 > \lambda_2 > \ldots > \lambda_p$ be eigenvalues of $\vec M$. Let $\vec \eigen_1$ be $\ell_2$-normalized first eigenvector of $\vec M$ and let $\lambda = \max\left\{\left|\lambda_2\right|, \left|\lambda_p\right|\right\}$. Then,for any $t \ge 0$
\[
- \lambda^t \vec J_{p}  \preceq \vec M^t - \lambda_1^t \vec \eigen_1 \vec \eigen_1^\top \preceq \lambda^t \vec J_{p}.
\]
\end{restatable}
\begin{proof}
Let $\lambda_i$ be the $i^\text{th}$ eigenvalue of $\vec M$ and let $\vec \eigen_i$ be $i^\text{th}$ eigenvector of $\vec M$ normalized with respect to $\ell_2$-norm. Since $\vec M$ is real and symmetric, we can write $\vec M^t$ as
\[
\vec M^t = \sum_{k = 1}^p\lambda_k^t \vec \eigen_k\vec \eigen_k^\top.
\]
Then
\[
\vec M^t - \lambda_1^t\vec \eigen_1 \vec \eigen_1^\top = \sum_{k = 2}^p\lambda_k^t \vec \eigen_k\vec \eigen_k^\top.
\]
In other words, $\vec M^t - \lambda_1^t\vec \eigen_1 \vec \eigen_1^\top$ has eigenvectors $\vec \eigen_1, \vec \eigen_2, \ldots, \vec \eigen_p $ with corresponding eigenvalues $0, \lambda_2^t,\ldots, \lambda_p^t$. Then,$\norm{\vec M^t - \lambda_1^t\vec \eigen_1 \vec \eigen_1^\top}_{op} = 
\max\left\{|\lambda^t_2|, |\lambda^t_p|\right\}= \lambda^t$. Let $\vec e_i$ be $i^\text{th}$ coordinate unit vector. For any $1 \le i,j \le p$, by Cauchy-Shwarz inequality
\begin{align*}
\left|\left\langle \vec e_i, \left(\vec M^t - \lambda_1^t\vec \eigen_1 \vec \eigen_1^\top\right)\vec e_j\right\rangle\right| 
&\le \norm{\vec e_i}_2 \norm{\left(\vec M^t - \lambda_1^t\vec \eigen_1 \vec \eigen_1^\top \right)\vec e_j}_2.
\intertext{By definition of operator norm, since $\norm{\vec e_j}_2 = 1$, we have}
\left|\left\langle \vec e_i, \left(\vec M^t - \lambda_1^t\vec \eigen_1 \vec \eigen_1^\top\right)\vec e_j\right\rangle\right| 
&\le \norm{\vec e_i}_2 \norm{\vec M^t - \lambda_1^t\vec \eigen_1 \vec \eigen_1^\top}_{op}\\
&\le \lambda^t.
\end{align*}
Accordingly, for any $1 \le i,j \le p$
\[
\left|\left(\vec M^t - \lambda_1^t\vec \eigen_1 \vec \eigen_1^\top\right)_{ij}\right| \le \lambda^t.
\]
Hence,
\[
- \lambda^t \vec J_{p}  \preceq \vec M^t - \lambda_1^t \vec \eigen_1 \vec \eigen_1^\top \preceq \lambda^t \vec J_{p},
\]
which concludes the proof.
\end{proof}

\section{Fundamental limits of source anonymity: proofs of ~\sectionref{sec:impossibility}}

\label{appendix:impossibility}

\UniversalImprossibility*
\begin{proof}
We separate the proof in two parts: the first one for the average-case adversary and then the second one for the worst-case adversary.

\paragraph*{(i) Average-case adversary.} Here, we aim to show that there exists no gossip protocol that can satisfy $\varepsilon$-DP against an average-case adversary with $\varepsilon < \ln(f-1)$. To do so, let us consider an arbitrary gossip protocol $\mathcal P$. By definition of differential privacy against an average-case adversary (\sectionref{sec:source-anonymity-with-DP}), to obtain the desired result it suffices to show that for any $\varepsilon< \ln(f-1)$, the following holds true 
\[\Pr_{F \sim \uniformsubset{f}{V}} \left[\max_{v,u\in \honest} \divinfty{\seqadv{v}}{\seqadv{u}} \le \varepsilon \right] < 1 - \frac{1}{n}, \]
where $\seqadv{v}$ and $\seqadv{u}$ correspond to the executions of $\mathcal{P}$.

\paragraph*{(i.1) Preliminary computations.} Consider an arbitrary bijection $\chi \colon V \to [n]$. Let us consider an arbitrary set of curious nodes $F \subset V$ of size $f$. Recall that $\tadv$ is the first round in which one of the curious nodes received the gossip. Then, for any $v \in \honest$, we denote by $\sigma_{v}$ the set of all adversarial observations, where in round $\tadv$ the adversary observes a communication of the form $(v\to \ast)$ and all other observed communications are of the form $(w \to \ast)$ with $\chi(w) \ge \chi(v)$ (i.e., $v$ has a minimal index assigned by $\chi$). Note that, since no curious node can be active in the beginning of round $\tadv$, for any $u \in \honest$, the events $\{\seqadv{u} \in \sigma_v\}_{v \in \honest}$ represent the  partition of the probability space. Hence, for any $u\in \honest$, we have
\[
\sum_{v \in \honest} \Pr\left[\seqadv{u} \in \sigma_v\right]= 1,
\]
and
\begin{equation}
    \label{eq:impossibility-sum-vu}
    \sum_{v \in \honest} \sum_{u \in \honest} \Pr\left[\seqadv{u} \in \sigma_v\right] = n-f.
\end{equation}

Note that we allow a communication of a node to itself (e.g., $(v\to v)$ for an active node $v$). Hence, it is possible that the source communicates to itself for several rounds before it communicates to other nodes. For $v\in V$ and $u \in N(v)$, let us denote by $\Xi(v,u)$ the probability that when $v$ is the source, node $u$ is the first node $v$ communicates to, apart from itself. If $u$ also belongs to $F$, then we have $\seqadv{v} \in \sigma_v$, since in round $\tadv$ only node $v$ is active, and hence all communications in round $\tadv$ have a form of $(v\to\ast)$. Then, for any $v \in \honest$
\[
\Pr\left[\seqadv{v} \in \sigma_v\right] \ge \sum_{u \in N(v)} \ind_{\{u \in F\}} \Xi(v,u).
\]
Hence,
\begin{align*}
\sum_{v\in \honest} \Pr\left[\seqadv{v} \in \sigma_v\right] &\ge \sum_{v\in \honest} \sum_{u\in N(v)} \ind_{\{u\in F\}}\Xi(v,u) \\
                                                            &= \sum_{v \in V} \sum_{u\in N(v)} \ind_{\{v\in \honest\}}\ind_{\{u\in F\}}\Xi(v,u).
\end{align*}
Set $\xi(F) =  \sum_{v \in V} \sum_{u\in N(v)} \ind_{\{v\in \honest\}}\ind_{\{u\in F\}}\Xi(v,u)$. Then, from the above, we have
\begin{equation}
    \label{eq:impossibility-sum-vv}
\sum_{v \in \honest} \Pr\left[\seqadv{v} \in \sigma_v\right] \ge \xi(F).
\end{equation}
Note that, from \equationref{eq:impossibility-sum-vu} and \equationref{eq:impossibility-sum-vv}, we have
\begin{equation}
\label{eq:ratio-sums-impossibility}
\sum_{v \in \honest} \Pr\left[\seqadv{v} \in \sigma_v\right] \ge \frac{\xi(F)}{n-f} \sum_{v \in \honest} \sum_{u \in \honest} \Pr\left[\seqadv{u} \in \sigma_v\right].
\end{equation}
Note that for any two non-negative sequences $(a_i)_{i \in [K]} $ and $(b_i)_{i \in [K]}$, if $\sum_{i \in [K]} b_i > 0$ and for some $\ell \ge 0$ we have $\sum_{i \in [K]} a_i \ge \ell \sum_{i \in [K]} b_i$, then there exists $i_\star \in [K]$ such that $b_{i_\star} > 0$ and $a_{i_\star}/b_{i_\star} \ge \ell$. Then, from (\ref{eq:ratio-sums-impossibility}), there exists $v_\star \in \honest$ such that 
\begin{equation}
\label{eq:pr-vstar}
\frac{\Pr\left[\seqadv{v_\star} \in \sigma_{v_\star}\right]}{\sum_{u \in \honest} \Pr\left[\seqadv{u} \in \sigma_{v_\star}\right]}\ge \frac{\xi(F)}{n-f}.
\end{equation}
Let us take $u_\star = \argmin_{u \in V\setminus (F\cup \{v_\star\})}  \Pr\left[\seqadv{u} \in \sigma_{v_\star}\right]$. Then,
\begin{equation}
\label{eq:pr-ustar}
\Pr\left[\seqadv{u_\star} \in \sigma_{v_\star}\right] \le \frac{1}{n-f-1}\sum_{u \in V\setminus (F\cup\{v_\star\})} \Pr\left[\seqadv{u} \in \sigma_{v_\star}\right].
\end{equation}
Combining \equationref{eq:pr-vstar} and \equationref{eq:pr-ustar}, we get
\[
\frac{\Pr\left[\seqadv{v_\star} \in \sigma_{v_\star}\right]}{\Pr\left[\seqadv{v_\star} \in \sigma_{v_\star}\right] + (n-f-1)\Pr\left[\seqadv{u_\star} \in \sigma_{v_\star}\right]} \ge \frac{\xi(F)}{n-f}.
\]
Rearranging the terms gives us, by definition of $D_\infty$
\[
\exp\left(\divinfty{\seqadv{v_\star}}{\seqadv{u_\star}}\right) \ge \frac{\Pr\left[\seqadv{v_\star} \in \sigma_{v_\star}\right]}{\Pr\left[\seqadv{u_\star} \in \sigma_{v_\star}\right]} \ge \frac{(n-f-1) \frac{\xi(F)}{n-f}}{1 - \frac{\xi(F)}{n-f}} = \frac{(n-f-1)\xi(F)}{n - f - \xi(F)}.
\]
The existence of such pair $v_\star,u_\star$ implies that for any $F \subseteq V$ of size $f$ we have
\begin{equation}
\label{eq:impossibility-div-lb}
\max_{v,u\in \honest} \exp\left(\divinfty{\seqadv{v}}{\seqadv{u}}\right) \ge \exp\left(\divinfty{\seqadv{v_\star}}{\seqadv{u_\star}}\right) \ge \frac{(n-f-1)\xi(F)}{n - f - \xi(F)}.
\end{equation}

\paragraph*{(i.2) Computing the expectation $\E_{F\sim \uniformsubset{f}{V}} \left[\xi(F)\right]$.} Note that, until now, we chose $F$ arbitrarily, hence the above holds true for any $F$. We can thus use this property to compute the expectation of $\xi(F)$ when $F$ is chosen uniformly at random, i.e., $F \sim \uniformsubset{f}{V}$. Specifically, we have
\begin{align*}
\E_{F\sim \uniformsubset{f}{V}} \left[\xi(F)\right] &= \sum_{v \in V} \sum_{u\in N(v)} \Xi(v,u) \E_{F\sim \uniformsubset{f}{V}}\left[\ind_{\{v\in \honest\}}\ind_{\{u\in F\}}\right]\\
&= \sum_{v \in V} \sum_{u\in N(v)} \Xi(v,u) \Pr_{F\sim \uniformsubset{f}{V}}\left[u \in F \land v \in \honest\right].
\intertext{Using the identity $\Pr[A \land B] = \Pr[A \mid B] \Pr[B]$ with  $\Pr_{F\sim \uniformsubset{f}{V}}\left[v \in \honest \mid u \in F\right] =(n-f)/(n-1)$ and $\Pr_{F\sim \uniformsubset{f}{V}}\left[u \in F\right] =f/n$ we get}
\E_{F\sim \uniformsubset{f}{V}} \left[\xi(F)\right]&= \sum_{v \in V} \sum_{u\in N(v)} \Xi(v,u) \frac{f(n-f)}{n(n-1)}\\
&= \frac{f(n-f)}{n(n-1)} \sum_{v \in V} \sum_{u\in N(v)} \Xi(v,u).
\intertext{Note that, by definition of $\Xi$, for every $v\in V$ we have $\sum_{u\in N(v)} \Xi(v,u) = 1$. Hence}
\E_{F\sim \uniformsubset{f}{V}} \left[\xi(F)\right]&= \frac{f(n-f)}{n(n-1)} \sum_{v \in V} 1\\
&= \frac{f(n-f)}{n-1}.
\end{align*}
Finally, we get
\begin{equation}
\label{eq:xif-expected}
\E_{F\sim \uniformsubset{f}{V}} \left[\xi(F)\right] = f\frac{n-f}{n-1}.
\end{equation}
\paragraph*{(i.3) Applying Markov inequality.}
Note also that $\xi(F)$ can be bounded from above as follows
\begin{align*}
\xi(F) &= \sum_{v \in V} \sum_{u\in N(v)} \ind_{\{v\in \honest\}}\ind_{\{u\in F\}}\Xi(v,u) \\
       &= \sum_{v \in \honest}  \sum_{u\in N(v)} \ind_{\{u\in F\}}\Xi(v,u) \\
       &\le \sum_{v \in \honest}  \sum_{u\in N(v)} \Xi(v,u) \\
       &= \sum_{v \in \honest} 1\\
       &= n-f.
\end{align*}
Then, we can define $\Tilde{\xi}(F) = n-f - \xi(F)$ as a non-negative random variable. From \equationref{eq:xif-expected}, we have
\[
\E_{F\sim\uniformsubset{f}{V}} \left[ \Tilde{\xi}(F) \right] = (n-f) - f\frac{n-f}{n-1} = \frac{(n-f)(n-f-1)}{n-1}.
\]
Then, by Markov inequality 
\begin{align*}
\Pr\left[\xi(F) \le \frac{(f-1)(n-f)}{n - 2}\right] &= \Pr\left[\Tilde{\xi}(F) \ge n-f - \frac{(f-1)(n-f)}{n - 2}\right]\\
                                                          &\le \frac{\E_{F\sim\uniformsubset{f}{V}} \left[\Tilde{\xi}(F)\right]}{n-f - \frac{(f-1)(n-f)}{n - 2}}\\ &= \frac{\frac{(n-f)(n-f-1)}{n-1}}{n-f - \frac{(f-1)(n-f)}{n - 2}}\\
                                                          &= \frac{(n-f-1)(n - 2)}{(n-1)(n - f -1)} \\
                                                          &= \frac{n-2}{n-1} \\
                                                          & = 1 - \frac{1}{n-1} \\
                                                          &< 1 - \frac{1}{n}.
\end{align*}
Then
\begin{equation}
\label{eq:xif-probability}
\Pr\left[\xi(F) \ge \frac{(f-1)(n-f)}{n - 2}\right] > \frac{1}{n}.
\end{equation}

\paragraph*{(i.4) Conclusion.} Note that the right-hand side of \equationref{eq:impossibility-div-lb} is an increasing function of $\xi(F)$. Hence combining \equationref{eq:impossibility-div-lb} and \equationref{eq:xif-probability}, we obtain that, with probability at least $1/n$, the following holds true
\begin{align*}
\max_{v,u\in \honest} \exp\left(\divinfty{\seqadv{v}}{\seqadv{u}}\right) &\ge \frac{(n-f-1) \frac{(f-1)(n-f)}{n - 2}}{n - f - \frac{(f-1)(n-f)}{n - 2}} \\
&= \frac{(n-f-1) (f-1)}{n -f-1} \\
&= f-1.
\end{align*}
In other words, we just showed that for any $\varepsilon < \ln(f-1)$ we have
\begin{equation}
\label{eq:avg-case-lower}
 \Pr_{F \sim \uniformsubset{f}{V}} \left[\max_{v,u\in \honest} \divinfty{\seqadv{v}}{\seqadv{u}} \le \varepsilon \right] < 1 - \frac{1}{n},
\end{equation}
which concludes the first part of the proof.

\paragraph*{(ii) Worst-case adversary.} We now turn our attention to the second part of the proof and consider the worst-case adversary. Here, we consider two subcases: when $f < \kappa(G)$ and when $f \geq \kappa(G)$.  

\paragraph*{(ii.1) When $f < \kappa(G)$.} We aim to show that there exists no gossip protocol that can satisfy $\varepsilon$-DP against a worst-case adversary with $\varepsilon < \ln(f-1)$. We show this statement by contradiction. Suppose there exists a protocol $\mathcal P$ that satisfies $\varepsilon$-DP against a worst-case adversary for some $\varepsilon < \ln(f-1)$. Then, for such an $\varepsilon$, by definition of differential privacy against a worst-case adversary, the following holds true 
\begin{equation}
    \label{eq:wort-caseDPcontradiction}
    \max_{v,u\in \honest} \divinfty{\seqadv{v}}{\seqadv{u}} \leq \varepsilon, \space \space \forall F \subset V \text{ of size $f$}.
\end{equation}
However, note that by \equationref{eq:avg-case-lower}, we have 
\[
 \Pr_{F \sim \uniformsubset{f}{V}} \left[\max_{v,u\in \honest} \divinfty{\seqadv{v}}{\seqadv{u}} > \varepsilon \right] \ge \frac{1}{n} > 0.
\]
This means that there exists a set $F_\star \subset V$ of size $f$ and $v_\star, u_\star \in V\setminus F_\star$ such that $\divinfty{\seqadv{v_\star}}{\seqadv{u_\star}} > \varepsilon$. This contradicts \eqref{eq:wort-caseDPcontradiction}, hence concludes this subcase of the proof.

\paragraph*{(ii.2) When $f \geq \kappa(G)$.} We now show that, if $\kappa(G) \le f$, no gossip protocol can satisfy $\varepsilon$-DP with $\varepsilon < \infty$ against a worst-case adversary. We also prove this statement by contradiction. Suppose that there exist a protocol $\mathcal P$ that satisfies $\varepsilon$-DP against a worst-case adversary with finite $\varepsilon$. Since we are in the worst-case adversary setting, the adversary may choose the placement of curious nodes, knowing the structure of the graph. Specifically, consider a minimal vertex cut $K$ of size $\kappa(G)$. Note that $\kappa(G) \le f < n-1$, and, hence, removing vertices of $K$ disconnects the graph into at least two connected components. Suppose $V_1$ is one of the components obtained after deleting nodes in $K$ from $G$, and set $V_2 = V \setminus (K \cup V_1)$. Then, the adversary choose nodes from $K$ as curious and select the rest of curious nodes in such a way that there is at least one non-curious node both in $V_1$ and $V_2$ (it is possible since $f < n-1$, as mentioned in \sectionref{sec:source-anonymity-with-DP}). Let $F$ be such a set. 
    
As the protocol satisfies $\varepsilon$-DP, for every pair $v,u\in \honest$ we have by definition of $\varepsilon$-DP
\[\divinfty{\seqadv{v}}{\seqadv{u}} \le \varepsilon.\]
Then, for a given observation $\sigma$ of the adversary and any $v,u\in \honest$, the following holds true
\begin{equation}
\label{eq:kappa-lb}
\Pr\left[\seqadv{v} \in \sigma\right] \le \exp(\varepsilon) \Pr\left[\seqadv{u} \in \sigma\right].
\end{equation}

Let us now consider $w_1$ and $w_2$ two arbitrary non-curious nodes respectively from $V_1$ and $V_2$.  We denote by $\sigma^{(1)}$ the set of all adversarial observations in which the first round consists only of communications of the form $(v\to \ast)$ where $v \in V_1$. Every path from $V_1$ to $V_2$ passes through $K$ by definition of $K$ as a minimal vertex cut. As $K \subset F$, we get $\Pr[\seqadv{w_1} \in \sigma^{(1)}] = 1$ and $\Pr[\seqadv{w_2} \in \sigma^{(1)}] = 0$. Using \equationref{eq:kappa-lb}, we get
\[
1 = \Pr\left[\seqadv{w_1} \in \sigma^{(1)}\right] \le \exp(\varepsilon) \Pr\left[\seqadv{w_2} \in \sigma^{(1)}\right] = \exp(\varepsilon)\cdot 0,
\]
which contradicts $\varepsilon$ being finite. Then, $\mathcal P$ cannot satisfy $\varepsilon$-DP with finite $\varepsilon$ if $f \ge \kappa(G)$. This concludes the last subcase of our proof, and hence the proof.

\end{proof}
\section{Privacy guarantees: proofs of Section~\ref{sec:guarantees}}

\label{sec:appendix-privacy-guarantees}

This section presents the proof of the main positive result of the paper (\theoremref{thm:main}).  
First, we show in \appendixref{appendix:vertex-connectivity-expansion} how the vertex connectivity, which is an essential bottleneck for differential privacy (\theoremref{thm:universal-impossibility}), can be controlled using expansion of the graph. Second, we prove \lemmaref{lemma:adversarial-density} in \appendixref{appendix:adversarial-density}, which provides an upper bound on adversarial density in the average-case with high probability. We then present the proof of \lemmaref{lemma:sufficient-varepsilon-0} regarding cobra walks and its analog (\lemmaref{lemma:dandelion-reduction}) for Dandelion. Using these results, in \appendixref{appendix:max-divergence} we bound from above the maximal divergence for any pair of possible sources for both the worst-case and the average-case adversaries.

\subsection{Controlling vertex connectivity with spectral expansion}
\label{appendix:vertex-connectivity-expansion}

As shown in \theoremref{thm:universal-impossibility}, the vertex connectivity of the graph is a bottleneck for the possibility of differential privacy in the graph. Indeed, to establish differential privacy, it is essential that the subgraph induced by non-curious nodes $\honest$ is connected. The following lemma shows that removing any set of curious nodes with small adversarial density does not disconnect the graph.
\begin{lemma}
\label{lemma:connectivity-condition}
Consider a $(d,\lambda)$-expander graph $G$ of size $n$ and a set $F \subseteq V$ with adversarial density $\alpha_F < 1 - \lambda$. Then, the subgraph of $G$ induced by $\honest$ is connected. 
\end{lemma}
\begin{proof}
Set $f = |F|$. Let $G'$ be a subgraph of $G$ induced by $\honest$. Let $A$ be an adjacency matrix of $G$ and let $B$ be an adjacency matrix of $G'$. Then, $B$ is also a principal submatrix of $A$. Then, via Cauchy interlacing law (\lemmaref{lemma:cauchy-interlace}), we have
\begin{equation}
\label{eq:eignen-honest-subgraph}
\max\{\lambda_2(B), \lambda_{n - f}(B)\} \le \max\{\lambda_2(A), \lambda_{n}(A)\} \le d\lambda.
\end{equation}
By definition of adversarial density $\alpha_F$, every $v \in \honest$ has at most $\alpha_F d$ neighbors among nodes of $F$. Then, the degree of every node in $G'$ is at least $(1 - \alpha_F)d$. Then, if $G'$ was disconnected, it would have at least two eigenvalues greater or equal than $(1 - \alpha_F)d > d\lambda$, which is not the case by \equationref{eq:eignen-honest-subgraph}.
\end{proof}

Using this lemma, we show that graphs with good spectral expansion also have a large vertex connectivity.
\begin{lemma}
\label{lemma:vertex-connectivity-expanders}
Consider a $(d,\lambda)$-expander graph $G$ of size $n$. Then, $\kappa(G) \ge (1 - \lambda)d$.
\end{lemma}
\begin{proof}
Consider the minimal vertex cut $K$ of $G$, $|K|=\kappa(G)$ and take $F = K$. Note that for any $v \in \honest$, $\degset{F}{v}/d \le \kappa(G)/d$, which implies $\alpha_F \le \kappa(G)/d$. However, a subgraph of $G$ induced by $V \setminus F$ is disconnected. Then, by the contrapositive of \lemmaref{lemma:connectivity-condition}, we must have $\kappa(G)/d \ge 1 - \lambda$.
\end{proof}

\subsection{Bound on the adversarial density for the average-case adversary (proof of \lemmaref{lemma:adversarial-density})} \label{subsec:appendix-bound-curious} \label{appendix:adversarial-density}

\AdversarialDensity*
\begin{proof} We establish the two lemma claims separately.

\paragraph*{(i) First part of the lemma.}
 Note that if $F \sim \uniformsubset{f}{V}$, then for any $v\in V$, the cardinality of the set $F \cap N(v)$ follows a hypergeometric distribution with mean $\frac{f}{n}|N(v)| = \frac{f}{n}d = \beta d$. By \cite{hoeffding1994probability, CHVATAL1979285}, $\degset{F}{v} = |F \cap N(v)|$ obeys Chernoff bounds. Then, for any $1 > \alpha \ge \beta$
\begin{equation}
    \label{eq:adversarial-density-kl-probability}
    \Pr\left[\degset{F}{v} \ge \alpha d\right] \le \exp\left(-d \divkl{B(\alpha)}{B(\beta)}\right),
\end{equation}
where $\divkl{B(x)}{B(y)} = x\ln\left(\frac{x}{y}\right) + (1-x)\ln\left(\frac{1 - x}{1 - y}\right)$ is the Kullback-Leibler divergence between Bernoulli distributed random variables $B(x)$ and $B(y)$ with success probabilities $x$ and $y$ respectively. Now, note that $\alpha, \beta < 1$. Hence,
\begin{align*}
    \divkl{B(\alpha)}{B(\beta)} &= \alpha\ln\left(\frac{\alpha}{\beta}\right) + (1 - 
    \alpha) \ln\left(\frac{1 - \alpha}{1 - \beta}\right) \\
    &\ge \alpha\ln\left(\frac{\alpha}{\beta}\right) + (1 - \alpha)\ln\left(1 - \alpha\right).
    \intertext{Futhermore, for any $0 \le x < 1$, we have $\ln(1-x) \ge 1 - \frac{1}{1 - x}$, hence}
    \divkl{B(\alpha)}{B(\beta)} &\ge \alpha\ln\left(\frac{\alpha}{\beta}\right) - \alpha\\
    &= \alpha\left(\ln\left(\alpha/\beta\right) - 1\right).
\end{align*}
Let us now consider the special case when $\alpha = C\frac{\max\{\gamma, \beta\}}{1 + \max\{\ln(\gamma) - \ln(\beta), 0\}}$, with $C = 4e$. In this case we have $C(\ln(C) - 1) \ge 2e$ and $C \ge e^2$. Note that if $\alpha \ge 1$, then the claim trivially holds true as $\alpha_F $ is always smaller than $1$ by definition of the adversarial density. Hence, without loss of generality, we will assume $\alpha < 1$. In the remaining, we consider two distinct cases: when $\beta \ge \gamma$ and when $\beta < \gamma$.

\paragraph*{(i.1) Case of $\beta \ge \gamma$.} First, suppose $\beta \ge \gamma$. Then, $\alpha = C\beta$ and
\begin{align*}
    d\divkl{B(\alpha)}{B(\beta)} &\ge d\alpha \left(\ln\left(\alpha/\beta\right) - 1\right) \\
                         &= dC\beta \left(\ln(C) - 1\right).
                         \intertext{Since $\beta \ge \gamma$, we have}
    d\divkl{B(\alpha)}{B(\beta)}&\ge dC\gamma\left(\ln(C) - 1\right) \\
                         &= \frac{C}{e}\left(\ln(C) - 1\right) \ln(n).
                         \intertext{Because $C(\ln(C) - 1) \ge 2e$, we have}
    d\divkl{B(\alpha)}{B(\beta)}&\ge 2\ln(n).
\end{align*}
Hence, by \equationref{eq:adversarial-density-kl-probability} and a union bound over all $v \in V$, the claim holds true in this case.

\paragraph*{(i.2) Case of $\beta < \gamma$.} Now, consider the second case when $\beta < \gamma$. Then, $\alpha = C\frac{\gamma}{\ln(\gamma) - \ln(\beta) + 1} = C \gamma \ln^{-1}\left(\frac{e \gamma}{\beta}\right)$. Then
\begin{align*}
    d\divkl{B(\alpha)}{B(\beta)} 
    &\ge d\alpha \left(\ln\left(\alpha/\beta\right) - 1\right) \\
    &= Cd\gamma \ln^{-1}\left(\frac{e \gamma}{\beta}\right) \left(\ln(C) - 1 + \ln\left(\frac{\gamma}{\beta}\right) - \ln\ln\left(\frac{e \gamma}{\beta}\right)\right)\\
    &= Cd\gamma \ln^{-1}\left(\frac{e \gamma}{\beta}\right) \left(\ln(C) - 2 + \ln\left(\frac{e\gamma}{\beta}\right) - \ln\ln\left(\frac{e \gamma}{\beta}\right)\right).
    \intertext{Using $\ln\ln(x) \le \frac{1}{2}\ln(x)$ for $x \ge e$, we have}
    d\divkl{B(\alpha)}{B(\beta)}
    &\ge Cd\gamma \ln^{-1}\left(\frac{e \gamma}{\beta}\right) \left(\ln(C) - 2 + \frac{1}{2}\ln\left(\frac{e\gamma}{\beta}\right) \right).
    \intertext{As $C \ge e^2$, $\ln(C) - 2\ge 0$, we get}
    d\divkl{B(\alpha)}{B(\beta)}
    &\ge Cd\gamma \ln^{-1}\left(\frac{e \gamma}{\beta}\right) \left( \frac{1}{2}\ln\left(\frac{e\gamma}{\beta}\right) \right)\\
    &= \frac{C}{2}d\gamma \\
    &= \frac{C}{2e}\ln(n).
    \intertext{Replacing $C$, we get}
    d\divkl{B(\alpha)}{B(\beta)}&\ge 2\ln(n).
\end{align*}
Hence, by \equationref{eq:adversarial-density-kl-probability} and a union bound over all $v \in V$, the claim holds true in this case as well. This concludes the first part of the lemma.

\paragraph*{(ii) Second part of the lemma.} We now consider the special case when $f/n > c$ for some positive constant $c$ and $d > \frac{\ln(n)}{c^2\delta^2}$ for some positive constant $\delta$. Set $\alpha = (1 + \delta) \beta$. Similarly to the previous case, without loss of generality, we will assume $\alpha < 1$, otherwise the claim trivially holds true. 
Then
\begin{align*}
    d\divkl{B(\alpha)}{B(\beta)} &= d\divkl{B((1 + \delta)\beta)}{B(\beta)}.
    \intertext{By~\cite{CHVATAL1979285}, we have} 
    d\divkl{B(\alpha)}{B(\beta)} &\ge 2d \beta^2 \delta^2.
    \intertext{Recall that $d > \frac{\ln(n)}{c^2\delta^2}$, hence}
    d\divkl{B(\alpha)}{B(\beta)}&\ge \frac{2\ln(n)}{c^2\delta^2} \beta^2\delta^2.
    \intertext{Since $\beta = f/n > c$, we get} 
    d\divkl{B(\alpha)}{B(\beta)}&\ge 2\ln(n).
\end{align*}
Hence, by \equationref{eq:adversarial-density-kl-probability} and a union bound over all $v \in V$, the second part of the lemma holds true.
\end{proof}

\subsection{Reduction to random walk with probabilistic die out for cobra walks}

\label{subsec:appendix-bound-divergence}
\label{appendix:max-divergence}

\begin{restatable}
{lemma}{SufficientEpsilonZero}
\label{lemma:sufficient-varepsilon-0}
\label{lemma:ratio-powers-of-Q}
\label{lemma:reduction-rw}

Consider a $(1 + \rho)$-cobra walk on a $d$-regular graph $G = (V,E)$. Let $F \subset V$ be a set of curious nodes such that the subgraph of $G$ induced by $\honest$ is connected. Let $\vec Q = \Hat{\vec A}[\honest]$ and, for $s \in \honest$, let $W^{(s)}$ be the absorbing state of the Markov chain as in \equationref{eq:block-form-absorbing-markov}. Then, for any  $v, u \in \honest$, the following holds true
\[
\divinfty{\seqadv{v}}{\seqadv{u}} \le \divinfty{W^{(v)}}{W^{(u)}} = \max_{w\in \honest} \ln\frac{(\vec I_{n-f} - (1 - \rho)\vec Q)^{-1}_{vw}}{(\vec I_{n-f} - (1 - \rho)\vec Q)^{-1}_{uw}}.\]
\end{restatable}
\begin{proof} We divide the proof in two parts: first, we show the left-hand side inequality, and then we show the right-hand side equality.

\paragraph*{(i) Left-hand side inequality.} 
To prove the left-hand side inequality, we will first introduce a notion of ``safe'' rounds. Intuitively, in a safe round, the execution behaves like a random walk on the set of non-curious nodes. Then, we will then use the definition of safe rounds to relate the cobra walk to absorbing Markov chain defined in (\ref{eq:block-form-absorbing-markov}). Finally, we will use Markovian property of cobra walks to apply Data Processing inequality and establish the left-hand side inequality.

\paragraph*{(i.1) Introducing safe rounds.}

For an execution with source node $s$, let us introduce an indicator variable $\nice_t^{(s)} \in \{0,1\}$ determined by $\activeset_t^{(s)}$ and $\commset_t^{(s)}$, which is defined as follows. 
\begin{equation}
\label{eq:nice-cobra}
\nice_t^{(s)} = \ind_{\left\{\exists v, u \in\honest \colon \activeset_t^{(s)} = \{u\} \land \commset_t^{(s)} = \{(u\to v)\}\right\}}.
\end{equation}
In other words, round $\nice_t^{(s)} = 1$ if and only if the active set consists of one non-curious node $u$ at the beginning of round $t$, and $u$ communicates a gossip to a single non-curious neighbour $v$ without branching. If $\nice_t^{(s)} = 1$, we will say that round $t$ is a safe round, and we will say it is unsafe otherwise.

\paragraph*{(i.2) Relating safe rounds to absorbing Markov Chain in (\ref{eq:block-form-absorbing-markov}).}

Recall that we consider $d$-regular graphs. Let $u,v \in \honest$ by two nodes connected by an edge. If $\activeset_t = \{u\}$, then the cobra walk does not branch in round $t$ with probability $1 - \rho$. Additionally, if it does not branch, $u$ %it 
contacts node $v \in N(u)$ with probability $1/d$. Hence, we have
\begin{align*}
    \Pr\left[\activeset_{t+1}^{(s)} = \{v\} \land \nice_t^{(s)} = 1 \mid \activeset_t^{(s)} = \{u\} \right] &= \frac{1 - \rho}{d} = (1-\rho) \vec Q_{vu}
\end{align*}
Note that for any $t$ and any $v,u \in \honest$ that are not connected by an edge, we have
\begin{align*}
\Pr[\activeset_{t+1}^{(s)} = \{v\} \land \nice_t^{(s)} = 1 \mid \activeset_t^{(s)} = \{u\}]= 0 = (1 - \rho)\vec Q_{vu}.
\end{align*}
Then, for any $v,u\in\honest$
\begin{align}
\label{eq:transient-flag-up}
\Pr[\activeset_{t+1}^{(s)} = \{v\} \land \nice_t^{(s)} = 1 \mid \activeset_t^{(s)} = \{u\}] = (1 - \rho)\vec Q_{vu},
\end{align}
which is equal to the transition probability between two transient states $v$ and $u$ of (\ref{eq:block-form-absorbing-markov}). Also, for any $u$, we have 
\begin{align}
\Pr[\nice_t^{(s)} = 0 \mid \activeset_t^{(s)} = \{u\}] &= 1 - (1-\rho)\frac{\deg_{\honest}(u)}{d} \notag\\
                                               &= \rho + \frac{\deg_F(u)}{d} (1-\rho) \notag\\
                                               &= \rho(\vec I_{n-f})_{uu} + (1 - \rho)\vec R_{uu} \label{eq:absorbing-flag-down},
\end{align}
which is equal to the probability of being absorbed at $\sink(u)$ from state $u$ in (\ref{eq:block-form-absorbing-markov}). Let $\tau^{(s)}$ be the first unsafe round, i.e.,
\begin{equation}
\label{eq:tau-s-cobra}
\tau^{(s)} = \min\{t \colon \nice_t^{(s)} = 0\}.
\end{equation}
Then $\tau^{(s)}$ is the first round in which cobra walk either branches (i.e., $|\commset_{\tau^{(s)}}| > 1$), or contacts a curious node by (\ref{eq:nice-cobra}). Then, from (\ref{eq:transient-flag-up}) and (\ref{eq:absorbing-flag-down}), at time $\tau^{(s)}$, we have $\activeset_{\tau^{(s)}}^{(s)} = \{W^{(s)}\}$, where $W^{(s)}$ is an absorbing state of chain defined in (\ref{eq:block-form-absorbing-markov}).

\paragraph*{(i.3) Applying the Data Processing inequality.}
Let $\round_t^{(s)} = (\activeset_t^{(s)}, \commset_t^{(s)})$ for every $t$, i.e., $\round_t^{(s)}$ describes round $t$ of the execution. Consider a sequence $\seq{s} = \{\round_t^{(s)}\}_{t \ge \starttime}$. Note that $\{\round_t^{(s)}\}_{t \ge \starttime}$ is \emph{Markovian} by definition of $\activeset_t^{(s)}$ and $\commset_t^{(s)}$ of cobra walk in \sectionref{section:gossip}. 
Also, note that $\tau^{(s)}$ is a \emph{stopping time} of this Markov process $\{\round_t^{(s)}\}_{t\ge \starttime}$ by definition of $\tau^{(s)}$ in (\ref{eq:tau-s-cobra}) (for definition of stopping time, see Section 6.2 of~\cite{LevinPeresWilmer2006}). 

Then, by Strong Markov Property (Proposition A.19 of~\cite{LevinPeresWilmer2006}), the law of $\{\round_t^{(s)}\}_{t \ge \tau^{(s)}}$ only depends on $\round_{\tau^{(s)}}^{(s)}$. Hence, by the Data Processing Inequality (Theorem 14 of~\cite{dataprocessing}), we have
\begin{align}
\label{eq:data-processing-one}
\divinfty{\{\round_t^{(v)}\}_{t \ge \tau^{(v)}}}{\{\round_t^{(u)}\}_{t \ge \tau^{(u)}}} \le \divinfty{\round_{\tau^{(v)}}^{(v)}}{\round_{\tau^{(u)}}^{(u)}}.
\end{align} 

Note that, the random variable $\commset_{\tau^{(s)}}^{(s)}$ characterizes a round in which the cobra walk has active set $\activeset_{\tau^{(s)}}^{(s)}$ and is conditioned to either branch or hit a curious node. Then, since cobra walk is Markovian, communications $\commset_{\tau^{(s)}}^{(s)}$ that happen in round $\tau^{(s)}$ only depend on $\activeset_{\tau^{(s)}}^{(s)}$. Hence, by the Data Processing Inequality (Theorem 14 of~\cite{dataprocessing}), we have $\divinfty{\round_{\tau^{(v)}}^{(v)}}{\round_{\tau^{(u)}}^{(u)}} \le \divinfty{\activeset_{\tau^{(v)}}^{(v)}}{\activeset_{\tau^{(u)}}^{(u)}}$. Then
\begin{align}
\label{eq:data-processing-two}
\divinfty{\{\round_t^{(v)}\}_{t \ge \tau^{(v)}}}{\{\round_t^{(u)}\}_{t \ge \tau^{(u)}}} \le \divinfty{\activeset_{\tau^{(v)}}^{(v)}}{\activeset_{\tau^{(u)}}^{(u)}}.
\end{align}
Finally, recall that we showed in (i.2) that $\activeset_{\tau^{(s)}}^{(s)} = \{W^{(s)}\}$, where $W^{(s)}$ is an absorbing state of the Markov chain (\ref{eq:block-form-absorbing-markov}). Hence,
\begin{align}
\label{eq:data-processing-three}
\divinfty{\{\round_t^{(v)}\}_{t \ge \tau^{(v)}}}{\{\round_t^{(u)}\}_{t \ge \tau^{(u)}}} \le \divinfty{W^{(v)}}{W^{(u)}}. 
\end{align}
    
    Note also that since no curious node is contacted before time $\tau^{(s)}$, $\seqadv{s}$ can be obtained from $\{\round_t^{(s)}\}_{t \ge \tau^{(s)}} = \{(\commset_{t}^{(s)}, \activeset_{t}^{(s)})\}_{t \ge \tau^{(s)}}$ via a deterministic mapping from definition of $\seqadv{s}$ in \sectionref{sec:source-anonymity-with-DP}. Then, applying the  Data Processing Inequality (Theorem 14 of~\cite{dataprocessing}) again, for any $v,u \in \honest$ we have 
    
    \begin{equation}
    \label{eq:data-process-seqadv-xt}
     \divinfty{\seqadv{v}}{\seqadv{u}} \le \divinfty{\{\round_t^{(v)}\}_{t \ge \tau^{(v)}}}{\{\round_t^{(u)}\}_{t \ge \tau^{(u)}}}.
    \end{equation}
    Combining the above with (\ref{eq:data-processing-three}), twe get
    \begin{equation}
    \label{eq:data-process-seqadv-w}
         \divinfty{\seqadv{v}}{\seqadv{u}} \leq \divinfty{W^{(v)}}{W^{(u)}},
    \end{equation}
    which concludes the proof of the left inequality.

    \paragraph*{(ii) Right-hand side equality.} 
    As we mentioned in (i.2), $W^{(s)}$ follows the Markov chain as in (\ref{eq:block-form-absorbing-markov}). Then, by \lemmaref{lemma:absorbtion-prob} the absorption probabilities matrix can be written as follows
    \[
    \vec M = \left((1-\rho)\vec R + \rho\vec I_{n-f}\right)\left(\vec I_{n - f} - (1-\rho)\vec Q\right)^{-1},
    \]
    where $M_{wv}$ corresponds to the probability of starting at $v$ and being absorbed at $\sink(w)$. Also, note that $\vec Q$ and $\vec I_{n-f}$ are symmetric, hence so is $\left(\vec I_{n - f} - (1-\rho)\vec Q\right)^{-1}$. Then, for any $v \in \honest$
    \begin{align*}
    \Pr[W^{(v)} = w] &= \left((1-\rho)R_{ww} + \rho\right) \left(\vec I_{n - f} - (1-\rho)\vec Q\right)^{-1}_{wv} \\
    &= \left((1-\rho)R_{ww} + \rho\right) \left(\vec I_{n - f} - (1-\rho)\vec Q\right)^{-1}_{vw}.
    \end{align*}
    Similarly for any $u\in \honest$
    \[
    \Pr[W^{(u)} = w] = \left((1-\rho)R_{ww} + \rho\right) \left(\vec I_{n - f} - (1-\rho)\vec Q\right)^{-1}_{uw}.
    \]
    Hence, by definition of max divergence, for any $v, u \in \honest$ we have
    \begin{align*}
    \divinfty{W^{(v)}}{W^{(u)}} &= \max_{w \in \honest}\ln \frac{\Pr[W^{(v)} = w]}{\Pr[W^{(u)} = w]}\\ 
                                &= \max_{w \in \honest}\ln \frac{ \left(\vec I_{n - f} - (1-\rho)\vec Q\right)^{-1}_{vw}}{ \left(\vec I_{n - f} - (1-\rho)\vec Q\right)^{-1}_{uw}}.
    \end{align*}
    The above concludes the proof.
\end{proof}

\begin{remark}
    One might notice that we in fact have an \emph{equality} in all of (\ref{eq:data-processing-one}), (\ref{eq:data-processing-two}) and (\ref{eq:data-processing-three}), since $W^{(s)}$ is such that we have $\round_{\tau^{(s)}}^{(s)} = (\{W^{(s)}\}, \commset_{\tau^{(s)}}^{(s)})$, i.e., we may recover $W^{(s)}$ simply by looking at the first element of sequence $\{\round_t^{(s)}\}_{t \ge \tau^{(s)}}$. Nevertheless, inequality in (\ref{eq:data-process-seqadv-xt}) cannot be reversed, since the curious nodes only observe part of the communications, and hence, the final inequality (\ref{eq:data-process-seqadv-w}) may not be an equality.
\end{remark}

\subsection{Reduction to random walk with probabilistic die out for Dandelion} \label{sec:dandelion-privacy}
We now show that the reduction in \lemmaref{lemma:sufficient-varepsilon-0} applies to $\rho$-Dandelion. The rest of the proof will be the same for Dandelion and cobra walks.

First, we explain some specificities of our threat model with respect to Dandelion protocol. Recall that $\commset_t$ is the set of communications occurred in round $t$, $\fadv$ is a function that takes as input communications $\commset_t$ from a single round and outputs only the communications of $\commset_t$ visible to the adversary and $\tadv$ is the first round in which one of the curious nodes received the gossip. Let $\anonphase_t$ be the phase indicator variable for round $t$ of the execution held by the oracle (equal to $1$ in the anonymity phase, and $0$ in the spreading phase). Note that any node can retrieve the value of $\anonphase$ by calling the oracle. Hence, besides observing communications in $\fadv\left(\commset_t\right)$, the adversary can also observe the phase in which round $t$ has occurred (i.e., the value $\anonphase_t$). To account for this, we modify slightly our definition of $S_\textsc{adv}$ in \sectionref{sec:privacy-gossip} that describes the observations of the adversary. Studying the Dandelion protocol, we define $S_\textsc{adv} = \{\left(\fadv(\commset_t), \anonphase_t\right)\}_{t \ge \tadv}$. That is, the definition of $S_\textsc{adv}$ is slightly stronger than the one given in \sectionref{sec:privacy-gossip} since we allow the adversary to observe the phase of the protocol.

Now, we demonstrate that the reduction in \lemmaref{lemma:sufficient-varepsilon-0} applies to Dandelion even in this slightly stronger setting.
\begin{restatable}{lemma}{DandelionReduction}
\label{lemma:dandelion-reduction}
    Consider $\rho$-Dandelion on a $d$-regular graph $G = (V,E)$. Let $F \subset V$ be a set of curious nodes such that the subgraph of $G$ induced by $\honest$ is connected. Let $\vec Q = \Hat{\vec A}[\honest]$ and, for $s \in \honest$, let $W^{(s)}$ be the absorbing state of the Markov chain as in \equationref{eq:block-form-absorbing-markov}. Then, for any  $v, u \in \honest$, the following holds true
    \[ \divinfty{\seqadv{v}}{\seqadv{u}} \le \divinfty{W^{(v)}}{W^{(u)}} = \max_{w\in \honest} \ln\frac{(\vec I_{n-f} - (1 - \rho)\vec Q)^{-1}_{vw}}{(\vec I_{n-f} - (1 - \rho)\vec Q)^{-1}_{uw}},\]
    where $\seqadv{s} = \left\{\left(\fadv(\commset_t^{(s)}), \anonphase_t^{(s)}\right)\right\}_{t \ge \tadv^{(s)}}$ as defined above.
\end{restatable}
\begin{proof}
Let $s\in \honest$ be arbitrary. Note that the right-hand equality follows from the proof of (\ref{lemma:sufficient-varepsilon-0}). To prove the left-hand side inequality, we will first introduce a notion of ``safe'' rounds, in a similar way as for the proof of \lemmaref{lemma:sufficient-varepsilon-0}. We call a round safe, if the execution is still in the anonymity phase (i.e., $\anonphase = 1$), and nodes which communicated during the current round are non-curious. We also introduce an indicator variable $\nice_t$ corresponding to a round $t$ being safe. Formally, 
\begin{equation}
\label{eq:nice-dandelion}
    \nice_t^{(s)} = \ind_{\left\{\anonphase_{t}^{(s)} = 1 \land \exists v,u \in \honest \colon \activeset_t^{(s)} = \{u\}\land\commset_t^{(s)} = \{(u \to v)\}\right\}}.
\end{equation}

\paragraph*{(i) Relating safe rounds to an absorbing Markov chain as in (\ref{eq:block-form-absorbing-markov}).}

If $\anonphase_{t-1}^{(s)} = 1$, then the execution of Dandelion does not enter a spreading phase in round $t$ with probability $1-\rho$. Also, if the execution is still in the anonymity phase, there is only one active node $u$. Since graph is $d$-regular, $u$ contacts any fixed $v \in N(u)$ with probability $1/d$. Hence, we have the following for any $u, v\in \honest, v\in N(u)$
\begin{align*}
    \Pr\left[\activeset_{t+1}^{(s)} = \{v\} \land \nice_t^{(s)} = 1 \mid \activeset_t^{(s)} = \{u\} \land \anonphase_{t-1}^{(s)}=1 \right] &= \frac{1 - \rho}{d} = (1-\rho) \vec Q_{vu}
\end{align*}
Note that for any $t$ and any $v,u \in \honest$ such that $v\not \in N(u)$, we have
\begin{align*}
\Pr[\activeset_{t+1}^{(s)} = \{v\} \land \nice_t^{(s)} = 1 \mid \activeset_t^{(s)} = \{u\} \land \anonphase_{t-1}^{(s)} =1]= 0 = (1 - \rho)\vec Q_{vu}.
\end{align*}
Then, for any $v,u \in \honest$, we get
\begin{align}
\label{eq:transient-flag-up-dandelion}
\Pr[\activeset_{t+1}^{(s)} = \{v\} \land \nice_t^{(s)} = 1 \mid \activeset_t^{(s)} = \{u\} \land \anonphase_{t-1}^{(s)}=1 ] = (1 - \rho)\vec Q_{vu},
\end{align}
which is equal to the transition probability between two transient states $v$ and $u$ of (\ref{eq:block-form-absorbing-markov}). Also, for any $u$, we have 
\begin{align}
\Pr[\nice_t^{(s)} = 0 \mid \activeset_t^{(s)} = \{u\}\land \anonphase_{t-1}^{(s)}=1 ] &= 1 - (1-\rho)\frac{\deg_{\honest}(u)}{d} \notag\\
                                               &= \rho + \frac{\deg_F(u)}{d} (1-\rho) \notag\\
                                               &= \rho(\vec I_{n-f})_{uu} + (1 - \rho)\vec R_{uu} \label{eq:absorbing-flag-down-dandelion},
\end{align}
which is equal to the probability of being absorbed at $\sink(u)$ from state $u$ in (\ref{eq:block-form-absorbing-markov}). Let $\tau^{(s)}$ be the first unsafe round, i.e., 
\begin{equation}
    \label{eq:tau-s-dandelion}
    \tau^{(s)} = \min\{t \colon \nice_t^{(s)} = 0\}.
\end{equation}
Then, $\tau^{(s)}$ is the first round in which either the execution enters the spreading phase (i.e., $\anonphase_{\tau^{(s)}-1} = 1$, but $\anonphase_{\tau^{(s)}} = 0$), or a curious node is contacted during the round. Then, from (\ref{eq:transient-flag-up-dandelion}) and (\ref{eq:absorbing-flag-down-dandelion}), at time $\tau^{(s)}$, we have $\activeset_{\tau^{(s)}}^{(s)} = \{W^{(s)}\}$, where $W^{(s)}$ is an absorbing state of the chain defined in (\ref{eq:block-form-absorbing-markov}).

\paragraph*{(ii) Applying the Data Processing inequality.}
Let $\round_t^{(s)} = (\activeset_t^{(s)}, \commset_t^{(s)}, \anonphase_t^{(s)})$ for every $t$, i.e., $\round_t^{(s)}$ describes a round $t$ of the execution of Dandelion. Consider a sequence $\{\round_t^{(s)}\}_{t \ge \starttime}$. Note that $\{\round_t^{(s)}\}_{t \ge \starttime}$ is \emph{Markovian}. Indeed, as we mention in \sectionref{section:gossip}, the oracle determines the value of $\anonphase_{t}^{(s)}$ in a randomized way based on the value of $\anonphase_{t-1}^{(s)}$; additionally, $\activeset_t^{(s)}, \commset_t^{(s)}$ only depend on the previous round $\round_{t-1}^{(s)}$ and the phase of the execution $\anonphase_t^{(s)}$.
Then $\tau^{(s)}$ is a \emph{stopping time} for a Markov chain $\{\round_t^{(s)}\}_{t\ge \starttime}$ by definition of $\tau^{(s)}$ in (\ref{eq:tau-s-dandelion}) (for definition of a stopping time, see Section 6.2 of~\cite{LevinPeresWilmer2006}).

Then, by Strong Markov Property (Proposition A.19 of~\cite{LevinPeresWilmer2006}), the law of $\{\round_t^{(s)}\}_{t \ge \tau^{(s)}}$ only depends on $\round_{\tau^{(s)}}^{(s)}$. Hence, by the Data Processing Inequality (Theorem 14 of~\cite{dataprocessing}), we have
\begin{align}
\label{eq:data-processing-one-dandelion}
\divinfty{\{\round_t^{(v)}\}_{t \ge \tau^{(v)}}}{\{\round_t^{(u)}\}_{t \ge \tau^{(u)}}} \le \divinfty{\round_{\tau^{(v)}}^{(v)}}{\round_{\tau^{(u)}}^{(u)}}.
\end{align}
Note that, by definition of $\nice_t^{(s)}$ in (\ref{eq:nice-dandelion}), the pair $(\commset_{\tau^{(s)}}^{(s)}, \anonphase_{\tau^{(s)}}^{(s)})$ characterizes the round of Dandelion in the anonymity phase in which it has active set $\activeset_{\tau^{(s)}}^{(s)}$ and is conditioned to hit a curious node or enter a spreading phase. Then, since Dandelion is Markovian, the pair $(\commset_{\tau^{(s)}}^{(s)}, \anonphase_{\tau^{(s)}}^{(s)})$ only depends on $\left(\activeset_{\tau^{(s)}}^{(s)}, \anonphase_{\tau^{(s)} - 1}^{(s)}\right) = \left(\activeset_{\tau^{(s)}}^{(s)}, 1\right)$. Hence, by the Data Processing Inequality (Theorem 14 of~\cite{dataprocessing}), we have \[\divinfty{\round_{\tau^{(v)}}^{(v)}}{\round_{\tau^{(u)}}^{(u)}} \le \divinfty{\left(\activeset_{\tau^{(v)}}^{(v)}, 1\right)}{\left(\activeset_{\tau^{(u)}}^{(u)}, 1\right)} = \divinfty{\activeset_{\tau^{(v)}}^{(v)}}{\activeset_{\tau^{(u)}}^{(u)}}.\] Combining this with (\ref{eq:data-processing-one-dandelion}), we get
\begin{align}
\label{eq:data-processing-two-dandelion}
\divinfty{\{\round_t^{(v)}\}_{t \ge \tau^{(v)}}}{\{\round_t^{(u)}\}_{t \ge \tau^{(u)}}} \le \divinfty{\activeset_{\tau^{(v)}}^{(v)}}{\activeset_{\tau^{(u)}}^{(u)}}.
\end{align}
Finally, recall that we showed in (i) that $\activeset_{\tau^{(s)}}^{(s)} = \{W^{(s)}\}$, where $W^{(s)}$ is an absorbing state of the Markov chain (\ref{eq:block-form-absorbing-markov}). Hence, from (\ref{eq:data-processing-two-dandelion}), we have
\begin{align}
\label{eq:data-processing-three-dandelion}
\divinfty{\{\round_t^{(v)}\}_{t \ge \tau^{(v)}}}{\{\round_t^{(u)}\}_{t \ge \tau^{(u)}}} \le \divinfty{W^{(v)}}{W^{(u)}}. 
\end{align}

Note also that since no curious node is contacted before time $\tau^{(s)}$, $\seqadv{s}$ can be obtained from $\{\round_t^{(s)}\}_{t \ge \tau^{(s)}} = \{(\activeset_t^{(s)}, \commset_t^{(s)}, \anonphase_t^{(s)})\}_{t \ge \tau^{(s)}}$ via a deterministic mapping by definition of $\seqadv{s}$ above in \appendixref{sec:dandelion-privacy}. Then, applying the  Data Processing Inequality (Theorem 14 of~\cite{dataprocessing}) again, for any $v,u \in \honest$ we have
\begin{equation}
\label{eq:data-process-seqadv-xt-dandelion}
\divinfty{\seqadv{v}}{\seqadv{u}} \le \divinfty{\{\round_t^{(v)}\}_{t \ge \tau^{(v)}}}{\{\round_t^{(u)}\}_{t \ge \tau^{(u)}}}.
\end{equation}
Combining the above with (\ref{eq:data-processing-three-dandelion}), we get
\begin{equation}
\label{eq:data-process-seqadv-w-dandelion}
        \divinfty{\seqadv{v}}{\seqadv{u}} \leq \divinfty{W^{(v)}}{W^{(u)}},
\end{equation}
which concludes the proof.
\end{proof}

\subsection{Upper bounding max divergence}

\begin{restatable}
{lemma}{SpectralQ}
\label{lemma:spectral-Q}
Consider a $(d,\lambda)$-expander graph $G = (V,E)$ of size $n$, and a set of curious nodes $F \subseteq V$, with $|F| = f$. Then, the following assertions hold true
\begin{enumerate}[(a)]
    \item $\max\{1 - \alpha_F, 1 - (1 + \lambda)\frac{f}{n}\} \le \lambda_1(\vec Q) \le 1 - (1 - \lambda)\frac{f}{n}$.
    \item $\max\{\left|\lambda_2(\vec Q)\right|, \left|\lambda_{n-f}(\vec Q)\right|\} \le \lambda$.
\end{enumerate}
\end{restatable}
\begin{proof}
We first show the left-hand side inequality of assertion (a). Then, we show the right-hand side of assertion (a). Finally, we show assertion (b).
\paragraph*{(i) Left-hand side inequality of (a).} We start by showing the left inequality of (a). 
Consider matrices $\vec Q + \vec R$ and $\vec Q$. Note that, by definition of $\vec Q$ and $\vec R$, each row of $\vec Q + \vec R$ sums up to 1. Since $\vec Q$ and $\vec R$ have non-negative entries, this implies that $1$ is the largest eigenvalue of $\vec Q + \vec R$. Additionally, since $\degset{F}{v} \le \alpha_F d$, we know $\norm{\vec R}_{op} \le \alpha_F$. Then, by eigenvalue stability inequality (\lemmaref{lemma:eigenvalue-stability}) we have
\[
\left|\lambda_1(\vec Q + \vec R) - \lambda_1(\vec Q)\right| \le \norm{\vec R}_{op} \le \alpha_F.
\]
Since $\lambda_1(\vec Q + \vec R) = 1$, we have
\begin{equation}
\label{eq:lambda1-alphaF}
\left|1 - \lambda_1(\vec Q)\right| \le \alpha_F.
\end{equation}
Furthermore, 
\begin{align}
\lambda_1(\vec Q) &= \sup_{\vec x \in \mathbb R^{n-f}\setminus \{0\}} \frac{\norm{\vec Qx}_2}{\norm{x}_2} \notag \\
             &\ge \frac{\norm{\vec Q \vec 1_{n-f}}_2}{\norm{\vec 1_{n-f}}_2} \notag \\ 
             &\ge \frac{\norm{\vec Q \vec 1_{n-f}}_1}{\sqrt{n-f}\norm{\vec 1_{n-f}}_2} \notag \\
             &= \frac{\norm{\vec Q \vec 1_{n-f}}_1}{n-f}.\label{eq:lambda1-through-1norm} 
\end{align}
Note that
\[
\norm{\vec Q \vec 1_{n-f}}_1 = \sum_{u \in \honest} \frac{\degset{\honest}{u}}{d} = (n - f) - \frac{|\partial F|}{d}.
\]
Additionally, by expander mixing lemma (\cite{Pseudorandomness}, Lemma 4.15) we have
\[
|\partial F| \le \frac{d}{n} f(n-f) (1+\lambda).
\]
Hence,
\begin{equation}
\label{eq:1norm-mixing-lemma}
\norm{\vec Q \vec 1_{n-f}}_1 = (n - f) - \frac{|\partial F|}{d} \ge (n-f) - (n-f) \frac{f}{n}(1 + \lambda) = (n-f)\left(1 - \frac{f}{n}(1 + \lambda)\right).
\end{equation}
By combining \equationref{eq:1norm-mixing-lemma} and \equationref{eq:lambda1-through-1norm}, we get
\begin{align*}
\lambda_1(\vec Q) &\ge \frac{\norm{\vec Q \vec 1_{n-f}}_1}{n-f} \\
             &\ge \frac{(n-f)\left(1 - \frac{f}{n}(1 + \lambda)\right)}{n-f} \\
             &= 1 - \frac{f}{n}(1 + \lambda),
\end{align*}
which together with \equationref{eq:lambda1-alphaF} establishes the left inequality of (a).

\paragraph*{(ii) Right-hand side inequality of (a).} Now, we show the right inequality of (a). To upper bound $\lambda_1(\vec Q)$, we use a more involved approach. Recall that the powers of $\vec Q$ characterize probability of a random walk going from one node to another while staying within $\honest$. Then, $\frac{1}{n-f}\vec 1^{T}_{n-f} \vec Q^t \vec 1_{n-f}$ is the probability of a random walk staying within $\honest$ after $t$ steps when the initial position is chosen uniformly at random. By Theorem 4.17 of~\cite{Pseudorandomness}, such probability can be upper bounded as follows
\[
\frac{1}{n-f}\vec 1^{T}_{n-f} \vec Q^t \vec 1_{n-f} \le \left(1 - \frac{f}{n} + \lambda\left(\frac{f}{n}\right)\right)^t = \left(1 - (1 - \lambda)\frac{f}{n}\right)^t.
\]
Hence,
\[
\vec 1_{n-f} \vec Q^t \vec 1_{n-f} \le (n-f)\left(1 - (1 - \lambda)\frac{f}{n}\right)^t.
\]
Note that since entries of $\vec Q^t$ are non-negative, the maximum row sum of $\vec Q^t$ upper bounds $\lambda_1(\vec Q^t)$. As $\vec 1_{n-f} \vec Q^t \vec 1_{n-f}$ is the sum of all entries of $\vec Q^t$, it also upper bounds $\lambda_1(\vec Q^t)$. Then
\[
\lambda_1(\vec Q)^t = \lambda_1(\vec Q^t) \le (n-f)\left(1 - (1 - \lambda)\frac{f}{n}\right)^t,
\]
which implies that for arbitrary $t \ge 0$
\[
\left(\frac{\lambda_1(\vec Q)}{1 - (1 - \lambda)\frac{f}{n}}\right)^t \le n-f.
\]
Since the above holds true for arbitrary large $t$, we must have $\frac{\lambda_1(\vec Q)}{1 - (1 - \lambda)\frac{f}{n}} \le 1$, which establishes the right inequality of (a).

\paragraph*{(iii) Assertion (b).} Now, we establish part (b). Since $\vec Q$ is a principal submatrix of normalized adjacency matrix $\Hat{\vec A}$ of $G$, by Cauchy interlacing law (\lemmaref{lemma:cauchy-interlace}) we know $\max\{\left|\lambda_2(\vec Q)\right|, \left|\lambda_{n-f}(\vec Q)\right|\} \le \lambda$. Thus part (b) holds true.
\end{proof}

To bound the coordinates of the first eigenvector of $\vec Q$, we introduce two perturbations of $\vec Q$, namely $\overline{\vec Q}$ and $\underline{\vec Q}$. They have spectral properties similar to $\vec Q$, and have an explicit first eigenvector. Moreover, all entries of $\overline{\vec Q}$ and $\underline{\vec Q}$ provide a lower and an upper bound respectively of the corresponding entries of $\vec Q$. Later, we will use them to derive bounds for the first eigenvector of $\vec Q$.
\begin{restatable}{lemma}{SpectralQHat}
\label{lemma:spectral-Q-hat}
Let $G, F, \vec Q, \vec R$ and $\alpha_F$ be as in \lemmaref{lemma:spectral-Q}. We define $\overline{\vec Q} =(\vec I_{n-f} - \vec R)^{-1/2}\vec Q(\vec I_{n-f} - \vec R)^{-1/2}$ and $\underline{\vec Q} = (1 - \alpha_F)\overline{\vec Q}$. Then, the following assertions hold
\begin{enumerate}[(a)]
    \item $\lambda_1\left(\overline{\vec Q}\right) = 1$ with $\ell_2$-normalized eigenvector $\overline{\vec \eigen}_1$, where $\overline{\eigen}_{1w}^2 = \frac{\degset{\honest}{w}}{\sum_{u\in \honest} \degset{\honest}{u}}$ for all $w \in \honest$.
    \item $\lambda_1\left(\underline{\vec Q}\right) = 1 - \alpha_F$ with $\ell_2$-normalized eigenvector $\underline{\vec \eigen}_1 = \overline{\vec \eigen}_1$.
    \item $\max\left\{\left|\lambda_2\left(\overline{\vec Q}\right)\right|, \left|\lambda_{n-f}\left(\overline{\vec Q}\right)\right|\right\} \le \frac{\lambda}{1-\alpha_F}$.
    \item $\max\left\{\left|\lambda_2\left(\underline{\vec Q}\right)\right|, \left|\lambda_{n-f}\left(\underline{\vec Q}\right)\right|\right\} \le \lambda$.
    \item $\vec O_{(n-f) \times (n-f)} \preceq \underline{\vec Q} \preceq \vec Q \preceq \overline{\vec Q}$.
\end{enumerate}
\end{restatable}
\begin{proof}
    We prove different assertions separately.
    
    \paragraph*{(i) Assertions (a), (b).}
    Recall that by definition of $\vec Q$ and $\vec R$, each row of the matrix $\vec Q + \vec R$ sums up to $1$. Then
    \[
    (\vec Q + \vec R)\vec 1_{n-f} = \vec 1_{n-f} \Rightarrow \vec Q\vec 1_{n-f} = (\vec I_{n-f} - \vec R)\vec 1_{n-f}.
    \]
    Hence,
    \[
    (\vec I_{n-f} - \vec R)^{-1/2}\vec Q (\vec I_{n-f} - \vec R)^{-1/2} (\vec I_{n-f} - \vec R)^{1/2}\vec 1_{n-f} = (\vec I_{n-f} - \vec R)^{1/2}\vec 1_{n-f}.
    \]
    which implies
    \[
    \overline{\vec Q} (\vec I_{n-f} - \vec R)^{1/2}\vec 1_{n-f} = (\vec I_{n-f} - \vec R)^{1/2}\vec 1_{n-f}.
    \]
    In other words, $1$ is an eigenvalue of $\overline{\vec Q}$ with eigenvector $(\vec I_{n-f} - \vec R)^{1/2}\vec 1_{n-f}$. Note that the $w^\text{th}$ coordinate of this vector is $\sqrt{\frac{\degset{\honest}{w}}{d}}$ by definition of $\vec R$. Denote $\overline{\vec \eigen}_1$ to be $\ell_2$-normalized version of this vector. Then, we have $\overline{\eigen}_{1w} = \sqrt{\frac{\degset{\honest}{w}}{\sum_{u \in \honest}\degset{\honest}{u}}}$. This establishes (a). Part (b) follows since $\underline{\vec Q} = (1 - \alpha_F)\overline{\vec Q}$.
    
   \paragraph*{(ii) Assertions (c), (d).}
   To show (c) and (d), we use Theorem H.1.c. of~\cite{Marshall1980InequalitiesTO}. Note that entries of $\vec Q$ and $(\vec I_{n-f} - \vec R)^{-1/2}$ are non-negative, which implies $\overline{\vec Q}$ is a matrix with non-negative entries. Then, by part (a) of \lemmaref{lemma:perron}, $\lambda_1\left(\overline{\vec Q}\right)$ is the largest eigenvalue in absolute value, i.e. $\lambda_1\left(\overline{\vec Q}\right)$ is the first singular value of $\overline{\vec Q}$. In this case, $\max\left\{\left|\lambda_2\left(\overline{\vec Q}\right)\right|, \left|\lambda_{n-f}\left(\overline{\vec Q}\right)\right|\right\}$ is the second  singular value of $\overline{\vec Q}$. Also, since $\vec R \preceq \alpha_F \vec I_{n-f}$, we have $\vec I_{n-f} - \vec R \succeq (1 - \alpha_F)\vec I_{n-f}$. Hence, $(\vec I_{n-f} - \vec R)^{-1/2}$ is diagonal with every element upper bounded by $(1 - \alpha_F)^{-1/2}$. Theorem H.1.c. of~\cite{Marshall1980InequalitiesTO} with $k = 1$ and $i_1 = 2$ implies
    \begin{align*}
    \max\left\{\left|\lambda_2\left(\overline{\vec Q}\right)\right|, \left|\lambda_{n-f}\left(\overline{\vec Q}\right)\right|\right\} &= \sigma_2\left(\overline{\vec Q}\right) \\
                          &\le \sigma_1\left((\vec I_{n-f} - \vec R)^{-1/2}\right) \sigma_2(\vec Q) \sigma_1\left((\vec I_{n-f} - \vec R)^{-1/2}\right)\\
                          &\le (1 - \alpha_F)^{-1/2} \lambda (1 - \alpha_F)^{-1/2} \\
                          &\le \frac{\lambda}{1 - \alpha_F}.
    \end{align*}
    This establishes (c). Part (d) follows since $\underline{\vec Q} = (1 - \alpha_F)\overline{\vec Q}$.
    
    \paragraph*{(iii) Assertion (e).} First, we show the leftmost inequality. We express $\underline{\vec Q}$ as follows 
    \begin{equation}
    \label{eq:underline-q-rewritten}
    \underline{\vec Q} = \left(\frac{\vec I_{n-f} - \vec R}{1 - \alpha_F}\right)^{-1/2}\vec Q\left(\frac{\vec I_{n-f} - \vec R}{1 - \alpha_F}\right)^{-1/2}.
    \end{equation}
    Since $\alpha_F < 1 - \lambda \le 1$, we have $\vec R \preceq \alpha_F \vec I_{n-f} \prec I_{n-f}$. Hence, entries of $\left(\frac{\vec I_{n-f} - \vec R}{1 - \alpha_F}\right)^{-1/2}$ are positive. Additionally, $\vec Q$ has non-negative entries, hence 
    \[
    \underline{\vec Q} \succeq \vec O_{(n-f) \times (n-f)}.
    \]
    
    Now we show the middle inequality. Since $\vec I_{n-f} - \vec R \succeq (1 - \alpha_F)\vec I_{n-f}$, every diagonal entry of $\left(\frac{\vec I_{n-f} - \vec R}{1 - \alpha_F}\right)$ is lower bounded by $1$. Therefore, $\left(\frac{\vec I_{n-f} - \vec R}{1 - \alpha_F}\right)^{-1/2} \preceq \vec I_{n-f}$. Then, by \lemmaref{lemma:partial-order-mult} and \equationref{eq:underline-q-rewritten}, we get $\underline{\vec Q} \preceq \vec Q$. 
    
    Finally, we show the rightmost inequality. Similarly to \equationref{eq:underline-q-rewritten}, we can express $\overline{\vec Q}$ as
    \[
    \overline{\vec Q} = \left(\vec I_{n-f} - \vec R\right)^{-1/2} \vec Q \left(\vec I_{n-f} - \vec R\right)^{-1/2}.
    \]
    Since $\vec I_{n-f} - \vec R \preceq \vec I_{n-f}$, we have $\overline{\vec Q} \succeq \vec I_{n-f}^{-1/2} \vec Q \vec I_{n-f}^{-1/2} = \vec Q$.
\end{proof}

We now provide upper and lower bounds on the entries of the first eigenvector of $\vec Q$. They will be used to upper and lower bound the infinite series in \equationref{eq:Q-series}. 
\begin{restatable}
{lemma}{QDelocalization}
\label{lemma:Q-delocalization}
Consider a $(d,\lambda)$-expander graph $G = (V,E)$ of size $n$, and a set of curious nodes $F \subset V$ with $|F| = f$ and adversarial density $\alpha_F < 1-\lambda$. Let $\vec Q$ and $\vec R$ be as in \eqref{eq:block-form-absorbing-markov}. Then, there exists an $\ell_2$-normalized eigenvector of $\vec Q$ with non-negative coordinates, denoted $\vec \eigen_1$. Moreover, for any $v\in \honest$, the following holds true
\[
(1 - \alpha_F)^{T + 1} \frac{1}{2(n-f)} \le \eigen_{1v}^2 \le (1 - \alpha_F)^{-T - 1} \frac{2}{n-f},
\]
where $T = \left\lceil\log_{\frac{\lambda}{1 - \alpha_F}}\left(\frac{1 - \alpha_F}{4(n-f)}\right)\right\rceil$.
\end{restatable}
\begin{remark} Note that \lemmaref{lemma:Q-delocalization} can be seen as a statement on \emph{delocalization} (as in, e.g.,~\cite{Rudelson2015Delocalization}) of the first eigenvector of an arbitrary subgraph of $G$ in the worst-case, or of a \emph{random} subgraph of $G$ in the average-case adversary setting. In that sense, \lemmaref{lemma:Q-delocalization} is of independent interest for research on delocalization. \end{remark}

\begin{proof}
Note that since $\vec Q \succeq \vec O_{(n-f) \times (n-f)}$, by part (b) of \lemmaref{lemma:perron} there exists the first eigenvector of $\vec Q$ with non-negative coordinates. Let $\vec\eigen_1$ be $\ell_2$ normalized version of such eigenvector. We establish upper and lower bounds for coordinates of $\vec \eigen_1$ using matrices $\overline{\vec Q}$ and $\underline{\vec Q}$ from \lemmaref{lemma:spectral-Q-hat}. From \lemmaref{lemma:spectral-Q-hat}, we know that
\[
\vec O_{(n-f) \times (n-f)} \preceq \underline{\vec Q} \preceq \vec Q \preceq \overline{\vec Q}.
\]
Then, using \lemmaref{lemma:partial-order-power}, for any $t \ge 0$ we get
\begin{equation}
\label{eq:chain-ineq-Q}
\underline{\vec Q}^t \preceq \vec Q^t \preceq \overline{\vec Q}^t.
\end{equation}
Using \lemmaref{lemma:matrix-powers-first-eigenvector} for $\vec Q^t$ we get
\begin{equation}
    \label{eq:sandwich-Q}
    \lambda_1(\vec Q)^t \vec\eigen_1\vec\eigen_1^\top - \lambda^t\vec J_{n-f} \preceq \vec Q^t \preceq \lambda_1(\vec Q)^t \vec\eigen_1\vec\eigen_1^\top + \lambda^t\vec J_{n-f}.
\end{equation}
We will first show a lower bound on the coordinates of $\vec\eigen_1$, and then an upper bound.

\paragraph*{(i) Lower bound on coordinates of $\vec\eigen_1$.}
Using \lemmaref{lemma:matrix-powers-first-eigenvector} for $\underline{\vec Q}^t$ with part (d) of \lemmaref{lemma:spectral-Q-hat}, we get
\begin{equation}
    \label{eq:sandwich-Q-underline}
    \lambda_1\left(\underline{\vec Q}\right)^t \underline{\vec\eigen}_1\underline{\vec\eigen}_1^\top - \lambda^t\vec J_{n-f} \preceq \underline{\vec Q}^t.
\end{equation}
Combining the left inequality of \equationref{eq:chain-ineq-Q} with \equationref{eq:sandwich-Q} and \equationref{eq:sandwich-Q-underline}, we have
\[
\lambda_1\left(\underline{\vec Q}\right)^t \underline{\vec\eigen}_1\underline{\vec\eigen}_1^\top - \lambda^t\vec J_{n-f} \preceq \underline{\vec Q}^t  \preceq \vec Q^t \preceq \lambda_1(\vec Q)^t \vec\eigen_1\vec\eigen_1^\top + \lambda^t\vec J_{n-f}.
\]
By part (a) of \lemmaref{lemma:spectral-Q}, we know $\lambda_1(\vec Q) \le 1$ and by part (b) of \lemmaref{lemma:spectral-Q-hat}, we know $\lambda_1\left(\underline{\vec Q}\right) = 1 -\alpha_F$. The above implies
\[
(1 - \alpha_F)^t \underline{\vec\eigen}_1\underline{\vec\eigen}_1^\top - \lambda^t\vec J_{n-f} \preceq \vec\eigen_1\vec\eigen_1^\top + \lambda^t\vec J_{n-f},
\]
hence,
\[
(1 - \alpha_F)^t \underline{\vec\eigen}_1\underline{\vec\eigen}_1^\top - 2\lambda^t\vec J_{n-f} \preceq \vec\eigen_1\vec\eigen_1^\top.
\]
Then, for every $v \in V\setminus F$, we have
\begin{equation}
\label{eq:lower-eigen1}
(1 - \alpha_F)^t \underline{\eigen}_{1v}^2 - 2\lambda^t \le \eigen_{1v}^2.
\end{equation}
Using part (b) of \lemmaref{lemma:spectral-Q-hat}, we have
\begin{align}
\underline{\eigen}_{1v}^2 &= \frac{\degset{\honest}{v}}{\sum_{u\in \honest} \degset{\honest}{u}}. \notag\\
\intertext{Note that $d(1-\alpha_F) \le \degset{\honest}{u} \le d$ for every $u\in\honest$ by definition of $\alpha_F$, which implies} 
\underline{\eigen}_{1v}^2 &\ge \frac{d(1-\alpha_F)}{(n-f)d} \notag \\
                          &\ge \frac{1 - \alpha_F}{n-f}. \label{eq:eigen-underline}
\end{align}
Combining \equationref{eq:lower-eigen1} and \equationref{eq:eigen-underline}, we have
\[
\frac{(1 - \alpha_F)^{t+1}}{n-f} - 2\lambda^t \le \eigen_{1v}^2.
\]
After rearranging, the above becomes 
\[
(1 - \alpha_F)^t\left(\frac{1 - \alpha_F}{n-f} - 2\left(\frac{\lambda}{1  - \alpha_F}\right)^t\right) \le \eigen_{1v}^2.
\]
Recall that $T = \left\lceil\log_{\frac{\lambda}{1 - \alpha_F}}\left(\frac{1 - \alpha_F}{4(n-f)}\right)\right\rceil$. Then, $\left(\frac{\lambda}{1  - \alpha_F}\right)^T = \frac{1 - \alpha_F}{4(n-f)}$. Then, by plugging $t = T$ in the above we get
\[
(1 - \alpha_F)^T \frac{1 - \alpha_F}{2(n-f)} \le \eigen_{1v}^2.
\]
This gives the desired lower bound. 

\paragraph*{(ii) Upper bound on coordinates of $\vec\eigen_1$.} We now show an upper bound on the coordinates of $\vec\eigen_1$. The proof is similar to the proof of the lower bound. Using \lemmaref{lemma:matrix-powers-first-eigenvector} for $\overline{\vec Q}$ with part (c) of \lemmaref{lemma:spectral-Q-hat}, we get
\begin{equation}
    \label{eq:sandwich-Q-overline}
     \overline{\vec Q}^t \preceq \lambda_1\left(\overline{\vec Q}\right)^t \underline{\vec\eigen}_1\underline{\vec\eigen}_1^\top + \left(\frac{\lambda}{1 - \alpha_F}\right)^t\vec J_{n-f}.
\end{equation}
Combining the right inequality of \equationref{eq:chain-ineq-Q} with \equationref{eq:sandwich-Q} and \equationref{eq:sandwich-Q-overline}, we get
\[
\lambda_1(\vec Q)^t \vec\eigen_1\vec\eigen_1^\top - \lambda^t\vec J_{n-f} \preceq \vec Q^t \preceq \overline{\vec Q}^t \preceq  \lambda_1\left(\overline{\vec Q}\right)^t \overline{\vec\eigen}_1\overline{\vec\eigen}_1^\top + \left(\frac{\lambda}{1 - \alpha_F}\right)^t\vec J_{n-f}.
\]
By part (a) of \lemmaref{lemma:spectral-Q}, we know $\lambda_1(\vec Q) \ge 1 - \alpha_F$ and by part (a) of \lemmaref{lemma:spectral-Q-hat}, we know $\lambda_1\left(\overline{\vec Q}\right) = 1$. Therefore, the above implies
\[
(1 - \alpha_F)^t \vec\eigen_1\vec\eigen_1^\top - \lambda^t\vec J_{n-f} \preceq\overline{\vec\eigen}_1\overline{\vec\eigen}_1^\top + \left(\frac{\lambda}{1 - \alpha_F}\right)^t\vec J_{n-f}.
\]
Hence,
\begin{align*}
(1 - \alpha_F)^t \vec\eigen_1\vec\eigen_1^\top &\preceq \overline{\vec\eigen}_1\overline{\vec\eigen}_1^\top + \left(\frac{\lambda}{1 - \alpha_F}\right)^t\vec J_{n-f} + \lambda^t\vec J_{n-f}\\ 
&\preceq \overline{\vec\eigen}_1\overline{\vec\eigen}_1^\top + 2\left(\frac{\lambda}{1 - \alpha_F}\right)^t\vec J_{n-f}.
\end{align*}
Then, for every $v \in V\setminus F$, we have
\[
(1 - \alpha_F)^t \eigen_{1v}^2 \le \overline{\eigen}_{1v}^2 + 2\left(\frac{\lambda}{1 - \alpha_F}\right)^t,
\]
hence,
\begin{equation}
\label{eq:upper-eigen1}
 \eigen_{1v}^2 \le (1 - \alpha_F)^{-t} \left( \overline{\eigen}_{1v}^2 + 2\left(\frac{\lambda}{1 - \alpha_F}\right)^t\right).
\end{equation}
From part (a) of \lemmaref{lemma:spectral-Q-hat}, we have
\begin{align}
\overline{\eigen}_{1v}^2 &= \frac{\degset{\honest}{v}}{\sum_{u\in \honest} \degset{\honest}{u}}. \notag\\
\intertext{Note that $d(1-\alpha_F)\degset{\honest}{u} \le d$ for every $u\in\honest$ by definition of $\alpha_F$, which implies} 
\overline{\eigen}_{1v}^2 &\le \frac{d}{(n-f)d(1-\alpha_F)} \notag \\
                         &\le \frac{1}{(n-f)(1-\alpha_F)}. \label{eq:eigen-overline}
\end{align}
Combining \equationref{eq:upper-eigen1} and \equationref{eq:eigen-overline}, we get
\[
\eigen_{1v}^2 \le (1 - \alpha_F)^{-t} \left( \frac{1}{(n-f)(1-\alpha_F)} + 2\left(\frac{\lambda}{1 - \alpha_F}\right)^t\right).
\]
Recall that $T = \left\lceil\log_{\frac{\lambda}{1 - \alpha_F}}\left(\frac{1 - \alpha_F}{4(n-f)}\right)\right\rceil$. Then, $\left(\frac{\lambda}{1  - \alpha_F}\right)^T = \frac{1 - \alpha_F}{4(n-f)}$. Then, by plugging in $t = T$ in the above we get
\begin{align*}
\eigen_{1v}^2 &\le (1 - \alpha_F)^{-T} \left( \frac{1}{(n-f)(1-\alpha_F)} + \frac{1 - \alpha_F}{2(n-f)}\right)\\
               &\le (1 - \alpha_F)^{-T}  \frac{2}{(n-f)(1-\alpha_F)}.
\end{align*}
This gives the desired upper bound.
\end{proof}

We now provide an upper and a lower bound for the matrix $(\vec I_{n-f} - (1-\rho)\vec Q)^{-1}$ from \lemmaref{lemma:reduction-rw}.
\begin{restatable}{lemma}{QPowersUpperBound}
\label{lemma:Q-powers-upper-bound}
\label{lemma:Q-powers-lower-bound}
\label{lemma:Q-powers-bound}
Consider a $(d,\lambda)$-expander graph $G = (V,E)$ of size $n$ and a set $F \subseteq V$ of size $f$ and adversarial density $\alpha_F$. Let $\vec Q$ and $\vec R$ be as in \eqref{eq:block-form-absorbing-markov}. Consider an arbitrary $\rho \in [0,1]$. Then, the following asserions hold true
\begin{enumerate}[(a)]
    \item $(\vec I_{n-f} - (1-\rho)\vec Q)^{-1} \preceq  \vec I_{n-f} + (1-\lambda)^{-1}\left(\frac{2n (1 - \alpha_F)^{-\Tilde{T}}}{(\rho(n-f) + f)(n-f)}  + \lambda\right) \vec J_{n-f}$.
    \item $(\vec I_{n-f} - (1-\rho)\vec Q)^{-1} \succeq \vec I_{n-f} + \frac{n(1-\alpha_F)^{\Tilde{T}}(1-\rho)^{\Tilde{T}}}{8(\rho(n-f) + f)(n-f)} \vec J_{n-f}$.
\end{enumerate}
where $\Tilde{T} = \left\lceil\log_{\frac{\lambda}{1 - \alpha_F}}\left(\frac{1 - \alpha_F}{4(n-f)}\right)\right\rceil \left(\log_{\frac{\lambda}{1 - \alpha_F}}(1 - \alpha_F) + 2\right) + 2$.
\end{restatable}
\begin{proof}
First, we establish the upper bound (a), and then the lower bound (b).
\paragraph*{(i) Assertion (a).} First, we show an upper bound. Let $\vec\eigen_1$ be the $\ell_2$-normalized first eigenvector of $\vec Q$ as given in \lemmaref{lemma:Q-delocalization}. Then, by \lemmaref{lemma:matrix-powers-first-eigenvector} we have for any $t \ge 0$
\begin{align}
\vec Q^t &\preceq \lambda_1(\vec Q)^t \vec \eigen_{1} \vec \eigen_{1}^\top + \max\{\left|\lambda_2(\vec Q)\right|, \left|\lambda_{n - f}(\vec Q)\right|\}^t \vec J_{n-f}. \notag \\
\intertext{By part (b) of \lemmaref{lemma:spectral-Q} , $\max\{\left|\lambda_2(\vec Q)\right|, \left|\lambda_{n - f}(\vec Q)\right|\} \le \lambda$, hence,}
\vec Q^t &\preceq \lambda_1(\vec Q)^t \vec \eigen_{1} \vec \eigen_{1}^\top + \lambda^t\vec J_{n-f}. \label{eq:qt-upper}
\end{align}
Then, by \equationref{eq:Q-series}, we get 
\begin{align}
    (\vec I_{n-f} - (1-\rho)\vec Q)^{-1} 
    &=\sum_{t = 0}^{\infty} (1-\rho)^t\vec Q^t \notag\\
    &\preceq \vec I_{n-f} + \sum_{t = 1}^{\infty} (1-\rho)^t \vec Q^t. \notag
    \intertext{From \equationref{eq:qt-upper}, we have}
    (\vec I_{n-f} - (1-\rho)\vec Q)^{-1}  &\preceq \vec I_{n-f} + \sum_{t = 1}^{\infty} (1- \rho)^t\left(\lambda_1(\vec Q)^t \vec \eigen_1 \vec \eigen_1^\top + \lambda^t \vec J_{n-f}\right) \notag\\
    &\preceq \vec I_{n-f} + \vec \eigen_1 \vec \eigen_1^\top  \sum_{t = 1}^{\infty} (1- \rho)^t \lambda_1(\vec Q)^t + \vec J_{n-f} \sum_{t = 1}^{\infty} (1-\rho)^t \lambda^t. \notag
    \intertext{By reducing the infinite geometric series, we get}
    (\vec I_{n-f} - (1-\rho)\vec Q)^{-1} &\preceq \vec I_{n-f} + \frac{(1-\rho)\lambda_1(\vec Q)}{1 - (1-\rho)\lambda_1(\vec Q)} \vec \eigen_1 \vec \eigen_1^\top + \frac{\lambda(1-\rho)}{1 - \lambda(1-\rho)} \vec J_{n-f}. \notag
    \intertext{Since $\lambda_1(\vec Q) \le 1$ and $1 - \rho \le 1$, we have}
    (\vec I_{n-f} - (1-\rho)\vec Q)^{-1} &\preceq \vec I_{n-f} + \frac{1}{1 - (1-\rho)\lambda_1(\vec Q)} \vec \eigen_1 \vec \eigen_1^\top + \frac{\lambda}{1 - \lambda} \vec J_{n-f}. \notag
    \intertext{We know $\lambda_1(\vec Q) \le 1 - (1 - \lambda)\frac{f}{n}$ by part (a) of \lemmaref{lemma:spectral-Q}, hence}
    (\vec I_{n-f} - (1-\rho)\vec Q)^{-1} &\preceq \vec I_{n-f} + \left(\frac{1}{1 - (1-\rho)(1 - (1-\lambda)f/n)} \vec \eigen_1 \vec \eigen_1^\top + \frac{\lambda}{1 - \lambda}\right) \vec J_{n-f} \notag\\
    &= \vec I_{n-f} + \left(\frac{1}{\rho(1 - (1-\lambda)f/n) + (1-\lambda)f/n} \vec \eigen_1 \vec \eigen_1^\top + \frac{\lambda}{1 - \lambda}\right) \vec J_{n-f}\notag\\
    &\preceq \vec I_{n-f} + \left(\frac{1}{(1-\lambda)(\rho(1 - f/n) + f/n)} \vec \eigen_1 \vec \eigen_1^\top + \frac{\lambda}{1 - \lambda}\right) \vec J_{n-f}\notag\\
    &\preceq \vec I_{n-f} + \frac{1}{1-\lambda}\left(\frac{n}{\rho(n - f) + f} \vec \eigen_1 \vec \eigen_1^\top + \lambda\right) \vec J_{n-f}. \label{eq:series-qt-upper}
\end{align}
Set $T_1 = \left\lceil\log_{\frac{\lambda}{1 - \alpha_F}}\left(\frac{1 - \alpha_F}{4(n-f)}\right)\right\rceil$. Recall that, from \lemmaref{lemma:Q-delocalization} we have
\[
\vec \eigen_1 \vec \eigen_1^\top \preceq  \frac{2 (1 - \alpha_F)^{-T_1 - 1}}{(n-f)}  \vec J_{n-f}.
\]
Substituting the above into \equationref{eq:series-qt-upper}, we get 
\begin{align*}
    (\vec I_{n-f} - (1-\rho)\vec Q)^{-1}
    &\preceq \vec I_{n-f} + (1-\lambda)^{-1}\left(\frac{2n (1 - \alpha_F)^{-T_1 - 1}}{(\rho(n-f) + f)(n-f)}  + \lambda\right) \vec J_{n-f}.
    \intertext{Since $T_1 + 1 \le \Tilde{T} = T_1\left(\log_{\frac{\lambda}{1 - \alpha_F}}(1 - \alpha_F) + 2\right) + 2$, we finally get}
    (\vec I_{n-f} - (1-\rho)\vec Q)^{-1}
    &\preceq \vec I_{n-f} + (1-\lambda)^{-1}\left(\frac{2n (1 - \alpha_F)^{-\Tilde{T}}}{(\rho(n-f) + f)(n-f)}  + \lambda\right) \vec J_{n-f},
\end{align*}
which concludes the proof of assertion (a).

\paragraph*{(ii) Assertion (b).} Now, we derive a lower bound in a similar way. Note that by \lemmaref{lemma:matrix-powers-first-eigenvector} we have for any $t \ge 0$
\begin{align}
\vec Q^t &\succeq \lambda_1(\vec Q)^t \vec \eigen_{1} \vec \eigen_{1}^\top - \max\{\left|\lambda_2(\vec Q)\right|, \left|\lambda_{n - f}(\vec Q)\right|\}^t \vec J_{n-f}. \notag \\
\intertext{By part (b) of \lemmaref{lemma:spectral-Q}, $\max\{\left|\lambda_2(\vec Q)\right|, \left|\lambda_{n - f}(\vec Q)\right|\} \le \lambda$, hence,}
\vec Q^t &\succeq \lambda_1(\vec Q)^t \vec \eigen_{1} \vec \eigen_{1}^\top - \lambda^t\vec J_{n-f}. \label{eq:qt-lower}
\end{align}

Note that all entries of $\vec Q$ are non-negative, hence entries of $\vec Q^t$ will be non-negative for any $t \ge 1$. By \equationref{eq:Q-series}, we get for any $T_2 \ge 1$
\begin{align}
     (\vec I_{n-f} - (1-\rho) \vec Q)^{-1} 
     &=\sum_{t = 0}^{\infty} \vec (1-\rho)^t \vec Q^t \notag\\ 
     &=\vec Q^0 + \sum_{t = 1}^{T_2 - 1} (1 - \rho)^t\vec Q^t + \sum_{t=T_2}^\infty \vec Q^t \vec (1-\rho)^t \vec Q^t \notag\\ 
     &\succeq \vec I_{n-f} + \sum_{t = T_2}^{\infty} (1-\rho)^t \vec Q^t. \notag
     \intertext{From \equationref{eq:qt-lower}, we have}
     (\vec I_{n-f} - (1-\rho) \vec Q)^{-1}  
     &\succeq \vec I_{n-f} + \sum_{t = T_2}^{\infty} (1-\rho)^t\left(\lambda_1(\vec Q)^t \vec \eigen_1 \vec \eigen_1^\top - \lambda^t \vec J_{n-f}\right). \label{eq:series-qt-lower}
\end{align}
Recall that, by \lemmaref{lemma:Q-delocalization} we have
\[
\vec \eigen_1 \vec \eigen_1^\top \succeq  \frac{(1 - \alpha_F)^{T_1 + 1}}{2(n-f)}  \vec J_{n-f}.
\]
Substituting this into \equationref{eq:series-qt-lower}, we get
\begin{align}
     (\vec I_{n-f} - (1-\rho) \vec Q)^{-1} &\succeq \vec I_{n-f} + \sum_{t = T_2}^{\infty} (1-\rho)^t\left(\lambda_1(\vec Q)^t \vec \eigen_1 \vec \eigen_1^\top - \lambda^t \vec J_{n-f}\right) \notag\\
     &\succeq \vec I_{n-f} + \sum_{t = T_2}^{\infty} (1-\rho)^t\left(\lambda_1(\vec Q)^t \frac{(1 - \alpha_F)^{T_1 + 1}}{2(n-f)} - \lambda^t\right) \vec J_{n-f}. \label{eq:series-qt-lower-eigen}
\end{align}
Set $T_2 = \left\lceil\log_{\frac{\lambda}{1 - \alpha_F}} \left( \frac{(1-\alpha_F)^{T_1 + 1}}{4(n-f)} \right)\right\rceil$. Then, for any $t \ge T_2$ we have 
\begin{align}
\lambda_1(\vec Q)^t \frac{(1 - \alpha_F)^{T_1 + 1}}{2(n-f)} - \lambda^t 
&= \lambda_1(\vec Q)^t \left(\frac{(1 - \alpha_F)^{T_1 + 1}}{2(n-f)} - \left(\frac{\lambda}{\lambda_1(\vec Q)}\right)^t\right). \notag \\
\intertext{By part (a) \lemmaref{lemma:spectral-Q}, $\lambda_1(\vec Q) \ge 1 - \alpha_F$, hence}
\lambda_1(\vec Q)^t \frac{(1 - \alpha_F)^{T_1 + 1}}{2(n-f)} - \lambda^t 
&\ge \lambda_1(\vec Q)^t \left(\frac{(1 - \alpha_F)^{T_1 + 1}}{2(n-f)} - \left(\frac{\lambda}{1-\alpha_F}\right)^t\right). \notag
\intertext{Since $t \ge T_2$ and $1 - \alpha_F > \lambda$, we get}
\lambda_1(\vec Q)^t \frac{(1 - \alpha_F)^{T_1 + 1}}{2(n-f)} - \lambda^t 
&\ge \lambda_1(\vec Q)^t \left(\frac{(1 - \alpha_F)^{T_1 + 1}}{2(n-f)} - \left(\frac{\lambda}{1-\alpha_F}\right)^{T_2}\right). \notag
\intertext{With $T_2$ defined as above, we get}
\lambda_1(\vec Q)^t \frac{(1 - \alpha_F)^{T_1 + 1}}{2(n-f)} - \lambda^t 
&\ge \lambda_1(\vec Q)^t \left(\frac{(1 - \alpha_F)^{T_1 + 1}}{2(n-f)} - \frac{(1-\alpha_F)^{T_1 + 1}}{4(n-f)}\right) \notag \\
&=\lambda_1(\vec Q)^t \frac{(1 - \alpha_F)^{T_1 + 1}}{4(n-f)}. \label{eq:series-qt-subexpession}
\end{align}
Substituting \equationref{eq:series-qt-subexpession} into \equationref{eq:series-qt-lower-eigen}, we get
\begin{align}
    (\vec I_{n-f} - (1-\rho) \vec Q)^{-1}
   &\succeq \vec I_{n-f} + \sum_{t = T_2}^{\infty} (1-\rho)^t\lambda_1(\vec Q)^t \frac{(1 - \alpha_F)^{T_1 + 1}}{4(n-f)} \vec J_{n-f}.\notag
    \intertext{By reducing the infinite geometric series, we get}
    (\vec I_{n-f} - (1-\rho) \vec Q)^{-1}
    &\succeq \vec I_{n-f} + \frac{(1-\rho)^{T_2}\lambda_1(\vec Q)^{T_2}}{1 - (1-\rho)\lambda_1(\vec Q)} \cdot \frac{(1 - \alpha_F)^{T_1 + 1}}{4(n-f)} \vec J_{n-f}.\notag
    \intertext{From part (a) of \lemmaref{lemma:spectral-Q}, we have $\lambda_1(\vec Q) \ge 1 - \alpha_F$, hence} 
    (\vec I_{n-f} - (1-\rho) \vec Q)^{-1}
    &\succeq \vec I_{n-f} + \frac{1}{1 - (1-\rho)\lambda_1(\vec Q)} \cdot \frac{(1 - \alpha_F)^{T_2 + T_1 + 1} (1-\rho)^{T_2}}{4(n-f)} \vec J_{n-f}. \notag
    \intertext{From part (a) of \lemmaref{lemma:spectral-Q}, we also have $\lambda_1(\vec Q) \ge 1 - (1 + \lambda)f/n$, hence} 
    (\vec I_{n-f} - (1-\rho) \vec Q)^{-1}
    &\succeq \vec I_{n-f} + \frac{1}{1 - (1-\rho)(1 - (1 + \lambda)f/n)} \cdot \frac{(1 - \alpha_F)^{T_2 + T_1 + 1} (1-\rho)^{T_2}}{4(n-f)} \vec J_{n-f}\notag\\
    &= \vec I_{n-f} + \frac{1}{(1-\rho)(1 + \lambda)f/n + \rho} \cdot \frac{(1 - \alpha_F)^{T_2 + T_1 + 1} (1-\rho)^{T_2}}{4(n-f)} \vec J_{n-f}. \notag
    \intertext{As $\lambda + 1 \ge 1$, we have}
    (\vec I_{n-f} - (1-\rho) \vec Q)^{-1}
    &\succeq \vec I_{n-f} + \frac{1}{(1 + \lambda)((1-\rho)f/n + \rho)} \cdot \frac{(1 - \alpha_F)^{T_2 + T_1 + 1} (1-\rho)^{T_2}}{4(n-f)} \vec J_{n-f}. \notag
    \intertext{As $\lambda + 1 \le 2$, we have}
    (\vec I_{n-f} - (1-\rho) \vec Q)^{-1}
    &\succeq \vec I_{n-f} +  \frac{n(1-\alpha_F)^{T_2 + T_1 + 1}(1-\rho)^{T_1}}{8(\rho(n-f) + f)(n-f)} \vec J_{n-f}. \label{eq:series-qt-almost}
\end{align}
Also, note that, by definition of $T_1$ and $T_2$, we have
\begin{align*}
T_2 + T_1 + 1 &= \left\lceil\log_{\frac{\lambda}{1 - \alpha_F}} \left( \frac{(1-\alpha_F)^{T_1 + 1}}{4(n-f)} \right)\right\rceil + T_1 + 1\\
              &\le \log_{\frac{\lambda}{1 - \alpha_F}} \left( \frac{(1-\alpha_F)^{T_1 + 1}}{4(n-f)} \right) + T_1 + 2 \\
              &= T_1\log_{\frac{\lambda}{1 - \alpha_F}}\left(1 - \alpha_F\right) + \log_{\frac{\lambda}{1 - \alpha_F}}\left(\frac{1 - \alpha_F}{4(n-f)}\right) + T_1 + 2\\
              &= T_1\log_{\frac{\lambda}{1 - \alpha_F}}\left(1 - \alpha_F\right) + 2T_1 + 2\\
              &= \Tilde{T}.
\end{align*}
Combining this and $\Tilde{T} \ge T_1$ with \equationref{eq:series-qt-almost}, we finally get
\[
(\vec I_{n-f} - (1-\rho) \vec Q)^{-1} \succeq \vec I_{n-f} +  \frac{n(1-\alpha_F)^{\Tilde{T}}(1-\rho)^{\Tilde{T}}}{8(\rho(n-f) + f)(n-f)} \vec J_{n-f},
\]
which concludes the proof of assertion (b).
\end{proof}

By putting everything together, we obtain the desired bound on the maximal divergence.
\begin{restatable}{lemma}{BoundDInftyJoint}
\label{lemma:bound-d-infty-joint}
Consider an undirected connected $(d,\lambda)$-expander graph $G=(V,E)$ of size $n$, a set of curious nodes $F \subset V$ with $|F| = f$ and adversarial density $\alpha_F < 1-\lambda$. Consider a $(1+\rho)$-cobra walk (or $\rho$-Dandelion) with $\rho < 1$. Then, for every $v,u \in \honest$, the following holds true
\[
\divinfty{\seqadv{v}}{\seqadv{u}} \le \ln(\rho(n-f) + f) - 2\Tilde{T}\ln(1 - \alpha_F) - \Tilde{T}\ln(1 - \rho) - \ln(1-\lambda) + \ln(24) ,
\]
where $\Tilde{T}  = \left\lceil\log_{\frac{\lambda}{1 - \alpha_F}}\left(\frac{1 - \alpha_F}{4(n-f)}\right)\right\rceil \left(\log_{\frac{\lambda}{1 - \alpha_F}}(1 - \alpha_F) + 2\right) + 2$.
\end{restatable}
\begin{proof}
From part (a) of \lemmaref{lemma:Q-powers-upper-bound}, for any $v,w \in \honest$, we have 
\begin{align}
    (\vec I_{n-f} - (1-\rho)\vec Q)^{-1}_{vw} 
    &\le 1 + (1-\lambda)^{-1}\left(\frac{2n (1 - \alpha_F)^{-\Tilde{T}}}{(\rho(n-f) + f)(n-f)}  + \lambda\right)\notag\\
    &= (1-\lambda)^{-1}\left(\frac{2n (1 - \alpha_F)^{-\Tilde{T}}}{(\rho(n-f) + f)(n-f)}  + 1\right).\label{eq:upper-vw}
\end{align}
Also, from part (b) of \lemmaref{lemma:Q-powers-upper-bound}, for any $u,w \in \honest$, we have 
\begin{align}
    (\vec I_{n-f} - (1-\rho)\vec Q)^{-1}_{uw} \ge \frac{n(1-\alpha_F)^{\Tilde{T}}(1-\rho)^{\Tilde{T}}}{8(\rho(n-f) + f)(n-f)}. \label{eq:lower-uw}
\end{align}
Then, by \lemmaref{lemma:sufficient-varepsilon-0} and \lemmaref{lemma:dandelion-reduction} we have for both $\rho$-Dandelion and $(1 + \rho)$-cobra walk. 
\begin{align*}
    \divinfty{\seqadv{v}}{\seqadv{u}}
    &\le \max_{w \in \honest}\ln\frac{(\vec I_{n-f} - (1 - \rho)\vec Q)^{-1}_{vw}}{(\vec I_{n-f} - (1 - \rho)\vec Q)^{-1}_{uw}}.
    \intertext{By substituting \equationref{eq:upper-vw} and \equationref{eq:lower-uw}, we get}
    \divinfty{\seqadv{v}}{\seqadv{u}}&\le \ln\left(\frac{(1-\lambda)^{-1}\left(\frac{2n (1 - \alpha_F)^{-\Tilde{T}}}{(\rho(n-f) + f)(n-f)}  + 1\right)}{\frac{n(1-\alpha_F)^{\Tilde{T}}(1-\rho)^{\Tilde{T}}}{8(\rho(n-f) + f)(n-f)}}\right).
    \intertext{By rearranging the terms, we get}
    \divinfty{\seqadv{v}}{\seqadv{u}} &\le \ln\left(16(1 - \rho)^{-\Tilde{T}} + \frac{8(\rho(n-f) + f)(n-f)}{n(1-\alpha_F)^{\Tilde{T}}(1-\rho)^{\Tilde{T}}}\right) - \ln(1 - \lambda) .
    \intertext{Since $n-f\le n$, we have}
    \divinfty{\seqadv{v}}{\seqadv{u}}
    &\le \ln\left(16(1 - \rho)^{-\Tilde{T}} + \frac{8(\rho(n-f) + f)}{(1-\alpha_F)^{\Tilde{T}}(1-\rho)^{\Tilde{T}}}\right) - \ln(1 - \lambda) \\
    &=\ln\left(16(1 - \rho)^{-\Tilde{T}} + 8(\rho(n-f) + f))(1-\alpha_F)^{-\Tilde{T}}(1-\rho)^{-\Tilde{T}}\right) - \ln(1 - \lambda).
    \intertext{Since $\rho(n-f) + f \ge 1$ and $(1-\alpha_F)^{-\Tilde{T}} \ge 1$, have }
    \divinfty{\seqadv{v}}{\seqadv{u}}&\le \ln\left(24(\rho(n-f) + f))(1-\alpha_F)^{-\Tilde{T}}(1-\rho)^{-\Tilde{T}}\right) - \ln(1 - \lambda) \\
    &\le \ln(\rho(n-f) + f) - \Tilde{T}\ln(1-\alpha_F) - \Tilde{T}\ln(1-\rho) - \ln(1 - \lambda) + \ln(24),
\end{align*}
which concludes the proof.
\end{proof}

\subsection{Proof of \theoremref{thm:main}}

\MainThm*
\begin{proof}
    First, put $\alpha = f/d$ and consider the worst-case adversary. The adversarial density $\alpha_F$ is upper bounded by $\alpha = f/d$ since the neighbourhood of every node has size $d$ and contains at most $f$ curious nodes. Then $\mathcal P$ is $\varepsilon$-DP with $\alpha = f/d$ against the worst-case adversary by \lemmaref{lemma:bound-d-infty-joint}.
    
    Now, put $\alpha$ as in \lemmaref{lemma:adversarial-density} and consider the average-case adversary. For the average-case adversary, by \lemmaref{lemma:adversarial-density}, $\alpha_F$ is upper bounded by $\alpha$ with high probability. Therefore $\mathcal P$ is $\varepsilon$-DP with $\alpha$ as in \lemmaref{lemma:adversarial-density} against the average-case adversary by \lemmaref{lemma:bound-d-infty-joint} and \equationref{eq:average-guarantee-probability}.
\end{proof}

\section{Dissemination time vs. privacy trade-off: proofs of Section~\ref{sec:analysis}}
\label{sec:appendix-tradeoff}

In this part of the appendix, we present detailed proofs of theorems from Section~\ref{sec:analysis}.

\subsection{Privacy guarantees on near-Ramanujan graphs (Proof of \corollaryref{corollary:privacy-nr})}\label{subsec:appendix-near-Ramunujan}

The proof of \corollaryref{corollary:privacy-nr} relies on the two following observations. First, we upper bound the value of $\alpha$ in \lemmaref{lemma:adversarial-density} for dense near-Ramanujan graphs.
\begin{lemma}
\label{lemma:alpha-nr}
Let $\mathcal G$ be a family of $d$-regular near-Ramanujan graphs with $n$ nodes and $d \in n^{\Omega_n(1)}$. Suppose $f/n \in 1 - \Omega_n(1)$. Let $G \in \mathcal G$ and let $\alpha$ be as in \lemmaref{lemma:adversarial-density}. Then
\[
\alpha \in 1 - \Omega_n(1).
\]
\end{lemma}
\begin{proof}
Set $\beta = f/n$ and $\gamma = \ln(n)/(ed)$, as in \lemmaref{lemma:adversarial-density}. Note that since $d \in n^{\Omega_n(1)}$, we have $\gamma \in o_n(1)$. Select $n$ large enough so that $\gamma < 1/(8e)$. We consider two cases: when $\beta < 1/(8e)$ and when $\beta \ge 1/(8e)$. In the first case, by the first part of \lemmaref{lemma:adversarial-density}, we have
\begin{equation}
   \label{eq:alpha-ub-nr-1} 
\alpha \le 4e\frac{\max\{\gamma, \beta\}}{1 + \max \{\ln(\gamma) - \ln(\beta), 0\}} \le 4e\max\{\gamma, \beta\} < 1/2,
\end{equation}
hence, $\alpha \in 1 - \Omega_n(1)$ in this case. Now, consider the case when $\beta \ge 1/(8e)$. Set $c = 1/(8e)$ and $\delta = \frac{n/f - 1}{2} \in O_n(1)$. Then, $\beta \ge c$. Note that since $d \in n^{\Omega_n(1)}$, we have $d \in \omega_n(\log(n))$. Select $n$ large enough so that $d > 64e^2\ln(n)/\delta^2$, i.e., $d > \frac{\ln(n)}{c^2\delta^2}$. Then, by second part of \lemmaref{lemma:adversarial-density}, we have
\begin{equation}
\label{eq:alpha-ub-nr-2}
\alpha \le (1 + \delta)\beta =  \left(\frac{n/f - 1}{2} + 1\right) \cdot \frac{f}{n} = \frac{f/n + 1}{2} \in 1 - \Omega_n(1),
\end{equation}
where the last transition follows from the fact that $f/n \in 1 -\Omega_n(1)$. Combining \equationref{eq:alpha-ub-nr-1} and \equationref{eq:alpha-ub-nr-2} concludes the proof.
\end{proof}

Now, we bound from above the value of $\Tilde{T}$ from \theoremref{thm:main}.
\begin{lemma}
\label{lemma:t-tilde-nr}
Let $\mathcal G$ be a family of $d$-regular near-Ramanujan graphs with $n$ nodes and $d \in n^{\Omega_n(1)}$. Suppose $f/d \in 1 - \Omega_n(1)$ (resp. $f/n \in 1 - \Omega_n(1)$). Let $G \in \mathcal G$ and let $\Tilde{T}$ be as in \theoremref{thm:main}. Then, both against a worst-case and an average-case adversary, we have
\[
\Tilde{T} \in O_n(1).
\]
\end{lemma}
\begin{proof}
By definition of near-Ramanujan graphs, we have $\lambda \in O_n\left(d^{-1/2}\right)$. Since $d \in n^{\Omega_n(1)}$, this yields 
\begin{equation}
\lambda \in n^{-\Omega_n(1)}. 
\label{eq:lambda-nr}\end{equation} 
Also, recall that
\[
\Tilde{T} = \left\lceil\log_{\frac{\lambda}{1 - \alpha}}\left(\frac{1 - \alpha}{4(n-f)}\right)\right\rceil \left(\log_{\frac{\lambda}{1 - \alpha}}(1 - \alpha) + 2\right) + 2.
\]
Note that against a worst-case adversary we have $\alpha = f/d \in 1 - \Omega_n(1)$. We also know that $\alpha \in 1 - \Omega_n(1)$ against an average-case adversary, as per \lemmaref{lemma:alpha-nr}. Then, regardless of the type of the adversary, we have
\begin{equation}
\label{eq:alpha-bounded-away-nr}
\alpha \in 1 -\Omega_n(1).
\end{equation}
From \equationref{eq:lambda-nr} and (\ref{eq:alpha-bounded-away-nr}), we get $\frac{\lambda}{1 - \alpha} \in n^{-\Omega_n(1)}$ and $\frac{1 - \alpha}{4(n-f)} \in n^{-O_n(1)}$. Hence, the following holds true
\begin{equation}
\label{eq:tilde-t-expression-ub}
\log_{\frac{\lambda}{1 - \alpha}}\left(\frac{1 - \alpha}{4(n-f)}\right) = \frac{\ln\left(\frac{1 - \alpha}{4(n-f)}\right)}{\ln\left(\frac{\lambda}{1 - \alpha}\right)} \in O_n(1).
\end{equation}
Additionally, since $\alpha \in 1 - \Omega_n(1)$ as per \equationref{eq:alpha-bounded-away-nr}, using \equationref{eq:lambda-nr} we get $\log_{\frac{\lambda}{1 - \alpha}}(1 - \alpha) \in O_n(1)$. Combining this with \equationref{eq:tilde-t-expression-ub}, we get
\[
\Tilde{T} = \left\lceil\log_{\frac{\lambda}{1 - \alpha}}\left(\frac{1 - \alpha}{4(n-f)}\right)\right\rceil \left(\log_{\frac{\lambda}{1 - \alpha}}(1 - \alpha) + 2\right) + 2\in O_n(1),
\]
which concludes the proof.
\end{proof}

\PrivacyNRWorst*
\begin{proof}
Recall, that by \theoremref{thm:main}
, $\mathcal P$ satisfies $\varepsilon$-DP with 
\[
\varepsilon = \ln(\rho(n-f) + f) - 2\Tilde{T}\ln(1 - \alpha) - \Tilde{T}\ln(1 - \rho) - \ln(1-\lambda) + \ln(24),
\]
and $\Tilde{T}  = \left\lceil\log_{\frac{\lambda}{1 - \alpha}}\left(\frac{1 - \alpha}{4(n-f)}\right)\right\rceil \left(\log_{\frac{\lambda}{1 - \alpha}}(1 - \alpha) + 2\right) + 2.$ Note that for both the worst and the average-case adversary, we have $\alpha \in 1 - \Omega_n(1)$ (\lemmaref{lemma:alpha-nr}) and $\Tilde{T} \in O_n(1)$ (\lemmaref{lemma:t-tilde-nr}). Then 
\begin{equation}
\label{eq:t-tilde-alpha}
- 2\Tilde{T}\ln(1 - \alpha) \in O_n(1).
\end{equation}
Since we are given $1 - \rho \in \Omega_n(1)$, we also have
\begin{equation}
\label{eq:t-tilde-rho}
- \Tilde{T}\ln(1 - \rho) \in O_n(1).
\end{equation}
Finally, for dense near-Ramanujan graphs we have $\lambda \in O_n(d^{-1/2}) \subseteq n^{-\Omega_n(1)}$, hence,
\[
- \ln(1-\lambda) \in O_n(1).
\]
Combining this with \equationref{eq:t-tilde-alpha} and \equationref{eq:t-tilde-rho}, for both adversaries we have
\begin{align*}
\varepsilon &= \ln(\rho(n-f) + f) - 2\Tilde{T}\ln(1 - \alpha) - \Tilde{T}\ln(1 - \rho) - \ln(1-\lambda) + \ln(24) \\ 
             &\in \ln(\rho(n-f) + f) + O_n(1),
\end{align*}
which concludes the proof.
\end{proof}

\subsection{Trade-off for cobra walks}
\subsubsection{Proof of the tightness of \corollaryref{corollary:privacy-nr}}
\label{apppendix:tighter-LB-for-cobra}

In this section, we show the following result.
\begin{restatable}{theorem}{CobraPrivacyLower}
\label{thm:cobra-privacy-lower} 
Let $\mathcal P$ be a $(1 + \rho)$-cobra walk with $\rho \in [0,1]$ and let $\mathcal G$ be a family of $d$-regular near-Ramanujan graphs with $n$ nodes, $f$ of which are curious and $d \in n^{\Omega_n(1)}$. Suppose $f \in n^{\Omega_n(1)}$ and $f/n \in 1-\Omega_n(1)$, and $\mathcal P$ satisfies $\varepsilon$-DP against either an average-case or a worst-case adversary on $G \in \mathcal G$. Then, 
\[\varepsilon \in \ln{(\rho (n-f) + f)}+\Omega_n(1).\]
\end{restatable}
This effectively means that \corollaryref{corollary:privacy-nr} is tight for cobra walks. Note that the statement of \theoremref{thm:cobra-privacy-lower} with $\rho = 0$ follows from \theoremref{thm:universal-impossibility}. Without loss of generality, we assume $\rho > 0$ in the remaining.

First, we prove two helpful auxiliary lemmas. 
\begin{lemma}
    \label{lemma:two-curious-prob}
    Let $G = (V,E)$ be an undirected connected graph of size $n \ge 3$. Let $F \sim \uniformsubset{f}{V}$ be a random subset of $V$ of size $f \ge 3$. Then, with probability at least $2/(n-1)$, $F$ contains two nodes connected by an edge.
\end{lemma}
\begin{proof}
   
    Let $w_1, \ldots, w_f$ be nodes of $F$, i.e., $w_1, \ldots, w_f$ are sampled uniformly at random without replacement. Fix an arbitrary node $v \in V$. Since the graph is of size at least $3$ and is connected, $v$ has at least one neighbour. Let $u \in N(v)$ be a fixed neighbour of $v$. 
    Note that 
\begin{align*}
 \Pr[u \in F \mid w_1 = v] &= \Pr[w_2=u \vee w_3 =u \vee...\vee w_f = u  \mid w_1 = v]\\
 \intertext{Since events $w_i = u$ are disjoint, we get}
    \Pr[u \in F \mid w_1 = v] & = \sum_{i = 2}^f \Pr[w_i = u\mid w_1 = v] = (f-1)\frac{1}{n-1} \ge \frac{2}{n-1}.
\end{align*}
    Then, regardless of the value of $w_1$, it has a neighbour in set $F$ with probability $\ge 2/(n-1)$. Then, $F$ contains two nodes connected by an edge with probability $\ge 2/(n-1)$.
    
\end{proof}
\begin{lemma}
    \label{lemma:two-curious-with-neighbour}
    Let $G = (V,E)$ be an undirected connected graph of size $n \ge 3$. Let $F\subseteq V$ be of size $2\le f \le n-1$, such that there are two nodes in $F$ connected by an edge. Then, there exist $v \in \honest$ and $w_1, w_2 \in F$ such that $\{v,w_1\}, \{w_1, w_2\} \in E$.
\end{lemma}
\begin{proof}

    Let $w_1', w_2'$ be two nodes in $F$ connected by an edge. Consider a subgraph $G'$ of $G$ induced on vertices of $F$, and consider a connected component of this graph containing $w_1'$ and $w_2'$. Let $F' \subseteq F$ be vertices of this component. Since $G$ is connected and $2\le|F|\le n -1$, there is an edge in $G$ between a vertex from $F'$ and a vertex from $V\setminus F'$. Let $v \in V\setminus F'$ and $w_1 \in F'$ be these vertices. Note that $v \not \in F$, because otherwise $v$ would belong to a connected component $F'$ of graph $G'$ (i.e., $v \in F'$). Then, we must have $v \in \honest$. Finally, notice that subgraph of $G$ induced on $F'$ is connected, hence, $w_1$ has a neighbour $w_2 \in F'$. Hence, we selected $v,w_1,w_2$ such that $v \in \honest$ and $w_1, w_2 \in F$ such that $\{v,w_1\}, \{w_1, w_2\} \in E$, which concludes the proof.
    
\end{proof}

Now, we give a definition of passage probability, which will help us to bound $D_\infty$ from below.
\begin{definition}
    \label{def:passage}
    Consider an execution of a $(1 + \rho)$-cobra walk with $\rho \in (0,1]$ on an undirected connected graph $G = (V,E)$. Let $F\subseteq V$ be a subset of curious nodes. For $v,u \in\honest$, define the \emph{passage probability} $\passage{u}{v}$ from $u$ to $v$ as the probability of a protocol started from $u$ to reach node $v$ while not interacting with any curious nodes
    \[
    \passage{u}{v} = \Pr\left[\exists t \ge \starttime: \activeset_t = \{v\} \land (\activeset_i \cap F = \emptyset, \forall i \le t) \mid \activeset_{\starttime} = \{u\}\right].
    \]
\end{definition}

We now bound $\divinfty{\seqadv{v}}{\seqadv{u}}$ from below using passage probabilities between $u$ and $v$. Recall that we say cobra walk \emph{branches} if an active node communicates gossip to two of its neighbours. Accordingly, we use the term ``\emph{branch}'' (as a noun) to refer to a sequence of nodes $v_1,v_2,\ldots$ of a cobra walk such that for some $t\ge \starttime$, $v_1$ communicated the gossip to $v_2$ at round $t$, $v_2$ communicated the gossip to $v_3$ at round $t+1$, etc. Note that active nodes of a branch of a cobra walk follow the same law as a simple random walk.
\begin{restatable}{lemma}{NecessaryPassage}
    \label{lemma:privacy-lb-necessary-passage}
        Let $\mathcal P$ be a $(1 + \rho)$-cobra walk where $\rho \in (0,1]$, and consider an undirected connected graph $G = (V,E)$ of size $n \ge 3$. Let $F\subseteq V$ be a set of curious nodes of size $f \ge 1$ such that there exist two curious nodes connected by an edge. Then, there exist $v \in \honest$ such that for any $u \in \honest$ we have
        \[
        \divinfty{\seqadv{v}}{\seqadv{u}} \ge \ln\left(\passage{u}{v}^{-1}\right).
        \]
\end{restatable}
\begin{proof}
   Note that, by \lemmaref{lemma:two-curious-with-neighbour} we can select $v \in \honest$ and $w_1, w_2 \in F$ such that $\{v, w_1\}, \{w_1, w_1\} \in E$. Let $\Pi_{u\Rightarrow v}$ be the event such that 
    \[
    \Pi_{u\Rightarrow v} = \left\{\exists t \ge \starttime: \activeset_t = \{v\} \land (\activeset_i \cap F = \emptyset, \forall i \le t) \mid \activeset_{\starttime} = \{u\}\right\}.
    \]
    Then, by \definitionref{def:passage}
    \begin{equation}
    \label{eq:huv}
    \Pr[\Pi_{u\Rightarrow v}] = \passage{u}{v}.
    \end{equation}
    Recall that $w_1$ and $w_2$ are connected. Then, set $\fakecomm_0 = \{(v\to w_1), (v\to w_1)\}$, $\fakecomm_i = \{(w_1\to w_2), (w_1\to w_2)\}$ for odd $i \ge 1$, and $\fakecomm_i = \{(w_2\to w_1), (w_2\to w_1)\}$ for even $i \ge 1$. Then, define $\sigma_{v,t}$ to be a set of adversarial observations which begins with a prefix $\left(\fakecomm_0, \fakecomm_1, \ldots, \fakecomm_t\right)$. In other words, in $\sigma_{v,t}$, the first contact with the curious nodes is $v$ contacting $w_1$ twice in the same round ($\fakecomm_0$), and then, during the following $t$ steps, $w_1$ and $w_2$ communicate to each other twice per round in an alternating manner ($\fakecomm_i$). Moreover, no other communication involving curious nodes happens in these $t$ rounds. Then, by the law of total probability, we can write
    \[
    \Pr[\seqadv{u} \in \sigma_{v,t}] = \Pr[\seqadv{u} \in \sigma_{v,t} \mid \Pi_{u\Rightarrow v}] \Pr[\Pi_{u\Rightarrow v}] + \Pr[\seqadv{u} \in \sigma_{v,t} \land \neg \Pi_{u\Rightarrow v}].
    \]
    Note that $\seqadv{v}$ and $\seqadv{u} \mid \Pi_{u\Rightarrow v}$ are equal in law, since cobra walk conditioned on $\Pi_{u\Rightarrow v}$ passes through an active set $\{v\}$ and cobra walk is Markovian. Then, the above becomes
    \begin{equation}
    \label{eq:seqadv-u-expansion}
    \Pr[\seqadv{u} \in \sigma_{v,t}] = \Pr[\seqadv{v} \in \sigma_{v,t}] \Pr[\Pi_{u\Rightarrow v}] + \Pr[\seqadv{u} \in \sigma_{v,t} \land \neg\Pi_{u \Rightarrow v}].
    \end{equation}    
    Now we will bound $\Pr[\seqadv{u} \in \sigma_{v,t} \land \neg\Pi_{u \Rightarrow v}]$. 
    First, for $t \ge 0$, let us denote by $\event^{(a)}_t$ the event that a cobra walk started from source $a$ did not communicate to curious nodes for the first $t$ rounds. Consider one of the branches of a cobra walk that started at $a$. It behaves like a simple random walk. Note that, from Theorem 4.17 of~\cite{Pseudorandomness}, we know that the probability of a random walk not hitting a given non-empty set in $t$ rounds approaches $0$ as $t \to \infty$. In particular, we have
    \begin{equation}
    \label{eq:prob-not-hitting}
    \lim_{t\to \infty}\Pr[\event^{(a)}_t] = 0.
    \end{equation}
    Now, note that, if $\seqadv{u} \in \sigma_{v,t}$, but $\Pi_{u \Rightarrow v}$ does not hold, there are at least two active nodes in $\honest$ when $v$ communicates to $w_1$. We know that $v$ is one of them, let $a \in \honest$ be another active node. By construction of $\sigma_{v,t}$, if $\seqadv{u} \in \sigma_{v,t}$, then for at least $t$ rounds after $v$ communicated with $w_1$, there is no communication between nodes of $\honest$ and $F$. 
    
    Since $a$ is active at the beginning of these $t$ rounds, the probability of non-curious nodes not communicating to curious nodes for $t$ rounds can be upper bounded by $\Pr[\event_t^{(a)}]$ by definition of $\event_t^{(a)}$. Then, as $(\rho/d^2)^t$ is the probability of $w_1$ and $w_2$ communicating with each other for $t$ rounds (twice in each round), we have
    \[
     \Pr[\seqadv{u} \in \sigma_{v,t} \land \neg \Pi_{u \Rightarrow v}] \le \left(\frac{\rho}{d^2}\right)^t\Pr[\event^{(a)}_t].
    \] 
    On the other hand, we have $\Pr[\seqadv{v} \in \sigma_{v,t}] \ge (\rho/d^2)^{t+1}$, since, with probability $\rho/d^2$, the node $v$ communicates to $w_1$ in the first round of the protocol, $w_1$ and $w_2$ exchange messages for $t$ rounds (twice in each round) with probability $(\rho/d^2)^t$. Then, by \equationref{eq:prob-not-hitting}
    \begin{equation}
    \label{eq:limsup-seqadv-u-residue}
    \limsup_{t\to \infty} \frac{ \Pr[\seqadv{u} \in \sigma_{v,t} \land \neg \Pi_{u \Rightarrow v}] }{\Pr[\seqadv{v} \in \sigma_{v,t}]} \le \limsup_{t\to \infty} \Pr[\event_t^{(a)}] \frac{d^2}{\rho} = 0.
    \end{equation}
    Also, by \equationref{eq:seqadv-u-expansion}
    \begin{align*}
    \limsup_{t\to \infty} \frac{\Pr[\seqadv{u} \in \sigma_{v,t}]}{\Pr[\seqadv{v} \in \sigma_{v,t}]} 
    &= \limsup_{t\to \infty} \frac{ \Pr[\seqadv{v} \in \sigma_{v,t}] \Pr[\Pi_{u\Rightarrow v}] + \Pr[\seqadv{u} \in \sigma_{v,t} \land \neg\Pi_{u \Rightarrow v}]}{\Pr[\seqadv{v} \in \sigma_{v,t}]} \\
    &= \Pr[\Pi_{u\Rightarrow v}] + \limsup_{t\to \infty} \frac{ \Pr[\seqadv{u} \in \sigma_{v,t} \land \neg\Pi_{u \Rightarrow v}] }{\Pr[\seqadv{v} \in \sigma_{v,t}]}.
    \intertext{Finally, by \equationref{eq:limsup-seqadv-u-residue}, we have}
    \limsup_{t\to \infty} \frac{\Pr[\seqadv{u} \in \sigma_{v,t}]}{\Pr[\seqadv{v} \in \sigma_{v,t}]} 
    &\le \Pr[\Pi_{u\Rightarrow v}] + 0 = \passage{u}{v}.
    \end{align*}
    Note that $\left(\limsup_{t \to \infty} x_t\right)^{-1} = \liminf_{t \to \infty} 1/x_t$ for any non-negative sequence $\{x_t\}$ since $x \mapsto 1/x$ is decreasing. Then, by definition of max divergence,
    \[
    \divinfty{\seqadv{v}}{\seqadv{u}} \ge \ln\left(\liminf_{t\to\infty} \frac{\Pr[\seqadv{v} \in \sigma_{v,t}]}{\Pr[\seqadv{u} \in \sigma_{v,t}]}\right) \ge \ln\left(\passage{u}{v}^{-1}\right),
    \]
    which concludes the proof.
\end{proof}

The rest of the proof boils down to upper bounding $\passage{u}{v}$. To do so, we introduce a new process, called \emph{$(b,\rho)$-anaconda walk} where $\rho \in [0,1]$ and $b \in \mathbb N$. Essentially, an anaconda walk resembles a cobra walk with two restrictions: (i) only one node (which we call ``the head'') is allowed to branch, and (ii) the total number of times the head branches is limited to $b$. Note that we only introduce anaconda walk to help us derive the upper bound for passage probabilities of cobra walk. 

\paragraph*{Anaconda walk} Now, we give a detailed description of $(b,\rho)$-anaconda. An example of an execution of anaconda walk can be found in \figureref{fig:anaconda}. Let $G = (V,E)$ be an undirected connected graph, and let $s \in \honest$ be the source node. Then, a $(b,\rho)$-anaconda walk started from $s$ can be described as follows. At time $t$, a $(b,\rho)$-anaconda walk is characterized by a triplet $(\anaconda_t, h_t, c_t)$, where $\anaconda_t\subseteq V$ is the set of active nodes, $h_t \in Y_t$ is an active node which we refer to as \emph{the head}, $c_t \in \mathbb Z$ is a counter storing the number the times head branched up to time $t$. Initially, we have $\anaconda_{\starttime} = \{s\}$, $h_{\starttime} = s$ and $c_{\starttime} = 0$ where $\starttime$ is the start time of the protocol. Consider a round $t \ge \starttime$. If $c_t < b$, the current head node $h_t$, with probability $1-\rho$, samples a random neighbor $v_1$ and communicates gossip to $v_1$, or, with probability $\rho$, samples two neighbors $v_1$ and $v_2$ with replacement and successively communicates gossip to each of them (in this case, we will say the head \emph{branches}). If $c_t = b$, the head $h_t$ simply samples one random neighbor $v_1$ and communicates the gossip to $v_1$. In both cases, $v_1$ becomes a new head node ($h_{t+1}$). If the head branches in round $t$, the value of $c_t$ is updated to $c_{t+1} = c_t + 1$; otherwise, we simply have $c_{t+1} = c_t$. Every node $u \in \anaconda_t \setminus \{h_t\}$ communicates the gossip to one random neighbour. As usual, $\anaconda_{t+1}$ includes all the nodes which received the gossip in round $t$, i.e., nodes which send a message in round $t$ but do not receive any will deactivate. 

\input{tikz/Anaconda.tex}

We call a ``branch'' of an anaconda walk a sequence of nodes $u_1,u_2,\ldots$ such that $u_1$ was a head in some round $t$, then it communicated gossip to $u_2$ in round $t$,  $u_2$ communicated gossip to $u_3$ in round $t+1$, etc.

\begin{restatable}{lemma}{AnacondaCoupling}
    \label{lemma:anaconda-coupling}
    Consider an arbitrary graph $G = (V,E)$. Let $s \in V$ be the source node and $b \in \mathbb N$. Consider a $(1+\rho)$-cobra walk with active set $\{\activeset_t\}_{t\ge \starttime}$. Then, there exists $\left\{(\anaconda_t, h_t, c_t)\right\}_{t\ge \starttime}$ such that $\anaconda_t \subseteq \activeset_t$ for every $t\ge \starttime$ and $\left\{(\anaconda_t, h_t, c_t)\right\}_{t\ge \starttime}$ has the same law as the triplet of a $(b,\rho)$-anaconda walk.
\end{restatable}
\begin{proof}
        We will construct such a sequence $\{(\anaconda_t,h_t,c_t)\}_{t\ge \starttime}$ iteratively. First, define $\anaconda_\starttime = \{s\}$, $h_\starttime = \{s\}$ and $c_\starttime = 0$. Then $\anaconda_\starttime \subseteq \activeset_\starttime$. Suppose, for some $t \ge \starttime$, we have $\anaconda_t \subseteq \activeset_t$. 
        
        Consider the head node $h_t \in \anaconda_t \subseteq \activeset_t$. First, consider the case when $c_t = b$, or $h_t$ did not branch in round $t$ of a cobra walk. Let $u \in N(h_t)$ be one of the nodes to which $h_t$ communicated the gossip in round $t$ of the cobra walk. Then set $h_{t+1} = u$ and add $u$ to $\anaconda_{t+1}$. Additionally, set $c_{t+1} = c_t$. Second, consider the case when $c_t < b$ and $h_t$ branches in round $t$ of a cobra walk. Let $u_1, u_2 \in N(h_t)$ be the two neighbours of $h_t$ to which it communicated the gossip in round $t$ of the cobra. Then set $h_{t+1}$ to be either $u_1$ or $u_2$ chosen at random, and add both $u_1$ and $u_2$ to $\anaconda_{t+1}$. Also, set $c_{t+1} = c_t + 1$.
        
        Now, consider a node $v \in \anaconda_t\setminus\{h_t\}$. Let $u \in N(v)$ be one of the nodes to which $v$ communicated the gossip in round $t$ of cobra walk. Add $u$ to $\anaconda_{t+1}$ for every $v \in \anaconda_t\setminus\{h_t\}$. 
        
        Then, $\anaconda_{t+1} \subseteq \activeset_{t+1}$ by construction. Also, notice that, while $c_t < b$, the head $h_t$ branches with probability $\rho$ at each step. After the value of $c_t$ reaches $b$, the head does not branch. Additionally, nodes in $\anaconda_t\setminus \{h_t\}$ never branch. Hence, $\left\{(\anaconda_t, h_t, c_t)\right\}_{t\ge \starttime}$ has the same law as the triplet of a $(b,\rho)$-anaconda walk.
\end{proof}

Similarly to the study of a cobra walk, we define passage probabilities for an anaconda walk.
\begin{definition}
    Consider an execution of a $(b,\rho)$-anaconda walk where $\rho \in [0,1]$ on a graph $G = (V,E)$. Let $F\subseteq V$ be the subset of curious nodes of size $f$. For $v,u \in\honest$, we denote by \emph{passage probability} $\passagehat{u}{v}$ of an anaconda walk the probability that an anaconda walk started from $u$ reached an active set $\{v\}$ while not interacting with any curious nodes. More formally,
    \[
    \passagehat{u}{v} = \Pr[\exists t \ge \starttime: \anaconda_t = \{v\} \land (\anaconda_i \cap F = \emptyset, \forall i \le t) \mid \anaconda_\starttime = \{u\}].
    \]
\end{definition}
Now, we can use passage probabilities of an anaconda walk to upper bound the passage probabilities of a cobra walk. 

\begin{restatable}{lemma}{AnacondaReduction}
    \label{lemma:anaconda-reduction}
    Let $\rho \in [0,1]$ and $b\in \mathbb N$. Consider a $(1 + \rho)$-cobra walk and a $(b,\rho)$-anaconda walk on an undirected connected graph $G= (V,E)$. Then, for any set of curious nodes $F \subset V$ and  $v,u \in \honest$, we have
    \[
    \passage{u}{v} \le \passagehat{u}{v}.
    \]
\end{restatable}
\begin{proof}
    Consider $(1+\rho)$-cobra walk $\{\activeset_t\}_{t\ge \starttime}$ started at $u$. Let $\left\{(\anaconda_t, h_t, c_t)\right\}_{t\ge \starttime}$ be as in \lemmaref{lemma:anaconda-coupling}. By definition of passage probabilities (\definitionref{def:passage}), we have
    \begin{align*}
     \passage{u}{v} &= \Pr[\exists t \ge \starttime: \activeset_t = \{v\} \land (\activeset_i \cap F = \emptyset, \forall i \le t) \mid \activeset_{\starttime} = \{u\}].
     \intertext{By \lemmaref{lemma:anaconda-coupling}, we have $\anaconda_i \subseteq \activeset_i$ for every $i$. Then, $\activeset_i \cap F = \emptyset$ implies $\anaconda_i \cap F = \emptyset$ and $\activeset_t = \{v\}$ implies $\anaconda_t = \{v\}$, hence}
     \passage{u}{v} &\le \Pr[\exists t \ge \starttime: \anaconda_t = \{v\} \land (\anaconda_i \cap F = \emptyset, \forall i \le t) \mid \anaconda_{\starttime} = \{u\}] \\
     &= \passagehat{u}{v},
    \end{align*}
    which concludes the proof.
\end{proof}

Then, to obtain an upper bound on $\passage{u}{v}$, it is enough to upper bound $\passagehat{u}{v}$. To do that, we introduce the notion of \emph{passage bandwidth}. We define bandwidth $\bandwidth$ of the passage from $u$ to $v$ to be the number of times head branched during that passage. Accordingly, $\passagehat[\bandwidth < x]{u}{v}$ and $\passagehat[\bandwidth = x]{u}{v}$ are probabilities of a passage with bandwidth less than $x$ and exactly $x$ respectively.

We first upper bound passage probability with small bandwidth.
\begin{restatable}{lemma}{ProbLowBandwidth}
    \label{lemma:passage-low-bandwidth}
    Consider an execution of a $(b,\rho)$-anaconda walk where $\rho \in (0,1]$ on $d$-regular graph $G = (V,E)$ of size $n$. Let $F\subseteq V$ be the set of curious nodes of size $f \ge 1$. Let $\vec \eigen_1 \in \mathbb R^{n-f}$ be a vector as in \lemmaref{lemma:Q-delocalization}. Then, for every $v \in \honest$, there exists $u_\star \in \honest$ such that
    \[
    \passagehat[\bandwidth < b]{u_\star}{v} \le \frac{\eigen_{1v}}{\norm{\vec\eigen_1}_1}\cdot \frac{b+1}{(f/n + \rho(1 - f/n))(1-\lambda) }.
    \]
\end{restatable}
\begin{proof}
Let $u\in \honest$ be arbitrary and consider a $(b,\rho)$-anaconda walk $\{(\anaconda_t, h_t,c_t)\}_{t\ge \starttime}$ that started at $u$. Consider the head $h_t$. It follows the distribution of a random walk on $G$. Let $\vec e_i\in \mathbb R^{n-f}$ be the $i^\text{th}$ coordinate unit vector (i.e., a vector where $i^\text{th}$ coordinate is equal to one, and the rest are $0$), and let $\vec Q$ be as in \equationref{eq:block-form-absorbing-markov}. Then, the probability that $h_t$ reaches $v$ in exactly $\tau$ steps without hitting any curious nodes is given by 
$(\vec Q)^\tau_{uv} = \vec e_u^\top \vec Q^\tau \vec e_v$
by definition of $\vec Q$. Also, since $h_t$ branches at each step with probability $\rho$, the probability that $h_t$ branches exactly $k < b$ times in $\tau$ steps is given by ${\tau\choose k} \rho^k (1 -\rho)^{\tau-k}$. Then
\begin{equation}
    \label{eq:passage-k-u-v}
    \passagehat[\bandwidth = k]{u}{v} \le \sum_{\tau = 0}^\infty \vec e_u ^\top\vec Q^\tau \vec e_v \cdot {\tau\choose k} \rho^k (1 -\rho)^{\tau-k}.
\end{equation}
By \lemmaref{lemma:Q-delocalization}, all coordinates of $\vec \eigen_1$ are non-negative. Define the distribution $\Phi$ over the nodes in $\honest$ so that $\Phi(w) = \frac{\eigen_{1w}}{\norm{\vec\eigen_1}_1}$. Consider sampling $u$ according to $\Phi$. Then, by \equationref{eq:passage-k-u-v}, we have
\begin{align}
    \E_{u\sim \Phi}\left[\passagehat[\bandwidth = k]{u}{v}\right] &\le \E_{u\sim \Phi}\left[\sum_{\tau = 0}^\infty \vec e_u^\top \vec Q^\tau \vec e_v \cdot {\tau\choose k} \rho^k (1 -\rho)^{\tau-k}\right] \notag\\
    &= \sum_{\tau = 0}^\infty \sum_{u\in\honest} \Phi(u) \vec e_u^\top \vec Q^\tau \vec e_v \cdot {\tau\choose k} \rho^k (1 -\rho)^{\tau-k} \notag\\
    &= \sum_{\tau = 0}^\infty \sum_{u\in\honest} \frac{\eigen_{1u}}{\norm{\vec\eigen_1}_1} \vec e_u^\top \vec Q^\tau \vec e_v \cdot {\tau\choose k} \rho^k (1 -\rho)^{\tau-k} \notag \\
    &= \sum_{\tau = 0}^\infty \frac{\vec \eigen_{1}^\top}{\norm{\vec\eigen_1}_1} \vec Q^\tau \vec e_v \cdot {\tau\choose k} \rho^k (1 -\rho)^{\tau-k}. \notag
    \intertext{By choice of $\vec \eigen_1$ from \lemmaref{lemma:Q-delocalization}, it is a first eigenvector of $\vec Q$, hence}
    \E_{u\sim \Phi}\left[\passagehat[\bandwidth = k]{u}{v}\right]&\le\sum_{\tau = 0}^\infty  \lambda_1(\vec Q)^{\tau} \frac{\vec \eigen_{1}^\top}{\norm{\vec\eigen_1}_1} \vec e_v \cdot {\tau\choose k} \rho^k (1 -\rho)^{\tau-k} \notag\\
    &= \sum_{\tau = 0}^\infty \lambda_1(\vec Q)^{\tau} \frac{\eigen_{1v}}{\norm{\vec\eigen_1}_1}  \cdot {\tau\choose k} \rho^k (1 -\rho)^{\tau-k}.    \label{eq:expect-passage-k-u-v}
\end{align}
Also,
\begin{align*}
    \E_{u\sim \Phi}\left[\passagehat[\bandwidth < b]{u}{v}\right] 
    &= \sum_{k = 0}^{b-1}\E_{u\sim \Phi}\left[\passagehat[\bandwidth < b]{u}{v}\right].
    \intertext{By \equationref{eq:expect-passage-k-u-v}, we have}
    \E_{u\sim \Phi}\left[\passagehat[\bandwidth < b]{u}{v}\right] 
    &\le \sum_{k = 0}^{b-1} \sum_{\tau = 0}^\infty \lambda_1(\vec Q)^{\tau} \frac{\eigen_{1v}}{\norm{\vec\eigen_1}_1}\cdot {\tau\choose k} \rho^k (1 -\rho)^{\tau-k}\\
    &= \frac{\eigen_{1v}}{\norm{\vec\eigen_1}_1} \sum_{k = 0}^{b-1} \sum_{\tau = 0}^\infty \lambda_1(\vec Q)^{\tau} {\tau\choose k} \rho^k (1 -\rho)^{\tau-k}.
    \intertext{Using the identity $\sum_{n = k}^\infty {n \choose k} x^{n-k} = (1-x)^{-k-1}$ for $|x| < 1$, we have}
    \E_{u\sim \Phi}\left[\passagehat[\bandwidth < b]{u}{v}\right] 
    &\le \frac{\eigen_{1v}}{\norm{\vec\eigen_1}_1} \sum_{k = 0}^{b-1} \frac{\rho^k}{(1 - \lambda_1(\vec Q)(1 -\rho))^{k+1}}.
    \intertext{Note that since $\lambda_1(\vec Q)\le 1$ by \lemmaref{lemma:spectral-Q}, we have $\frac{\rho}{1 - \lambda_1(\vec Q)(1 -\rho)} \le \frac{\rho}{1 - (1-\rho)} = 1$, hence}
    \E_{u\sim \Phi}\left[\passagehat[\bandwidth < b]{u}{v}\right] 
    &\le \frac{\eigen_{1v}}{\norm{\vec\eigen_1}_1}\cdot \frac{1}{1 - \lambda_1(\vec Q)(1-\rho)} \sum_{k = 0}^{b-1} 1\\
    &= \frac{\eigen_{1v}}{\norm{\vec\eigen_1}_1} \cdot \frac{b}{1 - \lambda_1(\vec Q)(1-\rho)}.
    \end{align*}
    Then, there exists $u_\star \in \honest$ such that
    \[
    \left[\passagehat[\bandwidth < b]{u_\star}{v}\right] \le \E_{u\sim \Phi}\left[\passagehat[\bandwidth < b]{u}{v}\right] \le \frac{\eigen_{1v}}{\norm{\vec\eigen_1}_1} \cdot \frac{b}{1 - \lambda_1(\vec Q)(1-\rho)},
    \]
    which concludes the proof.
\end{proof}

\begin{restatable}{lemma}{QPowersNR}
    \label{lemma:q-powers-nr}
    Let $\mathcal G$ be an infinite family of $d$-regular near-Ramanujan graphs of size $n$ with $d \in n^{\Omega_n(1)}$. Let $G = (V,E) \in \mathcal G$ be a graph of large enough size. Let $F\subseteq V$ be a set of $f$ curious nodes such that $f/n, \alpha_F \in 1 - \Omega_n(1)$. Let $\vec Q$ be as in \equationref{eq:block-form-absorbing-markov}. Then
    \[
    \sum_{t = 1}^\infty \max_{v,u \in\honest} \left(\vec Q^t\right)_{vu} \in O_n\left(\frac{1}{d} + \frac{1}{f}\right).
    \]
\end{restatable}
\begin{proof}
    Since $\mathcal G$ is near-Ramanujan with $d \in n^{\Omega_n(1)}$, we have 
    \begin{equation}
    \label{eq:lambda-nr-2}
    \lambda \in O_n(d^{-1/2}) \subseteq n^{-\Omega_n(1)}.
    \end{equation} 
    Then, there exists $L \in O_n(1)$ such that $\lambda^L < 1/n$ for all large enough $n$. 
    
    Note that since $\vec Q$ is a principal submatrix of 
    a normalized adjacency matrix $\Hat{\vec A}$ of $G$, every entry of $\vec Q$ is upper bounded by $1/d$. Also, since $\vec Q \succeq \vec O_{(n-f)\times(n-f)}$ and each row of $\vec Q$ has sum less or equal to $1$, we have for every $t \ge 1$
    \begin{equation}
    \label{eq:qt-vu-ub}
    (\vec Q^t)_{vu} \le \frac{1}{d}.
    \end{equation}
    Let $T = \left\lceil\log_{\frac{\lambda}{1 - \alpha_F}}\left(\frac{1 - \alpha_F}{4(n-f)}\right)\right\rceil$ be as in \lemmaref{lemma:Q-delocalization}. Recall that $\alpha_F \in 1 - \Omega_n(1)$ as given in the lemma statement. Combining this with \equationref{eq:lambda-nr-2}, we have $\frac{\lambda}{1 - \alpha_F} \in n^{-\Omega_n(1)}$. Since also $\frac{1 - \alpha_F}{4(n-f)} \in n^{-O_n(1)}$, we have 
    \begin{equation}
    \label{eq:t-is-o1}
    T \in O_n(1).
    \end{equation}
    By \lemmaref{lemma:matrix-powers-first-eigenvector}, for every $t \ge 1$ we have
    \begin{align*}
    (\vec Q^t)_{vu} &\le \lambda_1(\vec Q)^t \eigen_{1v} \eigen_{1u} + \lambda^t.
    \intertext{By \lemmaref{lemma:Q-delocalization}, we get }
    (\vec Q^t)_{vu} &\le \lambda_1(\vec Q)^t \frac{2(1 - \alpha_F)^{-T-1}}{n-f} + \lambda^t.
    \end{align*}
    Combining this with \equationref{eq:qt-vu-ub}, we get 
    \begin{align}
            \sum_{t = 1}^\infty \max_{v,u} \left(\vec Q^t\right)_{vu} 
            &= \sum_{t = 1}^L \max_{v,u} \left(\vec Q^t\right)_{vu} + \sum_{t = L + 1}^\infty \max_{v,u} \left(\vec Q^t\right)_{vu}\notag\\
            &\le \frac{L}{d} + \sum_{t = L + 1}^\infty \left(\lambda_1(\vec Q)^t \frac{2(1 - \alpha_F)^{-T-1}}{n-f} + \lambda^t\right)\notag\\
            &= \frac{L}{d} + \frac{2(1 - \alpha_F)^{-T-1}}{n-f}\sum_{t = L + 1}^\infty \lambda_1(\vec Q)^t + \lambda^L \sum_{t = 1}^\infty\lambda^t. \notag
            \intertext{By definition of $L$, we have}
            \sum_{t = 1}^\infty \max_{v,u} \left(\vec Q^t\right)_{vu} &\le \frac{L}{d} + \frac{2(1 - \alpha_F)^{-T-1}}{n-f}\sum_{t = L + 1}^\infty \lambda_1(\vec Q)^t + \frac{1}{n} \sum_{t = 1}^\infty\lambda^t. \notag
            \intertext{From part (a) of \lemmaref{lemma:spectral-Q}, we get}
            \sum_{t = 1}^\infty \max_{v,u} \left(\vec Q^t\right)_{vu} &\le \frac{L}{d} + \frac{2(1 - \alpha_F)^{-T-1}}{n-f}\sum_{t = L + 1}^\infty \left(1 - \frac{f}{n}(1-\lambda)\right)^t + \frac{1}{n} \sum_{t = 1}^\infty\lambda^t. \notag
            \intertext{Since $1 - \frac{f}{n}(1-\lambda) > 0$ and $\lambda > 0$, we get}
            \sum_{t = 1}^\infty \max_{v,u} \left(\vec Q^t\right)_{vu} &\le \frac{L}{d} + \frac{2(1 - \alpha_F)^{-T-1}}{n-f}\sum_{t = 0}^\infty \left(1 - \frac{f}{n}(1-\lambda)\right)^t + \frac{1}{n} \sum_{t = 0}^\infty\lambda^t. \notag
            \intertext{By reducing the infinite geometric series, we get}
            \sum_{t = 1}^\infty \max_{v,u} \left(\vec Q^t\right)_{vu} &\le \frac{L}{d} + \frac{2(1 - \alpha_F)^{-T-1}}{n-f}\cdot \frac{n}{f(1-\lambda)} + \frac{1}{n(1-\lambda)}.\label{eq:sum-maxvu-ub}
    \end{align}
    Since $f/n \in 1 - \Omega_n(1)$ as given by the lemma statement, we have $n/(n-f) \in O_n(1)$. Also, since $\alpha_F \in 1 - \Omega_n(1)$ as given in the lemma statement and using \equationref{eq:t-is-o1}, we have $(1 - \alpha_F)^{-T-1} \in O_n(1)$. Additionally, since $\lambda \in 1 - \Omega_n(1)$ as given by the lemma statement, then $(1-\lambda)^{-1}\in O_n(1)$. Finally, since $L \in O_n(1)$, by \equationref{eq:sum-maxvu-ub}, we have
    \[
    \sum_{t = 1}^\infty \max_{v,u} \left(\vec Q^t\right)_{vu}\in O_n\left(\frac{1}{d} + \frac{1}{f}\right),
    \]
    which concludes the proof.
\end{proof}
    
We now upper bound passage probability with large bandwidth.
\begin{restatable}{lemma}{PassageLargeBandwidth}
    \label{lemma:passage-large-bandwidth}
    Let $\mathcal G$ be an infinite family of $d$-regular near-Ramanujan graphs with $d \in n^{\Omega_n(1)}$. Let $G = (V,E) \in \mathcal G$ be a graph of large enough size. Let $F\subseteq V$ be a set of $f$ curious nodes such that $f/n, \alpha_F \in 1 - \Omega_n(1)$. Consider an execution of $(b,\rho)$-anaconda walk on $G$ with $b \in O_n(1)$. Then, for any $v,u$ the following holds true
        \[
        \passagehat[\bandwidth = b]{u}{v} \in O_n\left(\left(\frac{1}{d} + \frac{1}{f}\right)^b\right).
        \]
\end{restatable}
\begin{proof}
    Note that by definition of a bandwidth, in a passage with $\bandwidth = b$, the anaconda head branches $b$ times. Let $t_1 < t_2<\ldots < t_b$ be the times at which the head branches. For every $i \in [b]$ and $t \ge t_i$, let $y_{t}^{(i)}$ be the active node that belongs to the branch number $i$ at time $t$. For convenience, also denote $y_t^{(0)} = h_t$ to be a position of anaconda head at time $t$, and denote $t_0 = \starttime$ to be the start time of the anaconda walk. We will say that branch $i$ \emph{coalesces} with branch $j$ at time $t'$ if $y_{{t'}}^{(i)} = y_{{t'}}^{(j)}$. In that case, we also have $y_{t}^{(i)} = y_{t}^{(j)}$ for all $t \ge t'$.

    We will upper bound the passage probability $\passagehat[\bandwidth = b]{u}{v}$ by the probability that, starting at node $u$, all branches coalesce with each other before hitting the curious set. All branches coalesce before time $t > t_b$ if we have $y^{(0)}_t = y^{(1)}_t = \ldots = y^{(b)}_t$. If branch $i$ coalesces with branch $j < i$, we will say that $i$ merges into $j$. Then, all branches coalesce before the time $t > t_b$ if and only if for every $i \in [b]$, the branch $i$ merges into the branch $j$ for some $j < i$ before time $t$. Let $\mathcal A_{ij}$ be the event that the branch $i$ merges into the branch $j < i$ before any of the branches $0,1,\ldots,i$ hits the curious set, and let $\mathcal A_i =
    \bigvee_{j = 0}^{i-1} \mathcal A_{ij}$ (i.e. $\mathcal A_i$ corresponds to branch $i$ merging into one of the branches $\{0,1,\ldots, i-1\}$). Then, all branches coalesce before hitting the curious set if and only if the event $\bigwedge_{i = 1}^b \mathcal A_i$ holds true. 
    
    Now, we will show that the probability of event $\mathcal A_i$ can be upper bounded regardless of the positions of the nodes corresponding to branches $0,1,\ldots,i-1$.
    For $i \ge 0$, let $\Sigma_i$ be a sigma algebra generated by sequences $\{y_t^{(i)}\}_{t\ge t_i}$, $\{y_t^{(i-1)}\}_{t\ge t_{i-1}}$, $\ldots$, $\{y_t^{(0)}\}_{t\ge t_{0}}$. Then, $\mathcal A_{1}, \mathcal A_{2},\ldots, \mathcal A_{i} \in \Sigma_i$ by definition of $\mathcal A_j$. Consider $\Pr[\mathcal A_{ij} \mid \Sigma_{i-1}]$. Let $\vec Q$ be as in \equationref{eq:block-form-absorbing-markov}. Note that, conditioned on $\{y_t^{(j)}\}_{t\ge t_j}$, the probability of branch $i$ merging into $j$ at time $t > t_i$ is upper bounded by $(\vec Q^{t-t_i})_{y_{t_i}^{(i)}y_{t}^{(j)}}$. Then
    \begin{align}
    \Pr[\mathcal A_{ij} \mid \Sigma_{i-1}] &\le \sum_{t = t_i + 1}^\infty (\vec Q^{t-t_i})_{y_{t_i}^{(i)}y_{t}^{(j)}}\notag \\
                                           &\le \sum_{t = t_i + 1}^\infty \max_{v,u\in V\setminus F} (\vec Q^{t-t_i})_{vu}.\notag
                                           \intertext{By shifting the summation index, we get}
    \Pr[\mathcal A_{ij} \mid \Sigma_{i-1}] &\le \sum_{t = 1}^\infty \max_{v,u\in V\setminus F} (\vec Q^{t})_{vu}. \notag
                                           \intertext{From \lemmaref{lemma:q-powers-nr}, we have}
    \Pr[\mathcal A_{ij} \mid \Sigma_{i-1}] &\in O_n\left(\frac{1}{d} + \frac{1}{f}\right).\label{eq:aij-given-sigma-ub}
    \end{align}
    Then, by union bound, we have
    \begin{align}
    \Pr[\mathcal A_{i} \mid \Sigma_{i-1}] 
    &= \Pr\left[\bigvee_{j = 0}^{i-1} \mathcal A_{ij} \mid \Sigma_{i-1}\right] \notag\\
    &\le \sum_{j = 0}^{i-1} \Pr\left[\mathcal A_{ij} \mid \Sigma_{i-1}\right].\notag
    \intertext{Using \equationref{eq:aij-given-sigma-ub} and since $i < b \in O_n(1)$, we get}
    \Pr[\mathcal A_{i} \mid \Sigma_{i-1}] &\in O_n\left(\frac{1}{d} + \frac{1}{f}\right).\label{eq:ai-given-sigma-ub}
    \end{align}
    Then, we have
    \begin{align}
         \Pr\left[\bigwedge_{i=1}^b \mathcal A_i\mid\Sigma_0\right]&= \Pr\left[\bigwedge_{i=1}^b \mathcal A_i\mid \Sigma_0\right ]\notag\\
                                                                   &= \prod_{i=1}^b \Pr\left[\mathcal A_i \mid \mathcal A_{i-1},\ldots,\mathcal A_1, \Sigma_0\right]\notag.
                                                                   \intertext{Since $\mathcal A_{i-1},\ldots, \mathcal A_1 \in \Sigma_{i-1}$ for $i \ge 1$, we get}
         \Pr\left[\bigwedge_{i=1}^b \mathcal \mathcal A_i\mid\Sigma_0\right]&= \prod_{i=1}^b \Pr[A_i \mid \Sigma_{i-1}].\notag
                                                                   \intertext{Since $b \in O_n(1)$, using \equationref{eq:ai-given-sigma-ub}, we get}
         \Pr\left[\bigwedge_{i=1}^b \mathcal A_i\mid\Sigma_0\right]&\in O_n\left(\frac{1}{d} + \frac{1}{f}\right). \label{eq:union-ai-ub}
    \end{align}
    Recall that $\passagehat[\bandwidth = b]{u}{v}$ can be upper bounded by probability of $\bigwedge_{i=1}^b \mathcal A_i$, i.e., the event that all $b$ branches coalesce before hitting the curious set. 
    Then
    \[
    \passagehat[\bandwidth = b]{u}{v}\le \Pr\left[\bigwedge_{i=1}^b \mathcal A_i\right] = \E\left[\Pr\left[\bigwedge_{i=1}^b \mathcal A_i \mid \Sigma_0\right]\right] \in O_n\left(\frac{1}{d} + \frac{1}{f}\right),
    \]
    which concludes the proof.
\end{proof}

\CobraPrivacyLower*
\begin{proof}
First, for $G \in \mathcal G$, we describe such sets of curious nodes $F_\star$, such that proving a lower bound on $D_\infty$ for $F_\star$ is sufficient for a lower bound on $\varepsilon$ for both worst-case and average-case adversaries. Then, we show a lower bound on max divergence for two carefully chosen source nodes in $V\setminus F_\star$. Finally, we show that this yields a desired lower bound on $\varepsilon$ for both the worst-case and the average-case adversaries.

\paragraph*{(i) Selecting a set of curious nodes $F_\star$.}

First, let $n$ be large enough such that we have $f \ge 3$ (since $f \in n^{\Omega_n(1)}$, such $n$ exists). Let $G = (V,E)$ be an arbitrary element of $\mathcal G$ of size $n$. Consider $F \sim\uniformsubset{f}{V}$, a set of $f$ curious nodes sampled at random. Let $\alpha$ be as in \lemmaref{lemma:adversarial-density}. Since $\mathcal G$ is a near-Ramanujan family, by \lemmaref{lemma:alpha-nr} we have $\alpha \in 1 - \Omega_n(1)$. 

By \lemmaref{lemma:adversarial-density}, we have $\alpha_F \le \alpha$ with probability $1 - \frac{1}{n}$. Also, since $f\ge 3$, by \lemmaref{lemma:two-curious-prob}, with probability at least $2/(n-1)$, there is an edge between two nodes of $F$. Hence, we have both (1) $\alpha_F \le \alpha$ and (2) there is an edge between two nodes of $F$ with $\frac{2}{n-1} - \frac{1}{n} > \frac{1}{n}$. Let $F_\star$ be any set of curious nodes for which both (1) and (2) hold.

\paragraph*{(ii) Lower bounding max divergence.} Define $T$ and $\vec\eigen_1$ as in \lemmaref{lemma:Q-delocalization} for set $F_\star$. By \lemmaref{lemma:Q-delocalization}, we have for every $v \in V\setminus F_\star$, 
\[
\frac{(1-\alpha_{F_\star})^{(T+1)/2}}{\sqrt{2(n-f)}}\le \eigen_{1v} \le \frac{\sqrt{2}(1-\alpha_{F_\star})^{-(T+1)/2}}{\sqrt{(n-f)}}.
\]
Then, for every $v \in V\setminus F_\star$, we have
\[
\norm{\vec\eigen_1}_1 = \sum_{v \in V\setminus F_\star} \left|\eigen_{1v}\right| \ge \frac{(1-\alpha_{F_\star})^{(T+1)/2}\sqrt{n-f}}{\sqrt{2}}.
\]
This implies that for every $v \in V\setminus F_\star$
\begin{align}
\frac{\eigen_{1v}}{\norm{\vec\eigen_1}_1} &\le \frac{\frac{\sqrt{2}(1-\alpha_{F_\star})^{-(T+1)/2}}{\sqrt{(n-f)}}}{\frac{(1-\alpha_{F_\star})^{(T+1)/2}\sqrt{n-f}}{\sqrt{2}}}\notag\\
&= \frac{2(1-\alpha_{F_\star})^{(T+1)}}{n-f}.\notag
\intertext{By \equationref{eq:t-is-o1}, we have $T \in O_n(1)$, and also $\alpha_{F_\star} \le \alpha \in 1 - \Omega_n(1)$, therefore, }
\frac{\eigen_{1v}}{\norm{\vec\eigen_1}_1} &\in O_n\left(\frac{1}{n-f}\right).\notag
\intertext{Since $f/n\in 1-\Omega_n(1)$, we have $n-f \in \Theta_n(n)$, hence}
\frac{\eigen_{1v}}{\norm{\vec\eigen_1}_1} &\in O_n\left(\frac{1}{n}\right).
\label{eq:eigen-over-1-norm}
\end{align}
Recall that $f,d \in n^{\Omega_n(1)}$ as given by the lemma statement. Let $b \in O_n(1)$ be such that $(1/d + 1/f)^b\in O_n(1/n^2)$, and consider a $(b,\rho)$-anaconda walk.  

By our choice of $F_\star$ in (ii), $F_\star$ contains at least two nodes connected by an edge. Let $v_\star$ be as given by \lemmaref{lemma:privacy-lb-necessary-passage}. Then, by \lemmaref{lemma:passage-low-bandwidth}, there exist  $u_\star \in V\setminus F_\star$ such that
\begin{align}
    \passagehat[\bandwidth < b]{u_\star}{v_\star} &\le \frac{\eigen_{1v_\star}}{\norm{\vec\eigen_1}_1}\cdot \frac{b+1}{(f/n + \rho(1 - f/n))(1-\lambda)}.
    \intertext{By \equationref{eq:eigen-over-1-norm}, we get}
    \passagehat[\bandwidth < b]{u_\star}{v_\star}&\in O_n\left(\frac{1}{n}\cdot \frac{b+1}{(f/n + \rho(1 - f/n))(1-\lambda) } \right).
    \intertext{Hence,}
    \passagehat[\bandwidth < b]{u_\star}{v_\star} &\in O_n\left(\frac{b+1}{(f + \rho(n - f))(1-\lambda) } \right).
    \intertext{Recall that we have $b \in O_n(1)$ and $\mathcal G$ is near-Ramanujan. Hence,}
    \passagehat[\bandwidth < b]{u_\star}{v_\star} &\in O_n\left(\frac{1}{f + \rho(n - f)} \right).
    \label{eq:thm-proof-low-bandwidth}
\end{align}
From \lemmaref{lemma:passage-large-bandwidth}, we have
\begin{align}
    \passagehat[\bandwidth = b]{u_\star}{v_\star} 
    &\in O_n\left(\left(\frac{1}{d} + \frac{1}{f}\right)^b\right).\notag
    \intertext{Recall that $b$ is chosen such that $\left(1/d + 1/f\right)^b \in O_n(1/n^2)$, hence,}
    \passagehat[\bandwidth = b]{u_\star}{v_\star} &\in O_n\left(\frac{1}{n^2}\right).
    \label{eq:thm-proof-large-bandwidth}
\end{align}
Then, by \lemmaref{lemma:anaconda-reduction}, we have
\begin{align}
\passage{u_\star}{v_\star} 
&\le \passagehat{u_\star}{v_\star}\notag \\
&\le \passagehat[\bandwidth = b]{u_\star}{v_\star} + \passagehat[\bandwidth < b]{u_\star}{v_\star}. \notag
\intertext{Using \equationref{eq:thm-proof-low-bandwidth} and \equationref{eq:thm-proof-large-bandwidth}, we get}
\passage{u_\star}{v_\star}  &\in O_n\left(\frac{1}{f + \rho(n - f)} + \frac{1}{n^2}\right).\notag
\intertext{Since $f + \rho(n-f) \le n < n^2$, we have}
\passage{u_\star}{v_\star} &\in O_n\left(\frac{1}{f + \rho(n - f)}\right).\notag
\end{align}
Recall that $v_\star$ is chosen as in \lemmaref{lemma:privacy-lb-necessary-passage}. Hence, we have
\begin{equation*}
    \divinfty{\seqadv{v_\star}}{\seqadv{u_\star}} \ge \ln\left(\passage{u_\star}{v_\star}^{-1}\right) \in \ln\left(f + \rho(n - f)\right) + \Omega_n(1),
\end{equation*}
which yields
\begin{equation}
\label{eq:divinfty-in-omega}
\max_{v,u\in V\setminus F_\star}\divinfty{\seqadv{v}}{\seqadv{u}} \in \ln\left(f + \rho(n - f)\right) + \Omega_n(1),
\end{equation}

\paragraph*{(iii) Conclusion.} Recall that the above holds true for any set of curious nodes $F_\star$ for which both of the following conditions hold: (1) $\alpha_F \le \alpha$, and (2) there is an edge between two nodes of $F_\star$. As we explained in (i), this happens with probability $> 1/n$ as we select $F_\star$ randomly from the distribution $\uniformsubset{f}{V}$. Hence, as $F_\star$ is sampled uniformly at random, the lower bound in (\ref{eq:divinfty-in-omega}) holds with probability $>1/n$. Hence, as per equation (\ref{eq:average-guarantee-probability}) we have for average-case adversary 
\begin{equation}
\label{eq:epsilon-in-omega}
\varepsilon \in \ln\left(f + \rho(n - f)\right) + \Omega_n(1).
\end{equation} 
Since in the worst-case we require $\varepsilon$-DP for \emph{every} set of curious nodes, (\ref{eq:epsilon-in-omega}) applies to the worst-case adversary as well, which concludes the proof.
\end{proof}

\subsubsection{Dissemination time of cobra walks (proof of \theoremref{thm:cobra-speed})} \label{subsec:appendix-speed-vs-privacy}

First, we establish the following lemma, which we use to show that the expected dissemination time of $(1+\rho)$-cobra walk is a non-increasing  function of $\rho$.
\begin{lemma}
\label{lemma:cobra-coupling}
Consider an undirected connected graph $G = (V,E)$ of size $n$ and let $s \in V$ be an arbitrary source node. Let $\{\activeset_{t}\}_{t \ge \starttime}$ be an active set of $(1+\rho)$-cobra walk started at $s$. Then, for every $\rho' < \rho$, there exists $\{\activeset_{t}'\}_{t \ge \starttime}$ such that $\activeset_{t}' \subseteq \activeset_t$ for every $t \ge \starttime$, and $\{\activeset_{t}'\}_{t \ge \starttime}$ has the same law as the active set of a $(1+\rho')$-cobra walk.
\end{lemma}
\begin{proof}
    We construct the sequence $\{\activeset_{t}'\}_{t \ge \starttime}$ iteratively. First, define $\activeset_{\starttime}' = \activeset_\starttime = \{s\}$. Suppose for some $t \ge \starttime$, we have $\activeset_t' \subseteq \activeset_t$ and consider an arbitrary node $x \in \activeset_t'$. First, suppose $x$ does not branch in round $t$ of $(1+\rho)$-cobra walk. Then, we simply add the node $y$, to which $x$ communicated the gossip, to $\activeset_{t+1}'$. On the other hand, suppose $x$ branches in round $t$ of $(1+\rho)$-cobra walk, and let $y_1, y_2 \in N(x)$ be the two nodes to which $x$ communicated the gossip. Then, with probability $\frac{\rho'}{\rho}$ we include both $y_1,y_2$ in $\activeset_{t+1}'$ and, with probability $1 - \frac{\rho'}{\rho}$, we include only $y_1$ in $\activeset_{t+1}'$. We construct set $\activeset_{t+1}'$ by doing the described procedure for every $x \in \activeset_t'$.

    Note that, by construction, every node of $\activeset_t'$ has two neighbours in $\activeset_{t+1}'$ to which it communicated the gossip with probability $\rho \cdot \frac{\rho'}{\rho} = \rho'$. Accordingly, it has only one neighbour in $\activeset_{t+1}'$ to which it communicated the gossip with probability $1 - \rho'$. This implies that $\{\activeset_{t}'\}_{t \ge \starttime}$ has the same law as the active set of a $(1 + \rho')$-cobra walk which concludes the proof.
\end{proof}

\begin{restatable}{theorem}{CobraSpeed}
\label{thm:cobra-speed}
Consider an undirected connected $(d,\lambda)$-expander graph $G$ of size $n$ with $d \ge 3$ and $\lambda \in 1 - \Omega_n(1)$. Let $\mathcal P$ be a $(1+\rho)$-cobra walk with $\rho \in \omega_n\left(\sqrt{\log(n)/n}\right)$. Then, regardless of the identity of the source, the expected dissemination time of $\mathcal P$ is a non-increasing function of $\rho$. Furthermore, the dissemination time is at least $\Omega_n\left(\log(n)/\rho\right)$ rounds in expectation and at most $O_n\left(\log(n)/\rho^3\right)$ rounds in expectation and with high probability. \end{restatable}

\begin{proof}
First, we show that the dissemination time is a non-increasing function of $\rho$. Next, we show how the upper bound on the dissemination time follows from \cite{CobraExpanders2016}. Finally, we show the lower bound on dissemination time. For the remainder of the proof, we fix an arbitrary source $s \in V$. Additionally, we denote by $T_{\mathcal P}$ the dissemination time of $(1+\rho)$-cobra walk started from $s$.

\paragraph*{(i) Dissemination time is a non-decreasing function of $\rho$.} Let $\mathcal P$ and $\mathcal P'$ be a $(1 + \rho)$-cobra walk and a $(1 + \rho')$-cobra walk respectively started from $s$, where $0 \le \rho' < \rho \le 1$. Let $\activeset_t$ be an active set of $\mathcal P$ at time $t$. By \lemmaref{lemma:cobra-coupling}, there exist $\{\activeset_{t}'\}_{t \ge \starttime}$ such that $\activeset_{t}' \subseteq \activeset_{t}$ for every $t \ge \starttime$ and $\{\activeset_{t}'\}_{t \ge \starttime}$ follows $\mathcal P'$ in law. Note that the expected dissemination time can be expressed in the following way
\begin{align*}
    \E[T_{\mathcal P'}] &= \E\left[ \min\left\{t' \ge \starttime\colon \bigcup_{t \in [\starttime, t']} \activeset_{t}' = V \right\}\right].
    \intertext{Since $\activeset_{t}' \subseteq \activeset_t$ for every $t$, $\bigcup_{t \in [\starttime, t']} \activeset_{t}' = V$ also implies $\bigcup_{t \in [\starttime, t']} \activeset_{t} = V$. Hence, }
     \E[T_{\mathcal P'}] &\ge \E\left[ \min\left\{t' \ge \starttime\colon \bigcup_{t \in [\starttime, t']} \activeset_{t} = V \right\}\right]\\
                        &= \E[T_{\mathcal P}].
\end{align*}
Since the above holds true for every $\rho' < \rho$, this concludes the proof of this part.

\paragraph*{(ii) Upper bound.} The work of~\cite{CobraExpanders2016} only considers cases of $\rho \in \Omega_n(1)$. We now show how the analysis in~\cite{CobraExpanders2016} can be expended to other regimes of $\rho$. Their proof is based on analyzing a \emph{BIPS} proces, which is dual to cobra walk. Let $A_t$ be an active set of $(1+\rho)$-BIPS process at time $t$, as defined in~\cite{CobraExpanders2016}. Then, Corollary 1 of~\cite{CobraExpanders2016} states that for any $A \subseteq V$, we have
\begin{equation}
\label{eq:bips-growth}
\E\left[|A_{t+1}| \mid A_t = A\right] \ge \left(1 + \rho(1 - \lambda^2)\left(1 - \frac{|A|}{n}\right)\right).
\end{equation}
Note that since $\lambda < 1$, we have $\rho(1-\lambda^2) > 0$. Let $\lambda' < 1$ be such that $1 - (\lambda')^2 = \rho(1 - \lambda^2)$. Then, the above can be rewritten as 
\[
\E\left[|A_{t+1}| \mid A_t = A\right] \ge \left(1 + (1 - (\lambda')^2)\left(1 - \frac{|A|}{n}\right)\right),
\]
which is an analog of Lemma 1~\cite{CobraExpanders2016} with $\lambda$ replaced by $\lambda'$. The rest of the analysis of~\cite{CobraExpanders2016} holds true with $\lambda$ replaced by $\lambda'$. Note that 
\[
\lambda' = \sqrt{1 - \rho(1-\lambda^2)} \le \frac{1 + 1 - \rho(1-\lambda^2)}{2} = 1 - \frac{\rho(1 - \lambda^2)}{2}.
\]
Then, since $\rho \in \omega_n\left(\sqrt{\frac{\log(n)}{n}}\right)$ and $\lambda \in 1 - \Omega_n(1)$ (as the underlying graph $G$ is near-Ramanujan), we have 
\[
1 - \lambda' \ge  \frac{\rho(1 - \lambda^2)}{2}\in \omega_n\left(\sqrt{\frac{\log(n)}{n}}\right).
\]
Then, Theorem 1~\cite{CobraExpanders2016} applied for $\lambda'$ yields that the dissemination time $T_{\mathcal P}$ of $(1 + \rho)$-cobra walk is upper bounded in expectation and with high probability by 
\[
O_n\left(\frac{\log(n)}{(1 - \lambda')^3}\right) \subseteq O_n\left(\frac{\log(n)}{\rho^3(1 - \lambda^2)^3}\right) \subseteq O_n\left(\frac{\log(n)}{\rho^3}\right),
\]
which concludes the proof of the upper bound.

\paragraph*{(iii) Lower bound.} 
Denote $N_t = |\activeset_t|$ to be the number of active nodes of a $(1+\rho)$-cobra walk in round $t\ge \starttime$. Accordingly, $N_{\starttime} = 1$. Also, define $B_t$ to be the number of active nodes that branch in round $t$. Note that $N_{t+1} \le (N_t - B_t) + 2B_t$, since every node in round $t$ that did not branch corresponds to at most one active node at time $t$, and every node that branched corresponds to at most two. Hence, for every $t\ge\starttime$
\[
N_{t + 1} \le N_t + B_t,
\]
which implies
\begin{equation}
\label{eq:nt-plus-1-ub-conditional}
\E[N_{t+1}\mid N_t] \le N_t + \E[B_t \mid N_t].
\end{equation}
Note that, by definition of a cobra walk, every node in round $t$ branches with probability $\rho$. Then, we have $B_t \mid N_t \sim \Bin{N_t}{\rho}$. Then, from \equationref{eq:nt-plus-1-ub-conditional}, we have
\[
\E[N_{t+1}\mid N_t] \le N_t + N_t\rho = (1 + \rho)N_t.
\]
Hence, for every $t \ge \starttime$, we have
\[
\E[N_{t+1}] \le (1 + \rho)\E[N_t],
\]
which, since $N_t = 1$ and using a simple inductive argument, implies that for every $t \ge \starttime$ we have
\begin{equation}
\label{eq:nt-expected}
\E[N_t] \le (1+\rho)^{t-1}.
\end{equation}
Now, we define $\tau$ as the first round such that $N_{\tau} \ge n^{1/4}$. Also, we set $\tilde{t} = \left\lfloor \log_{1 + \rho}\left(\frac{n^{1/4}}{2}\right) \right\rfloor + 1$. Then, by definition of $\tau$, we have
    \begin{align}
    \Pr[\tau \le \tilde{t}] 
    &\leq \Pr[N_{\tilde{t}} \ge n^{1/4}].\notag
    \intertext{By Markov inequality, we get}
    \Pr[\tau \le \tilde{t}]  &\le \frac{\E[N_{\tilde{t}}]}{n^{1/4}}.\notag
    \intertext{By \equationref{eq:nt-expected}, we have}
    \Pr[\tau \le \tilde{t}]  &\le \frac{(1+\rho)^{\tilde{t}-1}}{n^{1/4}}.\notag
    \intertext{Recall that $\tilde{t} - 1 = \left\lfloor \log_{1 + \rho}\left(\frac{n^{1/4}}{2}\right) \right\rfloor\le \log_{1 + \rho}\left(\frac{n^{1/4}}{2}\right)$, hence}
    \Pr[\tau \le \tilde{t}] &\le \frac{1}{2}. \label{eq:tau-ub}
    \end{align}
    Note that while the size of the active set $\activeset_t$ does not exceed $n^{1/4}$, cobra walk needs at least $n/n^{1/4} = n^{3/4}$ rounds to disseminate information to all nodes. Thus, the dissemination time of the cobra walk can be bounded from below by $\min\{n^{3/4}, \tau \}$. Then,
    \begin{align}
    \E\left[T_{\mathcal P}\right] 
    &\ge \E[\min\{n^{3/4}, \tau\}]. \notag 
    \intertext{Note that by \equationref{eq:tau-ub}, we have $\tau \ge \tilde{t}$ with probability $\ge 1/2$, hence,}
    \E\left[T_{\mathcal P}\right] 
    &\ge \min\left\{n^{3/4}, \frac{1}{2}\tilde{t}\right\}. \label{eq:expected-tp-lb}
    \end{align}
    Finally, note that
    \begin{align*}
    \tilde{t} &= \left\lfloor \log_{1 + \rho}\left(\frac{n^{1/4}}{2}\right) \right\rfloor \\
              &\in \Omega_n\left(\log_{1 + \rho}(n)\right).
              \intertext{Using $\log_{1 + \rho}(n) = \frac{\log(n)}{\log(1 + \rho)} \ge \frac{\log(n)}{\rho}$, we get}
    \tilde{t} &\in \Omega_n\left(\frac{\log(n)}{\rho}\right).
    \end{align*}
    Combined with \equationref{eq:expected-tp-lb}, this implies
    \begin{align*}
    \E\left[T_{\mathcal P}\right] &\in \Omega_n\left( \min\left\{n^{3/4}, \frac{\log(n)}{\rho}\right\}\right).
    \intertext{Since $\rho \in \omega_n\left(\sqrt{\frac{\log(n)}{n}}\right)$, we have $\frac{\log(n)}{\rho} \in o_n\left(n^{1/2}\right) \subseteq  o_n\left(n^{3/4}\right)$, hence}
    \E\left[T_{\mathcal P}\right] &\in \Omega_n\left( \frac{\log(n)}{\rho}\right),
    \end{align*}
    which concludes the proof of the lower bound.
\end{proof}

\subsection{Trade-off for Dandelion}

\subsubsection{Proof of the tightness of \corollaryref{corollary:privacy-nr}}

To show that \corollaryref{corollary:privacy-nr} is also tight for Dandelion, we prove the following.
\begin{restatable}{theorem}{DandelionPrivacyLower}\label{thm:dandelion-privacy-lower}
    Let $\mathcal P$ be a $\rho$-Dandelion with $\rho \in [0,1)$ and let $\mathcal G$ be a family of $d$-regular near-Ramanujan graphs with $n$ nodes, $f$ of which are curious, and $d \in n^{\Omega_n(1)}$. Suppose $f/n\in1 - \Omega_n(1)$ and $\mathcal P$ satisfies $\varepsilon$-DP against either an average-case or a worst-case adversary on $G \in \mathcal G$. Then \[\varepsilon \in \ln{(\rho (n-f) + f)}+\Omega_n(1).\] 
\end{restatable}
The proof is similar to the proof of \theoremref{thm:cobra-privacy-lower}. Note that the statement of \theoremref{thm:dandelion-privacy-lower} with $\rho = 0$ follows from \theoremref{thm:universal-impossibility}. Without loss of generality, we assume $\rho \in (0,1)$ in the remaining. Let $\anonphase_t$ be the value of $\anonphase$ at round $t$.
\begin{definition}
    \label{def:explosive-passage}
    Consider an execution of a $\rho$-Dandelion where $\rho \in (0,1]$ on an undirected connected graph $G = (V,E)$. Let $F\subseteq V$ be a subset of curious nodes of size $f$. For $v,u \in\honest$, define the \emph{passage probability} $\passagedandelion{v}{u}$ from $v$ to $u$ as the probability of a protocol started from $v$ to reach node $u$ while not interacting with any curious nodes, and while still in the anonymity phase. More formally,
    \[
    \passagedandelion{v}{u} = \Pr\left[\exists t \ge \starttime: \activeset_t = \{u\} \land \anonphase_t = 1 \land (\activeset_i \cap F = \emptyset, \forall i \le t) \mid \activeset_{\starttime} = \{v\}\right].
    \]
\end{definition}

\begin{lemma}
    \label{lemma:necessary-explosive-passage}
        Let $\mathcal P$ be a $\rho$-Dandelion where $\rho \in (0,1)$, and consider an undirected connected graph $G = (V,E)$ of size $n$. Let $F\subseteq V$ be a set of curious nodes of size $1\le f < n-2$. Then, for any $v \in N(F)$ and $u \in \honest$ we have
        \[
        \divinfty{\seqadv{v}}{\seqadv{u}} \ge \ln\left(\passagedandelion{u}{v}^{-1}\right).
        \]
\end{lemma}
\begin{proof}

    Since $v \in N(F)$, there exists a $w \in F$, such that $v$ and $w$ are neighbors. Recall that, for Dandelion, $S_\textsc{adv} = \{\left(\fadv(\commset_t), \anonphase_t\right)\}_{t \ge \tadv}$. Define $\sigma_v$ to be the set of adversarial observations which begin with $\fadv\left(\commset_{\tadv}\right) = \{(v \to w)\}$ and $\anonphase_{\tadv} = 1$ (i.e., the protocol is still in the anonymity phase).

    Let us define $\Pi_{u \Rightarrow v}$ as 
    \[
    \Pi_{u \Rightarrow v} = \left\{\exists t \ge \starttime: \activeset_t = \{u\} \land \anonphase_t = 1 \land (\activeset_i \cap F = \emptyset, \forall i \le t) \mid \activeset_{\starttime} = \{v\}\right\}
    \]

    By \definitionref{def:explosive-passage}, we have
    \begin{equation}
        \label{eq:passage-uv}
        \passagedandelion{u}{v} = \Pr[\Pi_{u \Rightarrow v}].
    \end{equation}
    Now, consider the probability of $\seqadv{v} \in \sigma_v$. By law of total probability, we have
    \[
    \Pr[\seqadv{u} \in \sigma_v] = \Pr[\seqadv{u} \in \sigma_v \mid \Pi_{u \Rightarrow v}] \Pr[\Pi_{u \Rightarrow v}] + \Pr[\seqadv{u} \in \sigma_v \mid \neg \Pi_{u \Rightarrow v}] \Pr[\neg \Pi_{u \Rightarrow v}].
    \]

    Note that it is impossible for curious nodes to observe communication $(v\to w)$ in the anonymity phase if node $v$ was never active during the anonymity phase. Then $\Pr[\seqadv{u} \in \sigma_v \mid \neg \Pi_{u \Rightarrow v}] = 0$. Also, note that $\seqadv{v}$ and $\seqadv{u} \mid \Pi_{u \Rightarrow v}$ are equal in law, since an execution conditioned on $\Pi_{u \Rightarrow v}$ passes through an active set $\{v\}$ in the anonymity phase which is equivalent to executing the protocol with $v$ as the source node. Hence, the above becomes
    \[
    \Pr[\seqadv{u} \in \sigma_v] = \Pr[\seqadv{u} \in \sigma_v \mid \Pi_{u \Rightarrow v}] \Pr[\Pi_{u \Rightarrow v}] =\Pr[\seqadv{v} \in \sigma_v] \Pr[\Pi_{u \Rightarrow v}].
    \]
    From \equationref{eq:passage-uv}, we get
    \[
    \Pr[\seqadv{u} \in \sigma_v] = \Pr[\seqadv{v} \in \sigma_v] \passagedandelion{u}{v}.
    \]
    It remains to notice that $\Pr[\seqadv{v} \in \sigma_v] \not = 0$, since there is a positive probability of node $v$ contacting $w$ in the first round of the protocol while still in the anonymity phase (since $\rho > 0$). Hence,
    \[
    \divinfty{\seqadv{v}}{\seqadv{u}} \ge \ln\left(\frac{\Pr[\seqadv{v} \in \sigma_v]}{\Pr[\seqadv{u} \in \sigma_v]}\right) = \ln\left(\passagedandelion{u}{v}^{-1}\right),
    \]
    which concludes the proof.
\end{proof}

Now, the proof boils down to upper bounding $\passagedandelion{u}{v}$ appropriately.
\begin{restatable}{lemma}{PassageUB}
    \label{lemma:passage-explosive-ub}
    Consider an execution of a $\rho$-Dandelion where $\rho \in (0,1)$ on $d$-regular graph $G = (V,E)$ of size $n$. Let $F\subseteq V$ be the set of curious nodes of size $f \ge 1$. Let $\vec \eigen_1 \in \mathbb R^{n-f}$ be a vector as in \lemmaref{lemma:Q-delocalization}. Then, for every $v \in \honest$, there exists $u_\star \in \honest$ such that
    \[
    \passagedandelion{u_\star}{v} \le \frac{\eigen_{1v}}{\norm{\vec\eigen_1}_1}\cdot \frac{1}{(f/n + \rho(1 - f/n))(1-\lambda) }.
    \]
\end{restatable}
\begin{proof}
Let $a_t$ be the active node in round $t$ of the anonymity phase of the execution of the protocol. It follows the distribution of a random walk on $G$. Let $\vec e_i \in \mathbb R^{n-f}$ be an $i^\text{th}$ coordinate unit vector ($i^\text{th}$ coordinate is $1$, and the rest is $0$). Let $\vec Q$ be as in \equationref{eq:block-form-absorbing-markov}. 
Then, the probability that $a_t$ reaches $v$ starting from some node $u\in\honest$ in exactly $\tau$ steps without contacting curious nodes is given by 
$(\vec Q)^\tau_{uv} = \vec e_u^\top \vec Q^\tau \vec e_v$
by definition of $\vec Q$. Also, the execution must stay in the anonymity phase for $\tau$ steps, which happens with probability $(1 -\rho)^{\tau}$. Then
\begin{equation}
    \label{eq:passage-k-u-v-dandelion}
    \passagedandelion{u}{v} = \sum_{\tau = 0}^\infty \vec e_u^\top\vec Q^\tau \vec e_v (1 -\rho)^{\tau}.
\end{equation}
By \lemmaref{lemma:Q-delocalization}, all coordinates of $\vec \eigen_1$ are non-negative. Define the distribution $\Phi$ over the nodes in $\honest$ so that $\Phi(w) = \frac{\eigen_{1w}}{\norm{\vec\eigen_1}_1}$. Consider sampling $u$ according to $\Phi$. Then, by \equationref{eq:passage-k-u-v-dandelion}, we have
\begin{align}
    \E_{u\sim \Phi}\left[\passagedandelion{u}{v}\right] &= \E_{u\sim \Phi}\left[\sum_{\tau = 0}^\infty (1-\rho)^{\tau}\vec e_u^\top \vec Q^\tau \vec e_v \right] \notag\\
    &= \sum_{\tau = 0}^\infty \sum_{u\in\honest} \Phi(u)  (1-\rho)^{\tau}\vec e_u^\top \vec Q^\tau \vec e_v  \notag\\
    &= \sum_{\tau = 0}^\infty \sum_{u\in\honest} \frac{\eigen_{1u}}{\norm{\vec\eigen_1}_1} (1-\rho)^{\tau}\vec e_u^\top \vec Q^\tau \vec e_v  \notag \\
    &= \sum_{\tau = 0}^\infty \frac{\vec \eigen_{1}^\top}{\norm{\vec\eigen_1}_1} (1-\rho)^{\tau} \vec Q^\tau \vec e_v. \notag
    \intertext{By choice of $\vec \eigen_1$ from \lemmaref{lemma:Q-delocalization}, it is a first eigenvector of a symmetric matrix $\vec Q$, hence}
    \E_{u\sim \Phi}\left[\passagedandelion{u}{v}\right]&=\sum_{\tau = 0}^\infty  \lambda_1(\vec Q)^{\tau} (1-\rho)^{\tau}\frac{\vec \eigen_{1}^\top}{\norm{\vec\eigen_1}_1} \vec e_v \notag\\
    &= \sum_{\tau = 0}^\infty \lambda_1(\vec Q)^{\tau} \frac{\eigen_{1v}}{\norm{\vec\eigen_1}_1} (1-\rho)^{\tau} \notag \\
    &= \frac{\eigen_{1v}}{\norm{\vec\eigen_1}_1} \cdot \frac{1}{1 - \lambda_1(\vec Q) (1-\rho)}. \notag\\
    \intertext{From \lemmaref{lemma:spectral-Q}(a), we have $\lambda_1(\vec Q) \le 1 - (1 - \lambda)\frac{f}{n}$, which gives}
    \E_{u\sim \Phi}\left[\passagedandelion{u}{v}\right] &\le \frac{\eigen_{1v}}{\norm{\vec\eigen_1}_1} \cdot \frac{1}{1 - (1 - (1 - \lambda)\frac{f}{n}) (1-\rho)} \notag\\
    &\le \frac{\eigen_{1v}}{\norm{\vec\eigen_1}_1} \cdot \frac{1}{(1 - \lambda)\frac{f}{n} + \rho(1 - (1-\lambda)\frac{f}{n})} \notag\\
    &\le \frac{\eigen_{1v}}{\norm{\vec\eigen_1}_1} \cdot \frac{1}{(1 - \lambda)\left(\frac{f}{n} + \rho(1 - \frac{f}{n})\right)}
    \label{eq:expect-passage-explosive}
\end{align}
The upper bounds above hold in expectation as $u \sim \Phi$, hence there exists $u_\star \in \honest$ such that
\[
\passagedandelion{u_\star}{v} \le \frac{\eigen_{1v}}{\norm{\vec\eigen_1}_1} \cdot \frac{1}{(1 - \lambda)\left(\frac{f}{n} + \rho(1 - \frac{f}{n})\right)},
\]
which concludes the proof.
\end{proof}

\DandelionPrivacyLower*
\begin{proof}

For $G \in \mathcal G$, we show a lower bound on max divergence for two carefully chosen source nodes. Let $G = (V,E)$ be an arbitrary element of $\mathcal G$. Let $F \sim\uniformsubset{f}{V}$ be a set of $f$ curious nodes sampled at random. Let $\alpha$ be as in \lemmaref{lemma:adversarial-density}. Since $\mathcal G$ is a near-Ramanujan family, by \lemmaref{lemma:alpha-nr} we have $\alpha \in 1 - \Omega_n(1)$. Also, by \lemmaref{lemma:adversarial-density}, we have $\alpha_F \le \alpha$ with high probability. Let $F_\star \subset V$ be an arbitrary set of size $f$ for which $\alpha_{F_\star} \le \alpha$ holds true. Since $\mathcal G$ is near-Ramanujan with $d \in n^{\Omega_n(1)}$, we have 
\begin{equation}
    \label{eq:lambda-nr-explosive}
    \lambda \in O_n(d^{-1/2}) \subseteq n^{-\Omega_n(1)}.
\end{equation} 
Let $T = \left\lceil\log_{\frac{\lambda}{1 - \alpha_{F_\star}}}\left(\frac{1 - \alpha_{F_\star}}{4(n-f)}\right)\right\rceil$ be as in \lemmaref{lemma:Q-delocalization} for set $F_\star$. Recall that $\alpha_{F_\star} \le \alpha \in 1 - \Omega_n(1)$. Combining this with \equationref{eq:lambda-nr-explosive}, we have $\frac{\lambda}{1 - \alpha_F} \in n^{-\Omega_n(1)}$. Since also $\frac{1 - \alpha_F}{4(n-f)} \in n^{-O_n(1)}$, we have 
\begin{equation}
    \label{eq:t-is-o1-explosive}
    T \in O_n(1).
\end{equation}
By \lemmaref{lemma:Q-delocalization}, we have for every $v \in V\setminus F_\star$, 
\[
\frac{(1-\alpha_{F_\star})^{(T+1)/2}}{\sqrt{2(n-f)}}\le \eigen_{1v} \le \frac{\sqrt{2}(1-\alpha_{F_\star})^{-(T+1)/2}}{\sqrt{(n-f)}}.
\]
Then, for every $v \in V\setminus F_\star$, we have
\[
\norm{\vec\eigen_1}_1 = \sum_{v \in V\setminus F_\star} \left|\eigen_{1v}\right| \ge (n-f) \frac{(1-\alpha_{F_\star})^{(T+1)/2}}{\sqrt{2(n-f)}} = \frac{(1-\alpha_{F_\star})^{(T+1)/2}\sqrt{n-f}}{\sqrt{2}}.
\]
This implies that for every $v \in V\setminus F_\star$
\begin{align}
\frac{\eigen_{1v}}{\norm{\vec\eigen_1}_1} &\le \frac{\frac{\sqrt{2}(1-\alpha_{F_\star})^{-(T+1)/2}}{\sqrt{(n-f)}}}{\frac{(1-\alpha_{F_\star})^{(T+1)/2}\sqrt{n-f}}{\sqrt{2}}} = \frac{2(1-\alpha_{F_\star})^{-(T+1)}}{n-f}.\notag
\intertext{By \equationref{eq:t-is-o1-explosive}, we have $T \in O_n(1)$, and also $\alpha_{F_\star} \le \alpha \in 1 - \Omega_n(1)$, therefore, }
\frac{\eigen_{1v}}{\norm{\vec\eigen_1}_1} &\in O_n\left(\frac{1}{n-f}\right).\notag
\intertext{Since $f/n\in 1-\Omega_n(1)$, we have $n-f \in \Theta_n(n)$, hence}
\frac{\eigen_{1v}}{\norm{\vec\eigen_1}_1} &\in O_n\left(\frac{1}{n}\right).
\label{eq:eigen-over-1-norm-explosive}
\end{align}
Since $G$ is connected and $F_\star \not = V$, $N(F_\star)\setminus F_\star$ has at least one node. Let $v_\star \in N(F_\star)\setminus F_\star$ be an arbitrary non-curious node in a neighborhood of $F$. Then, by \lemmaref{lemma:passage-explosive-ub}, there exists $u_\star \in V\setminus F_\star$ such that
\begin{align*}
    \passagedandelion{u_\star}{v_\star} &\le \frac{\eigen_{1v_\star}}{\norm{\vec\eigen_1}_1}\cdot \frac{1}{(f/n + \rho(1 - f/n))(1-\lambda)}.
    \intertext{By \equationref{eq:eigen-over-1-norm-explosive}, we get}
    \passagedandelion{u_\star}{v_\star}&\in O_n\left(\frac{1}{n}\cdot \frac{1}{(f/n + \rho(1 - f/n))(1-\lambda) } \right).
    \intertext{Hence,}
    \passagedandelion{u_\star}{v_\star} &\in O_n\left(\frac{1}{(f + \rho(n - f))(1-\lambda) } \right).
    \intertext{Recall that $\mathcal G$ is near-Ramanujan. Hence,}
    \passagedandelion{u_\star}{v_\star} &\in O_n\left(\frac{1}{f + \rho(n - f)} \right).
\end{align*}
Then, by \lemmaref{lemma:necessary-explosive-passage}, we have
\begin{align*}
\divinfty{\seqadv{v_\star}}{\seqadv{u_\star}} \ge \ln\left(\passage{u_\star}{v_\star}^{-1}\right) \in \ln\left(f + \rho(n - f)\right) + \Omega_n(1),
\end{align*}
which yields
\begin{align}
\label{eq:divinfty-lower-dandelion}
\max_{v,u\in V\setminus F_\star}\divinfty{\seqadv{v}}{\seqadv{u}} \in \ln\left(f + \rho(n - f)\right) + \Omega_n(1).
\end{align}
Note that the above holds for arbitrary set of curious nodes $F_\star$ for which $\alpha_{F_\star} \le \alpha$ holds. By \lemmaref{lemma:adversarial-density}, $\alpha_{F_\star} \le \alpha$ holds true with high probability as $F_\star$ is selected uniformly at random. Hence, (\ref{eq:divinfty-lower-dandelion}) implies that for both worst-case and average-case adversaries we have 
\[
\varepsilon \in \ln\left(f + \rho(n - f)\right) + \Omega_n(1).
\]
\end{proof}

\subsubsection{Dissemination time of Dandelion (proof of \theoremref{thm:dandelion-speed})}

\begin{restatable}{theorem}{DandelionSpeed}
\label{thm:dandelion-speed}
Consider an undirected connected graph $G$ of size $n$ and diameter $D$. Let $\mathcal P$ be a $\rho$-Dandelion protocol with $\rho \le 1$ and $\rho \in \Omega_n(1/n)$. Then, regardless of the source node, the dissemination time of $\mathcal P$ is in order of $\Theta_n\left(1/\rho + D\right)$ rounds in expectation. 
\end{restatable}

\begin{proof}
Let $\tau$ be the number of rounds it takes for the protocol to transition into a spreading phase. Then $\tau \sim \Geom(\rho)$, and hence $\E[\tau] = 1/\rho$.

\paragraph*{(i) Upper bound.} Note that broadcast spreads the gossip to all nodes in at most $D$ rounds. Since the protocol spends $1/\rho$ rounds in expectation in the anonymity phase, the upper bound of $1/\rho + D$ follows.

\paragraph*{(ii) Lower bound.} First, note that for any $v \in V$ there exists $u \in V$ at distance at least $D/2$ away from $v$. Indeed, if every vertex is at distance $< D/2$ away from $v$, then for any two $w_1$, $w_2$ we have $d(w_1, w_2) \le d(w_1, v) + d(v, w_2) < D$ by triangle inequality, i.e., the diameter of the graph is $<D$, which is a contradiction. With this in mind, consider the two following cases.

\paragraph*{(ii.1) Case $1/\rho < D/4$.} Let $s$ be the source, and let $u$ be the node at distance at least $D/2$ away from $s$. Note that, regardless of the phase of the execution, the protocol requires at least $D/2$ rounds to propagate the gossip from $s$ to $u$. 
Then, the expected dissemination time can be lower bounded by 
\[
D/2 \in \Omega_n(D) \subseteq \Omega_n(D + 1/\rho),
\]
where the last transition follows from the fact that $1/\rho < D/4$.

\paragraph*{(ii.2) Case $1/\rho \ge D/4$.} Note that, if have $\rho > 1/2$, the lower bound is trivial, since it amounts to asserting that dissemination time is in $\Omega_n(1)$. In the remaining, assume $\rho \le 1/2$. Recall that also $\rho \in \Omega_n(1/n)$. Define $t = \min\{1/\rho, n/2\}$. Since $\rho \in \Omega_n(1/n)$, we have $t = \Theta_n(1/\rho)$. Since also $\rho \le 1/2$, we have $\Pr[\tau \ge t] = (1-\rho)^{t} \in (1-\rho)^{\Theta_n(1/\rho)} \subseteq \Omega_n(1)$. Finally, note that if $\tau \le t \le n/2$, then after $t$ round of the protocol, there are nodes which have not received the gossip yet. Hence, the expected dissemination time is lower bounded by 
\[
\Pr[\tau \ge t] t \in \Omega_n(t) \subseteq \Omega_n(1/\rho) \subseteq \Omega_n(1/\rho + D),
\]
where the last transition follows from the fact that $1/\rho \ge D/4$. This concludes the proof of the lower bound.

\end{proof}

\section{Extension to muting push}
\label{sec:muting-privacy}

\subsection{Protocol description}

\label{appendix:muting-def}

Muting push is a protocol introduced in~\cite{who_started_this_rumor}, which constitutes a parameterized version of standard push protocol~\cite{pittel1987spreading}. This is a gossip protocol where at every round $t \geq \starttime$, each node active node $u \in \activeset_t$ pushes the gossip $g$ to a neighbour $v \in N(u)$ chosen uniformly at random (i.e., $(u \to v)$ is added to $\commset_t$). Additionally, $u$ samples a token from a Bernoulli distribution with parameter $\rho$. If the token equals one, $u$ stays active for the next round (i.e., $(u \to u)$ is added to $\commset_t$). If the token equals zero, $u$ deactivates (``mutes''). Note that, in our setting, $\rho$ corresponds to the probability of a node \emph{not} muting and staying active for a subsequent round.
Note that, when $\rho = 0$, this protocol degenerates into a classic random walk on the graph. On the other hand, $\rho = 1$ matches the standard push gossip~\cite{pittel1987spreading}.

\subsection{Privacy guarantees}

Now, we extend our positive results in \theoremref{thm:main} and \corollaryref{corollary:privacy-nr} to muting push~\cite{who_started_this_rumor}. We will prove the following results.

\begin{restatable}{theorem}{MainThmMuting}
\label{thm:main-muting}
Consider an undirected connected $(d,\lambda)$-expander graph $G=(V,E)$ of size $n$, let $f$ be the number of curious nodes, and let $\mathcal P$ be a $\rho$-muting push with $\rho < 1$. Set $\alpha = f/d$ (resp. set $\alpha$ as in \lemmaref{lemma:adversarial-density}). If $\lambda < 1 - \alpha$, then $\mathcal P$ satisfies $\varepsilon$-DP against a worst-case adversary (resp. an average-case adversary) with 
\[
\varepsilon = \ln(\rho(n-f) + f) - 2\Tilde{T}\ln(1 - \alpha) - \Tilde{T}\ln(1 - \rho) - \ln(1-\lambda) + \ln(24),
\]
and $\Tilde{T}  = \left\lceil\log_{\frac{\lambda}{1 - \alpha}}\left(\frac{1 - \alpha}{4(n-f)}\right)\right\rceil \left(\log_{\frac{\lambda}{1 - \alpha}}(1 - \alpha) + 2\right) + 2.$
\end{restatable}
\begin{restatable}{corollary}{PrivacyNRMuting}
    \label{cor:muting}
    Let $\mathcal P$ be a $\rho$-muting push and let $\mathcal G$ be a family of $d$-regular near-Ramanujan graphs with $n$ nodes and $d \in n^{\Omega_n(1)}$. Suppose $f/d \in 1 - \Omega_n(1)$ (resp. $f/n \in 1 - \Omega_n(1)$). Then, for any $G \in \mathcal G$ of large enough size $n$ and any $\rho \in 1 - \Omega_n(1)$, $\mathcal P$ satisfies $\varepsilon$-DP against a worst-case adversary (resp. an average-case adversary) for some 
    \[
    \varepsilon \in \ln\left(\rho(n-f) + f\right) + O_n(1).
    \]
\end{restatable}

To establish the results above, it is sufficient to show that a reduction in \lemmaref{lemma:sufficient-varepsilon-0} applies to $\rho$-muting push. The rest of the proof of \theoremref{thm:main-muting} and \corollaryref{cor:muting} will be exactly the same as for \theoremref{thm:main} and \corollaryref{corollary:privacy-nr}. Formally, we have the following.
\begin{restatable}{lemma}{MutingPushReduction}
    Conside $\rho$-muting push on a $d$-regular graph $G = (V,E)$. Let $F \subset V$ be a set of curious nodes such that the subgraph of $G$ induced by $\honest$ is connected. Let $\vec Q = \Hat{\vec A}[\honest]$ and, for $s \in \honest$, let $W^{(s)}$ be the absorbing state of the Markov chain as in \equationref{eq:block-form-absorbing-markov}. Then, for any  $v, u \in \honest$, the following holds true
    \[ \divinfty{\seqadv{v}}{\seqadv{u}} \le \divinfty{W^{(v)}}{W^{(u)}} = \max_{w\in \honest} \ln\frac{(\vec I_{n-f} - (1 - \rho)\vec Q)^{-1}_{vw}}{(\vec I_{n-f} - (1 - \rho)\vec Q)^{-1}_{uw}}.\]
\end{restatable}
\begin{proof}
Let $s\in \honest$ be arbitrary. Note that the right-hand side equality follows from the proof of (\ref{lemma:sufficient-varepsilon-0}). To prove the left-hand side inequality, we will first introduce a notion of ``safe'' rounds, in a similar way as for the proof of \lemmaref{lemma:sufficient-varepsilon-0}. We call a round of the execution safe if the set of active nodes then has size $1$, and the only active node contacts a non-curious node during the current round and then deactivates (``mutes''). We also introduce an indicator variable $\nice_t^{(s)}$ corresponding to a round $t$ being safe. Formally, 
\begin{equation}
\label{eq:nice-muting}
    \nice_t^{(s)} = \ind_{\left\{\exists v,u \in \honest \colon \activeset_t^{(s)} = \{u\}\land \commset_t^{(s)} = \{(u \to v)\}\right\}}.
\end{equation}
Note that since $\commset_t^{(s)}$ does not contain $(u \to u)$ in the event under the indicator in (\ref{eq:nice-muting}), $u$ will deactivate after round $t$ (i.e., $u \not \in \activeset_{t+1}^{(s)}$) according to the formalism we defined in \sectionref{sec:preliminaries}.

\paragraph*{(i) Relating safe rounds to an absorbing Markov chain in (\ref{eq:block-form-absorbing-markov}).}

If $\activeset_t^{(s)} = \{u\}$ for some $u \in \honest$, then $u$ will deactivate after round $t$ probability $1-\rho$. Additionally, since we consider $d$-regular graphs, $u$ will contact a fixed node $v \in N(u)$ with probability $1/d$. Hence, for any $u, v \in \honest$ such that $v \in N(u)$ we have
\begin{align*}
    \Pr\left[\activeset_{t+1}^{(s)} = \{v\} \land \nice_t^{(s)} = 1 \mid \activeset_t^{(s)} = \{u\} \right] &= \frac{1 - \rho}{d} = (1-\rho) \vec Q_{vu}
\end{align*}
Note that for any $t$ and any $u, v \in \honest$ such that $v \notin N(u)$, we have
\begin{align*}
\Pr[\activeset_{t+1}^{(s)} = \{v\} \land \nice_t^{(s)} = 1 \mid \activeset_t^{(s)} = \{u\}]= 0 = (1 - \rho)\vec Q_{vu}.
\end{align*}
Then, for any $v,u \in \honest$, we get
\begin{align}
\label{eq:transient-flag-up-muting}
\Pr[\activeset_{t+1}^{(s)} = \{v\} \land \nice_t^{(s)} = 1 \mid \activeset_t^{(s)} = \{u\}] = (1 - \rho)\vec Q_{vu},
\end{align}
which is equal to the transition probability between two transient states $v$ and $u$ of (\ref{eq:block-form-absorbing-markov}). Also, for any $u$, we have 
\begin{align}
\Pr[\nice_t^{(s)} = 0 \mid \activeset_t^{(s)} = \{u\}] &= 1 - (1-\rho)\frac{\deg_{\honest}(u)}{d} \notag\\
                                               &= \rho + \frac{\deg_F(u)}{d} (1-\rho) \notag\\
                                               &= \rho(\vec I_{n-f})_{uu} + (1 - \rho)\vec R_{uu} \label{eq:absorbing-flag-down-muting},
\end{align}
which is equal to the probability of being absorbed at $\sink(u)$ from state $u$ in (\ref{eq:block-form-absorbing-markov}). Let $\tau^{(s)}$ be the first unsafe round, i.e., 
\begin{equation}
    \label{eq:tau-s-muting}
    \tau^{(s)} = \min\{t \colon \nice_t^{(s)} = 0\}.
\end{equation}
Then, $\tau^{(s)}$ is the first round in which the active node either does not mute (i.e., $(u\to u)\in \commset_t$ for $u \in \activeset_t^{(s)}$) or contacts a curious node. Then, from (\ref{eq:transient-flag-up-muting}) and (\ref{eq:absorbing-flag-down-muting}), at time $\tau^{(s)}$, we have $\activeset_{\tau^{(s)}}^{(s)} = \{W^{(s)}\}$, where $W^{(s)}$ is an absorbing state of chain defined in (\ref{eq:block-form-absorbing-markov}).

\paragraph*{(ii) Applying the Data Processing inequality.}
Let $\round_t^{(s)} = (\activeset_t^{(s)}, \commset_t^{(s)})$ for every $t$, i.e., $\round_t^{(s)}$ describes a round $t$ of the execution. Consider a sequence $\{\round_t^{(s)}\}_{t \ge \starttime}$. Note that $\{\round_t^{(s)}\}_{t \ge \starttime}$ is \emph{Markovian} by definition of muting push in \appendixref{appendix:muting-def}.
% since in muting push round $t$ of the execution only depends on which nodes are active in the beginning of the round, and does not depend on any execution history apart from that.
Then $\tau^{(s)}$ is a \emph{stopping time} for a Markov chain $\{\round_t^{(s)}\}_{t\ge \starttime}$ by definition of $\tau^{(s)}$ in (\ref{eq:tau-s-muting}) (for definition of a stopping time, see Section 6.2 of~\cite{LevinPeresWilmer2006}).

Then, by Strong Markov Property (Proposition A.19 of~\cite{LevinPeresWilmer2006}), the law of $\{\round_t^{(s)}\}_{t \ge \tau^{(s)}}$ only depends on $\round_{\tau^{(s)}}^{(s)}$. Hence, by the Data Processing Inequality (Theorem 14 of~\cite{dataprocessing}), we have
\begin{align}
\label{eq:data-processing-one-muting}
\divinfty{\{\round_t^{(v)}\}_{t \ge \tau^{(v)}}}{\{\round_t^{(u)}\}_{t \ge \tau^{(u)}}} \le \divinfty{\round_{\tau^{(v)}}^{(v)}}{\round_{\tau^{(u)}}^{(u)}}.
\end{align}
By definition of $\tau^{(s)}$ in (\ref{eq:tau-s-muting}), for all $t \le \tau^{(s)}$, we have $|\activeset_t^{(s)}| = 1$. Also, note that the random variable $\commset_{\tau^{(s)}}^{(s)}$ characterizes a round of muting push in which it has active set $\activeset_{\tau^{(s)}}^{(s)}$ and is conditioned on the only active node either not muting for next round (i.e., $(u\to u)\in \commset_t$ for $u \in \activeset_t^{(s)}$) or contacting a curious node. Then, since muting push is Markovian, communications $\commset_{\tau^{(s)}}^{(s)}$ that happen in round $\tau^{(s)}$ only depend on $\activeset_{\tau^{(s)}}^{(s)}$. Hence, by the Data Processing Inequality (Theorem 14 of~\cite{dataprocessing}), we have $\divinfty{\round_{\tau^{(v)}}^{(v)}}{\round_{\tau^{(u)}}^{(u)}} \le \divinfty{\activeset_{\tau^{(v)}}^{(v)}}{\activeset_{\tau^{(u)}}^{(u)}}$. Then
\begin{align}
\label{eq:data-processing-two-muting}
\divinfty{\{\round_t^{(v)}\}_{t \ge \tau^{(v)}}}{\{\round_t^{(u)}\}_{t \ge \tau^{(u)}}} \le \divinfty{\activeset_{\tau^{(v)}}^{(v)}}{\activeset_{\tau^{(u)}}^{(u)}}.
\end{align}
Finally, recall that we showed in (i) that $\activeset_{\tau^{(s)}}^{(s)} = \{W^{(s)}\}$, where $W^{(s)}$ is an absorbing state of the Markov chain (\ref{eq:block-form-absorbing-markov}). Hence,
\begin{align}
\label{eq:data-processing-three-muting}
\divinfty{\{\round_t^{(v)}\}_{t \ge \tau^{(v)}}}{\{\round_t^{(u)}\}_{t \ge \tau^{(u)}}} \le \divinfty{W^{(v)}}{W^{(u)}}. 
\end{align}

    Note also that since no curious node is contacted before time $\tau^{(s)}$, $\seqadv{s}$ can be obtained from $\{\round_t^{(s)}\}_{t \ge \tau^{(s)}} = \{(\commset_{t}^{(s)}, \activeset_{t}^{(s)})\}_{t \ge \tau^{(s)}}$ via a deterministic mapping from definition of $\seqadv{s}$ in \sectionref{sec:source-anonymity-with-DP}. Then, applying the  Data Processing Inequality (Theorem 14 of~\cite{dataprocessing}) again, for any $v,u \in \honest$ we have 
    \begin{equation}
    \label{eq:data-process-seqadv-xt-muting}
     \divinfty{\seqadv{v}}{\seqadv{u}} \le \divinfty{\{\round_t^{(v)}\}_{t \ge \tau^{(v)}}}{\{\round_t^{(u)}\}_{t \ge \tau^{(u)}}}.
    \end{equation}
    By combining the above with (\ref{eq:data-processing-three-muting}), we get
    \begin{equation}
    \label{eq:data-process-seqadv-w-muting}
         \divinfty{\seqadv{v}}{\seqadv{u}} \leq \divinfty{W^{(v)}}{W^{(u)}},
    \end{equation}
    which concludes the proof.
\end{proof}

\section{Semantic for source anonymity: proofs of \sectionref{subsection:DPsemantics}}

\label{appendix:bayesian}
\label{appendix:semantic}

\subsection{Proof of (\ref{eq:map-success})}

\begin{proof}[Proof of (\ref{eq:map-success})]
    First, note that by definition of $\varepsilon$-DP, we have for any $v \in \honest$.
    \begin{equation*}
        \frac{\Pr\left[\seqadv{\smap} \in \sigmaevent\right]}{\Pr\left[\seqadv{v} \in \sigmaevent\right]} \le \exp\left(\divinfty{\seqadv{\smap}}{\seqadv{v}}\right) \le \exp(\varepsilon).
    \end{equation*}
    Then, we have
    \begin{equation}
        \label{eq:smap-ratio-bound}
        \frac{\Pr\left[\seqadv{\smap} \in \sigmaevent\right]}{\min_{v\in\honest}\Pr\left[\seqadv{v} \in \sigmaevent\right]} \le \exp(\varepsilon).
    \end{equation}
    By Bayes law, we get 
    \begin{align*}
    \Pr_{s \sim \prior}\left[\smap = s \mid \seqadv{s} \in \sigmaevent \right] &= \frac{\Pr\left[\seqadv{\smap} \in \sigmaevent\right] p(\smap)}{\sum_{v \in \honest} \Pr\left[\seqadv{v} \in \sigmaevent\right] p(v)}. \\
    &\le \frac{\Pr\left[\seqadv{\smap} \in \sigmaevent\right] p(\smap)}{\min_{v\in \honest} \Pr\left[\seqadv{v} \in \sigmaevent\right] \sum_{v \in \honest} p(v)}. \\
    \intertext{Using (\ref{eq:smap-ratio-bound}), we get}
    \Pr_{s \sim \prior}\left[\smap = s \mid \seqadv{s} \in \sigmaevent \right] &\le \frac{\exp(\varepsilon) p(\smap)}{\sum_{v \in \honest} p(v)} \\
    &\le \exp(\varepsilon) p(\smap),
    \end{align*}
    which concludes the proof.
\end{proof}

\subsection{Proof of (\ref{eq:mle-success})}

\begin{proof}[Proof of (\ref{eq:mle-success})]
    Plugging $p = \uniform{\honest}$ into (\ref{eq:map-success}), we have
    \[
    \Pr_{s \sim \uniform{\honest}} \left[\smle = s \mid \seqadv{s} \in \sigmaevent \right] \le \exp(\varepsilon) p(\smle) = \frac{\exp(\varepsilon)}{n - f},
    \]
    as desired.
\end{proof}

%% file: tikz/Anaconda.tex
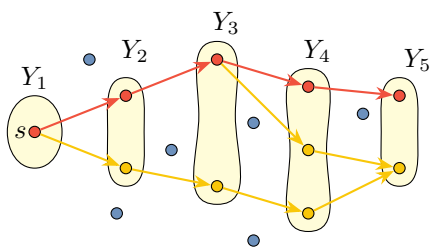
\begin{figure}[!ht]
\centering
    \begin{tikzpicture}
        %\begin{scope}[scale=2.0]
        %    \path[draw,use Hobby shortcut,closed=true, fill=yellow, fill opacity=0.15]
        %     (1.9,-0.1).. (2,-0.5).. (2.1, -0.1) .. (2, 0.3);
        %    \node[circle] (curious) at (2.2, 0.15) {$Y_6$};
        %\end{scope}
\begin{scope}[scale=1.2]
    
        \begin{scope}[scale=2.0]
            \path[draw,use Hobby shortcut,closed=true, fill=yellow, fill opacity=0.15]
             (1.4, 0).. (1.5,-0.3).. (1.6, 0) .. (1.5, 0.3);
            \node[circle] (curious) at (1.6, 0.4) {$Y_5$};
        \end{scope}

        \begin{scope}[scale=2.0]
            \path[draw,use Hobby shortcut,closed=true, fill=yellow, fill opacity=0.15]
             (0.9, -0.1).. (1,0.35).. (1.1, -0.1) .. (1, -0.55);
            \node[circle] (curious) at (1.05, 0.45) {$Y_4$};
        \end{scope}

        \begin{scope}[scale=2.0]
            \path[draw,use Hobby shortcut,closed=true, fill=yellow, fill opacity=0.15]
             (0.4, 0.1).. (0.5,0.5).. (0.6, 0.1) .. (0.5, -0.4);
            \node[circle] (curious) at (0.55, 0.6) {$Y_3$};
        \end{scope}

        \begin{scope}[scale=2.0]
            \path[draw,use Hobby shortcut,closed=true, fill=yellow, fill opacity=0.15]
             (-0.1, 0).. (0,0.3).. (0.1, 0) .. (0, -0.3);
            \node[circle] (curious) at (0.05, 0.45) {$Y_2$};
        \end{scope}

         \begin{scope}[scale=2.0]
            \path[draw, use Hobby shortcut, closed=true, fill=yellow, fill opacity=0.15]
             (-0.65, 0).. (-0.5,0.2).. (-0.35, 0) .. (-0.5, -0.2);
            \node[circle] (curious) at (-0.5, 0.3) {$Y_1$};
        \end{scope}
    
        \begin{scope}[every node/.style={circle, draw, inner sep=0pt,text width=1.5mm,fill=curious}, scale=2.0]
            %%% cobra nodes
            %Main branch
            \node[label=left:{$s$}] (B) at (-0.5,0) {}; 
            \node (C) at (0,0.2) {};
            \node (D) at (0.5,0.4) {};
            \node (E) at (1,0.25) {};
            \node (F) at (1.5,0.2) {};
            %\node (G) at (2,0.15) {};
        \end{scope}

        \begin{scope}[every node/.style={circle, draw, inner sep=0pt,text width=1.5mm,fill=cobra}, scale=2.0]
            \node (H) at (0,-0.2) {};
            \node (I) at (0.5,-0.3) {};
            \node (J) at (1,-0.1) {};
            \node (K) at (1,-0.45) {};
            \node (L) at (1.5,-0.2) {};
            %\node (M) at (2.0,-0.35) {};
        \end{scope}
        \begin{scope}[every node/.style={circle, draw, inner sep=0pt,text width=1.5mm,fill=node}, scale=2.0]
            %%% regular nodes
            
            \node (1) at (0.7,0.05) {};
            \node (2) at (-0.05,-0.45) {}; 
            % \node (3) at (0.8,-0.1) {};
            \node (4) at (1.3, 0.1) {};
            \node (10) at (0.25,-0.1) {};
            \node (6) at (0.7,-0.6) {};
            % \node (7) at (1.15,-0.05) {};
            \node (9) at (-0.2,0.4) {};

            % \node (8) at (1.6,-0.6) {};
            % \node (5) at (0.3,-0.55) {};

        \end{scope}

        \begin{scope}[>={Stealth[curious]},
                      every edge/.style={draw=curious, thick}]
            %%% the main branch edges
            %\path[->] (A) edge (B);
            \path[->] (B) edge (C);
            \path[->] (C) edge (D);
            \path[->] (D) edge (E);
            \path[->] (E) edge (F);
            %\path[->] (F) edge (G);
        \end{scope}
        \begin{scope}[>={Stealth[cobra]},
                      every edge/.style={draw=cobra, thick}]

         %%% other branches
            \path[->] (B) edge (H);
            \path[->] (H) edge (I);
            \path[->] (D) edge (J);
            \path[->] (I) edge (K); 
            \path[->] (J) edge (L);
            \path[->] (K) edge (L);
            %\path[->] (L) edge (M);

        \end{scope}
       \end{scope}

    \end{tikzpicture}

    \caption{Illustration of an anaconda walk. $s$ is the source of dissemination, $Y_i$ corresponds to nodes active at round $i$. Red nodes correspond to the main branch: each red node in $Y_i$ is a head $h_i$ respectively. As the head branches, the branching counter is updated at each step: $c_1=0; c_2 = c_3= 1; c_4=c_5=2$. Blue nodes are those that did not receive the gossip yet.}
    \label{fig:anaconda}
\end{figure}

%% file: Bibliography.bib
@book{LevinPeresWilmer2006,
  title={Markov chains and mixing times},
  author={Levin, David A and Peres, Yuval},
  volume={107},
  year={2017},
  publisher={American Mathematical Soc.}
}

@article{Yu2014AUV,
  title={A useful variant of the {D}avis--{K}ahan theorem for statisticians},
  author={Yi Yu and Tengyao Wang and Richard J. Samworth},
  journal={Biometrika},
  year={2014},
  volume={102},
  pages={315-323}
}

@article{Tikhomirov2019TheSG,
  title={The spectral gap of dense random regular graphs},
  author={Konstantin E. Tikhomirov and Pierre Youssef},
  journal={The Annals of Probability},
  year={2019}
}

@article{hethcote2000mathematics,
  title={The mathematics of infectious diseases},
  author={Hethcote, Herbert W},
  journal={SIAM review},
  volume={42},
  number={4},
  pages={599--653},
  year={2000},
  publisher={SIAM}
}

@inproceedings{doerr2011social,
  title={Social networks spread rumors in sublogarithmic time},
  author={Doerr, Benjamin and Fouz, Mahmoud and Friedrich, Tobias},
  booktitle={Proceedings of the forty-third annual ACM symposium on Theory of computing (STOC 2011)},
  year = {2011}
}

@inproceedings{giakkoupis2015privacy,
  title={Privacy-conscious information diffusion in social networks},
  author={Giakkoupis, George and Guerraoui, Rachid and J{\'e}gou, Arnaud and Kermarrec, Anne-Marie and Mittal, Nupur},
  booktitle={International Symposium on Distributed Computing (DISC 2015)},
  year = {2015}
  }

@article{pittel1987spreading,
  title={On spreading a rumor},
  author={Pittel, Boris},
  journal={SIAM Journal on Applied Mathematics},
  volume={47},
  number={1},
  pages={213--223},
  year={1987},
  publisher={SIAM}
}

@article{acan2017push,
  title={On the push\&pull protocol for rumor spreading},
  author={Acan, Huseyin and Collevecchio, Andrea and Mehrabian, Abbas and Wormald, Nick},
  journal={SIAM Journal on Discrete Mathematics},
  volume={31},
  number={2},
  pages={647--668},
  year={2017},
  publisher={SIAM}
}

@inproceedings{karp2000randomized,
  title={Randomized rumor spreading},
  author={Karp, Richard and Schindelhauer, Christian and Shenker, Scott and Vocking, Berthold},
  booktitle={41st Annual Symposium on Foundations of Computer Science (FOCS 2000)},
  year = 2000,
}

@article{dwork1988consensus,
  title={Consensus in the presence of partial synchrony},
  author={Dwork, Cynthia and Lynch, Nancy and Stockmeyer, Larry},
  journal={Journal of the ACM },
  volume={35},
  number={2},
  pages={288--323},
  year={1988},
  publisher={ACM New York, NY, USA}
}

@book{meyer2000matrix,
  title={Matrix analysis and applied linear algebra},
  author={Meyer, Carl D and Stewart, Ian},
  year={2023},
  publisher={SIAM}
}

@article{10.1145/2450142.2450147,
  title={Asynchronous gossip},
  author={Georgiou, Chryssis and Gilbert, Seth and Guerraoui, Rachid and Kowalski, Dariusz R},
  journal={Journal of the ACM},
  volume={60},
  number={2},
  pages={1--42},
  year={2013},
  publisher={ACM New York, NY, USA}
}

@inproceedings{kowalski2013estimating,
  title={Estimating time complexity of rumor spreading in ad-hoc networks},
  author={Kowalski, Dariusz R and Caro, Christopher Thraves},
  booktitle={International Conference on Ad-Hoc Networks and Wireless (ADHOC-NOW 2013)},
  year = {2013}
 }

@inproceedings{Dwork_2006,
  title={Calibrating noise to sensitivity in private data analysis},
  author={Dwork, Cynthia and McSherry, Frank and Nissim, Kobbi and Smith, Adam},
  booktitle={Theory of Cryptography: Third Theory of Cryptography Conference, TCC 2006, New York, NY, USA, March 4-7, 2006. Proceedings 3},
  pages={265--284},
  year={2006}
}

@Article{Dwork_2013,
  author    = {Cynthia Dwork and Aaron Roth},
  title     = {The Algorithmic Foundations of Differential Privacy},
  journal   = {Foundations and Trends in Theoretical Computer Science},
  year      = {2013},
  volume    = {9},
  number    = {3-4},
  pages     = {211--407},
}

@inproceedings{Melancon2006HowDenseAreGraphs,
author = {Melancon, Guy},
title = {Just How Dense Are Dense Graphs in the Real World? {A} Methodological Note},
booktitle = {Proceedings of the AVI Workshop on BEyond Time and Errors: Novel Evaluation Methods for Information Visualization (BELIV 2006)},
year = {2006}
}

@article{Marshall1980InequalitiesTO,
  title={Inequalities: theory of majorization and its applications},
  author={Marshall, Albert W and Olkin, Ingram and Arnold, Barry C},
  year={1979}
}

@article{Cook2018SizeBC,
  title={Size biased couplings and the spectral gap for random regular graphs},
  author={Nicholas A. Cook and Larry Goldstein and Tobias Johnson},
  journal={Annals of Probability},
  year={2018},
  volume={46},
  pages={72-125}
}

@article{Kemeny1960FiniteMC,
  title={Finite markov chains},
  author={Kemeny, John G and Snell, J Laurie},
  journal={(No Title)},
  year={1960}
}

@INPROCEEDINGS{DP_gossip_in_general_networks,
  author={Huang, Yufan and Jin, Richeng and Dai, Huaiyu},
  booktitle={IEEE Global Communications Conference (GLOBECOM 2020)}, 
  title={Differential Privacy and Prediction Uncertainty of Gossip Protocols in General Networks}, 
  year = {2020}
 }

@INPROCEEDINGS{gossip_that_preserves_privacy_for_distr_computing,
  author={Liu, Yang and Wu, Junfeng and Manchester, Ian R. and Shi, Guodong},
  booktitle={IEEE Conference on Decision and Control (CDC 2018)}, 
  title={Gossip Algorithms that Preserve Privacy for Distributed Computation Part {I}: The Algorithms and Convergence Conditions}, 
  year = {2018}
}

@ARTICLE{hiding_the_source_fanti,
  author={Fanti, Giulia and Kairouz, Peter and Oh, Sewoong and Ramchandran, Kannan and Viswanath, Pramod},
  journal={IEEE Transactions on Information Theory}, 
  title={Hiding the Rumor Source}, 
  year={2017},
  volume={63},
  number={10},
  pages={6679-6713},
}

@article{Anonymous_communication_Crowds,
  title={Anonymous web transactions with crowds},
  author={Reiter, Michael K and Rubin, Aviel D},
  journal={Communications of the ACM},
  volume={42},
  number={2},
  pages={32--48},
  year={1999}
}

@inproceedings{degradation_anonymous_protocols,
  title={An Analysis of the Degradation of Anonymous Protocols},
  author={Matthew K. Wright and Micah Adler and Brian Neil Levine and Clay Shields},
  booktitle={Network and Distributed System Security Symposium (NDSS 2002)},
  year = {2002}
}

@inproceedings{towards_measuring_anonymity,
author={D{\'i}az, Claudia and Seys, Stefaan and Claessens, Joris and Preneel, Bart},
title={Towards Measuring Anonymity},
booktitle={Privacy Enhancing Technologies},
year={2003},
pages={54--68},
}

@ARTICLE{review_on_identifying_sources,
  author={Jiang, Jiaojiao and Wen, Sheng and Yu, Shui and Xiang, Yang and Zhou, Wanlei},
  journal={IEEE Communications Surveys \& Tutorials}, 
  title={Identifying Propagation Sources in Networks: State-of-the-Art and Comparative Studies}, 
  year={2017},
  volume={19},
  number={1},
  pages={465-481},
 }

@article{Pinto_2012,
	year = 2012,
	volume = {109},
	number = {6},
	author = {Pedro C. Pinto and Patrick Thiran and Martin Vetterli},
	title = {Locating the Source of Diffusion in Large-Scale Networks},
	journal = {Physical Review Letters}
}

@inproceedings{10.1145/1811099.181106,
author = {Shah, Devavrat and Zaman, Tauhid},
title = {Detecting Sources of Computer Viruses in Networks: Theory and Experiment},
booktitle = {Proceedings of ACM  International Conference on Measurement and Modeling of Computer Systems (SIGMETRICS 2010)},
year = {2010}
}

@inproceedings{who_started_this_rumor,
  title={Who started this rumor? {Q}uantifying the natural differential privacy of gossip protocols},
  author={Bellet, Aur{\'e}lien and Guerraoui, Rachid and Hendrikx, Hadrien},
  booktitle={International Symposium on Distributed Computing (DISC 2020)},
  year = {2020}
}

@inproceedings{5961737,
  author={Georgiou, Chryssis and Gilbert, Seth and Kowalski, Dariusz R.},
  booktitle={International Conference on Distributed Computing Systems (DISC 2011)}, 
  title={Confidential Gossip}, 
  year = {2011}
}

@article{who_is_the_culprit,
author = {Shah, D. and Zaman, T.},
title = {Rumors in a Network: Who's the Culprit?},
year = {2011},
volume = {57},
number = {8},
journal = {IEEE Transactions on Information Theory},
pages = {5163-5181}
}

@book{tao2012topics,
  title={Topics in random matrix theory},
  author={Tao, Terence},
  volume={132},
  year={2012},
  publisher={American Mathematical Soc.}
}

@inproceedings{DesfontainesPejo2020Survey,
author = {Damien Desfontaines and Balazs Pejo},
title = {SoK: Differential privacies},
booktitle = {Proceedings on Privacy Enhancing Technologies Symposium (PETS 2020)},
year = {2020}
}

@InProceedings{Karol2017LocationHiding,
author="Gotfryd, Karol
and Klonowski, Marek
and Paj{\k{a}}k, Dominik",
title="On Location Hiding in Distributed Systems",
booktitle="Structural Information and Communication Complexity",
year="2017",
pages="174--192",
}

@article{math_of_epidemics,
  title={Mathematics of epidemics on networks},
  author={Kiss, Istv{\'a}n Z and Miller, Joel C and Simon, P{\'e}ter L and others},
  journal={Cham: Springer},
  volume={598},
  pages={31},
  year={2017}
}

@inproceedings{using_features_predict_epidemics,
author = {Yeganeh Alimohammadi and Christian Borgs and Amin Saberi},
title = {Algorithms Using Local Graph Features to Predict Epidemics},
booktitle = {Proceedings of the Annual ACM-SIAM Symposium on Discrete Algorithms (SODA 2022)},
year = 2022
}

@article{boyd2006randomized,
  title={Randomized gossip algorithms},
  author={Boyd, Stephen and Ghosh, Arpita and Prabhakar, Balaji and Shah, Devavrat},
  journal={IEEE Transactions on Information Theory},
  volume={52},
  number={6},
  pages={2508--2530},
  year={2006},
  publisher={IEEE}
}

@inproceedings{guo2014gossip,
  title={Gossip vs. markov chains, and randomness-efficient rumor spreading},
  author={Guo, Zeyu and Sun, He},
  booktitle={Proceedings of the Twenty-Sixth Annual ACM-SIAM Symposium on Discrete Algorithms (SODA 2014)},
  year={2014},
}

@inproceedings{CobraExpanders2016,
author = {Cooper, Colin and Radzik, Tomasz and Rivera, Nicolas},
title = {The Coalescing-Branching Random Walk on Expanders and the Dual Epidemic Process},
booktitle = {Proceedings of the 2016 ACM Symposium on Principles of Distributed Computing (PODC 2016)},
year = {2016}
}

@article{Mitzenmacher2018BetterBF,
  title={Better bounds for coalescing-branching random walks},
  author={Mitzenmacher, Michael and Rajaraman, Rajmohan and Roche, Scott},
  journal={ACM Transactions on Parallel Computing},
  volume={5},
  number={1},
  pages={1--23},
  year={2018},
 }

@inproceedings{Dutta2013Cobra,
author = {Dutta, Chinmoy and Pandurangan, Gopal and Rajaraman, Rajmohan and Roche, Scott},
title = {Coalescing-Branching Random Walks on Graphs},
booktitle = {Proceedings of the Twenty-Fifth Annual ACM Symposium on Parallelism in Algorithms and Architectures (SPAA 2013)},
year = {2013}
}

@article{CHVATAL1979285,
title = {The tail of the hypergeometric distribution},
journal = {Discrete Mathematics},
volume = {25},
number = {3},
pages = {285-287},
year = {1979},
author = {V. Chvátal}
}

@article{hoeffding1994probability,
  title={Probability inequalities for sums of bounded random variables},
  author={Hoeffding, Wassily},
  journal={The collected works of Wassily Hoeffding},
  pages={409--426},
  year={1994}
}

@inproceedings{cooper_improved_2017,
	title = {Improved Cover Time Bounds for the Coalescing-Branching Random Walk on Graphs},
	booktitle = {Proceedings of the 29th {ACM} Symposium on Parallelism in Algorithms and Architectures (SPAA 2017)},
	author = {Cooper, Colin and Radzik, Tomasz and Rivera, Nicolás},
        year = 2017,
}

@inproceedings{berenbrin_tight_2018,
	title = {Tight Bounds for Coalescing-Branching Random Walks on Regular Graphs},
	booktitle = {Proceedings of the Annual {ACM}-{SIAM} Symposium on Discrete Algorithms (SODA 2018)},
	author = {Berenbrin, Petra and Giakkoupis, George and Kling, Peter},
	urldate = {2022-11-28},
	year = {2018},
}

@inproceedings{bojja_venkatakrishnan_dandelion_2017,
  title={Dandelion: {R}edesigning the bitcoin network for anonymity},
  author={Bojja Venkatakrishnan, Shaileshh and Fanti, Giulia and Viswanath, Pramod},
  booktitle ={Proceedings of the ACM on Measurement and Analysis of Computing Systems ({SIGMETRICS} 2017)},
  year = 2017
}

@article{Rudelson2015Delocalization,
author = {Mark Rudelson and Roman Vershynin},
title = {{Delocalization of eigenvectors of random matrices with independent entries}},
volume = {164},
journal = {Duke Mathematical Journal},
number = {13},
publisher = {Duke University Press},
pages = {2507 -- 2538},
year = {2015},
}

@article{Pseudorandomness,
year = {2012},
journal = {Foundations and Trends in Theoretical Computer Science},
title = {Pseudorandomness},
author = {Salil P. Vadhan}
}

@article{Aldous1989LowerBF,
  title={Lower bounds for covering times for reversible Markov chains and random walks on graphs},
  author={David J. Aldous},
  journal={Journal of Theoretical Probability},
  year={1989},
  volume={2},
  pages={91-100}
}

@inproceedings{Cyffers2020PrivacyAB,
  title={Privacy Amplification by Decentralization},
  author={Edwige Cyffers and Aur{\'e}lien Bellet},
  booktitle={International Conference on Artificial Intelligence and Statistics (AIStat 2020)},
  year={2020}
}

@inproceedings{venkitasubramaniam_anonymity_2008,
	title = {Anonymity under light traffic conditions using a network of mixes},
	booktitle = { 46th Annual Allerton Conference on Communication, Control, and Computing (ALLERTON 2008)},
	author = {Venkitasubramaniam, Parv and Anantharam, Venkat},
     year = 2008
}

@article{beimel2003buses,
  title={Buses for anonymous message delivery.},
  author={Beimel, Amos and Dolev, Shlomi},
  journal={Journal of Cryptology},
  volume={16},
  number={1},
  year={2003}
}

@inproceedings{das2018anonymity,
  title={Anonymity trilemma: Strong anonymity, low bandwidth overhead, low latency-choose two},
  author={Das, Debajyoti and Meiser, Sebastian and Mohammadi, Esfandiar and Kate, Aniket},
  booktitle={IEEE Symposium on Security and Privacy (SP)},
  pages={108--126},
  year={2018},
}

@inproceedings{snader2008tune,
  title={A Tune-up for Tor: Improving Security and Performance in the Tor Network.},
  author={Snader, Robin and Borisov, Nikita},
  booktitle={Network and Distributed System Security Symposium (NDSS 2008) },
  year={2008}
}

@inproceedings{lee2011much,
  title={How much is enough? {C}hoosing $\varepsilon$ for differential privacy},
  author={Lee, Jaewoo and Clifton, Chris},
  booktitle={Information Security: 14th International Conference},
  pages={325--340},
  year={2011}
}

@inproceedings{6957125,
  title={Differential privacy: An economic method for choosing epsilon},
  author={Hsu, Justin and Gaboardi, Marco and Haeberlen, Andreas and Khanna, Sanjeev and Narayan, Arjun and Pierce, Benjamin C and Roth, Aaron},
  booktitle={ IEEE 27th Computer Security Foundations Symposium},
  pages={398--410},
  year={2014}
}

@article{9318999,
  title={Information sources estimation in time-varying networks},
  author={Chai, Yun and Wang, Youguo and Zhu, Liang},
  journal={IEEE Transactions on Information Forensics and Security},
  volume={16},
  pages={2621--2636},
  year={2021}
}

@inproceedings{liu2019information,
  title={Information Source Detection with Limited Time Knowledge},
  author={Liu, Xuecheng and Fu, Luoyi and Jiang, Bo and Lin, Xiaojun and Wang, Xinbing},
  booktitle={Proceedings of the Twentieth ACM International Symposium on Mobile Ad Hoc Networking and Computing (MobiHoc 2019)},
  pages={389--390},
  year={2019}
}

@inproceedings{irreg_source_obfuscation,
author = {Fanti, Giulia and Kairouz, Peter and Oh, Sewoong and Ramchandran, Kannan and Viswanath, Pramod},
title = {Rumor Source Obfuscation on Irregular Trees},
year = {2016},
booktitle = {International Conference on Measurement and Modeling of Computer Systems, (SIGMETRICS 2016)}
}

@inproceedings{mironov2017renyi,
  title={R{\'e}nyi differential privacy},
  author={Mironov, Ilya},
  booktitle={2017 IEEE 30th computer security foundations symposium (CSF)},
  pages={263--275},
  year={2017},
  organization={IEEE}
}

@article{dataprocessing,
  title={On divergences and informations in statistics and information theory},
  author={Liese, Friedrich and Vajda, Igor},
  journal={IEEE Transactions on Information Theory},
  volume={52},
  number={10},
  pages={4394--4412},
  year={2006},
  publisher={IEEE}
}

@article{kifer2014pufferfish,
  title={Pufferfish: A framework for mathematical privacy definitions},
  author={Kifer, Daniel and Machanavajjhala, Ashwin},
  journal={ACM Transactions on Database Systems (TODS)},
  volume={39},
  number={1},
  pages={1--36},
  year={2014},
  publisher={ACM New York, NY, USA}
}
